\documentclass[11pt,letterpaper,english,notitlepage]{article}
\usepackage[latin9]{inputenc}
\usepackage{xcolor}
\usepackage[toc,page]{appendix}
\usepackage{pdfcolmk}
\usepackage{babel}
\usepackage{apalike}
\usepackage{amsmath}
\usepackage{array}
\usepackage{amsthm}
\usepackage{amssymb}
\usepackage{comment}
\usepackage{setspace}
\newtheorem{assumption}{Assumption}
\usepackage{mathrsfs,pgf,pgfplots,tabularx}
\usepackage{longtable,booktabs}
\usepackage{amsthm}
\usepackage{algorithm}
\usepackage{algpseudocode}

\usepackage[authoryear]{natbib}
\PassOptionsToPackage{normalem}{ulem}
\usepackage{ulem}
\DeclareMathOperator*{\argmin}{\arg\!\min}
\usepackage[unicode=true,pdfusetitle,
 bookmarks=true,bookmarksnumbered=false,bookmarksopen=false,
 breaklinks=true,pdfborder={0 0 0},pdfborderstyle={},backref=false,colorlinks=true]
 {hyperref}
\makeatletter

\pdfpageheight\paperheight
\pdfpagewidth\paperwidth

\providecolor{lyxadded}{rgb}{0,0,1}
\providecolor{lyxdeleted}{rgb}{1,0,0}

\DeclareRobustCommand{\lyxsout}[1]{\ifx\\#1\else\sout{#1}\fi}

\theoremstyle{plain}

\theoremstyle{plain}

\theoremstyle{plain}
\newtheorem*{assumption*}{\protect\assumptionname}
\theoremstyle{plain}
\newtheorem{theorem}{Theorem}[section]
\newtheorem{lemma}{Lemma}[section]
\newtheorem{corollary}{Corollary}
\@ifundefined{date}{}{\date{}}
\usepackage[margin=1in]{geometry}
\usepackage[multiple]{footmisc}
\usepackage{mathpazo}
\usepackage[flushleft]{threeparttable}
\usepackage{graphicx}
\usepackage{float}

\definecolor{webgreen}{rgb}{0,.5,0}
\definecolor{webbrown}{rgb}{.6,0,0}
\definecolor{webpurple}{rgb}{0.7,0,0.7}

\hypersetup{breaklinks = true,  letterpaper=false,
colorlinks=true, anchorcolor= webbrown, citecolor= blue,
filecolor= webbrown, linkcolor= webbrown, menucolor= webbrown,
urlcolor= webbrown, citebordercolor= 1 0 0, menubordercolor=1 0
0, urlbordercolor=1 0 0, runbordercolor=1 0 0 }

\makeatother

\providecommand{\assumptionname}{Assumption}
\providecommand{\corollaryname}{Corollary}

\providecommand{\theoremname}{Theorem}
\newenvironment{myproof}[1][\proofname]{\noindent\textbf{#1.}\ }{\qed\par}

\begin{document}
\title{
Bootstrap Inference in Nonlinear Panel Data Models with Interactive Fixed Effects\thanks{
We are grateful to participants at the KU Leuven Applied Micro PhD Workshop (2024), the Netherlands Econometric Study Group (2025), the International Association for Applied Econometrics Conference (2025), the International Panel Data Conference (2025), and the Asian Summer School in Econometrics and Statistics (2025).
We thank Andrei Zeleneev and Weisheng Zhang for sharing the codes from \citet{zeleneev2026tractable}. Haoyuan Xu acknowledges funding from the Research Foundation--Flanders (FWO) fellowship no.~1179925N. Wei Miao acknowledges financial support from the China Scholarship Council (grant no.~202306130036). Jad Beyhum acknowledges financial support from the Research Fund KU Leuven (grant STG/23/014).
}
}
\author{
Haoyuan Xu$^\dag$,
Wei Miao$^\ddag$,
Geert Dhaene$^\dag$,
Jad Beyhum$^\dag$
\\[6pt]
$^\dag$Department of Economics, KU Leuven\\
$^\ddag$ORSTAT, KU Leuven
}

\date{\today}

\maketitle
\vspace{-2em}
\begin{abstract}
\noindent
The maximum likelihood estimator in nonlinear panel data models with interactive fixed effects is biased. Several bias correction methods, such as analytical and jackknife approaches, have been proposed to enable valid inference. This paper shows that the parametric bootstrap also enables valid inference in such models. In particular, we show that the parametric bootstrap replicates the asymptotic distribution of the maximum likelihood estimator. Therefore, it yields asymptotically unbiased estimates and confidence sets with asymptotically correct coverage. We also propose a transformation-based bootstrap confidence interval that delivers improved finite-sample performance. Simulation results support the theoretical findings. Finally, we apply the proposed method to examine technological and product market spillover effects on firms' innovation behavior.
\end{abstract}

\clearpage

\section{Introduction}
The importance of unobserved heterogeneity for economic analysis and the design of effective public policies has long been recognized by economists and policymakers. Interactive fixed effects models provide a flexible and powerful framework for capturing such heterogeneity. Unlike classical one-way fixed-effect models \citep{chamberlain1984panel}, which rely on agent-specific scalar parameters to absorb unobserved heterogeneity, interactive fixed-effect models allow for multidimensional latent components, thereby accommodating substantially richer forms of unobserved heterogeneity. Moreover, these models naturally incorporate aggregate time-varying shocks through a factor structure, a feature that is particularly important in macroeconomic and financial applications \citep{stock2002forecasting}.

Since the seminal contributions of \citet{pesaran2006estimation} and \citet{Bai2009}, a large body of research has developed the theoretical foundations of linear panel data models with interactive fixed effects. These advances have, in turn, facilitated the widespread empirical adoption of such models. Nevertheless, linear specifications are often inadequate in applied work. In many empirical settings, the outcome variable is discrete, rendering linear models potentially misspecified. Recent work, such as \citet{mao2025adaptive}, demonstrates that nonlinear factor models can provide a more appropriate framework in a wide range of empirical applications.

In contrast to the extensive literature on linear factor models \citep{Bai2009,MoonWeidner2015,MoonWeidner2017}, the theoretical foundations of nonlinear panel data models with interactive fixed effects are less developed. A notable recent contribution is \citet{ChenFVWeidner2021}, who study the maximum likelihood estimator (MLE) in nonlinear panel models with interactive fixed effects. Their approach treats the individual-specific and time-specific effects as fixed parameters and jointly estimates these nuisance parameters together with the common parameters of interest. This fixed-effect framework is attractive because it imposes no restrictions on the joint distribution of the unobserved effects and the covariates, thereby reducing the risk of model misspecification.

\citet{ChenFVWeidner2021} derive the asymptotic distribution of the MLE of the common parameters and show that it is asymptotically biased due to the incidental parameter problem. Importantly, this source of bias differs fundamentally from that encountered in linear factor models estimated by least squares, where bias in the slope parameters arises from cross-sectional and time-series dependence in the error terms \citep{Bai2009}. To address the incidental parameter bias, they propose analytical and jackknife corrections that eliminate the leading bias term and enable asymptotically valid inference. Relatedly, \citet{GaoLiuPengYan2023} study the MLE in binary choice models with interactive fixed effects and heterogeneous slope parameters.

Although analytical and jackknife bias corrections are asymptotically valid, their finite-sample performance may be limited in empirically relevant settings with moderate sample sizes and short time dimensions, as commonly encountered in microeconomic panel datasets. Analytical bias corrections rely on estimating the asymptotic bias and subtracting the estimated bias from the original estimator. However, as shown in \citet{ChenFVWeidner2021}, the asymptotic bias is a highly intricate expression involving up to third-order derivatives of the likelihood function, which may be numerically unstable. The split-panel jackknife method is comparatively straightforward to implement but can suffer from efficiency losses in finite samples \citep{hahn2024efficient}. Moreover, analytical and jackknife corrections rely on normal approximations to construct confidence intervals, which may be inaccurate.

Recently, \citet{HigginsJochmans2024} show that the parametric bootstrap delivers asymptotically valid confidence intervals without explicit bias correction in one-way fixed-effect nonlinear panel data models, and that bootstrap-based inference can exhibit superior finite-sample performance when the time dimension is short. In this paper, we extend the parametric bootstrap approach to nonlinear panel data models with interactive fixed effects. Unlike analytical bias-correction methods, the parametric bootstrap avoids the explicit calculation of the asymptotic bias by approximating it through parametric bootstrap simulations. An important feature of the parametric bootstrap, emphasized by \citet{HigginsJochmans2024}, is that it yields asymptotically valid inference even without preliminary bias correction. Confidence intervals are constructed directly from the quantiles of the bootstrap distribution rather than relying on normal approximations, which can improve accuracy in finite samples. We show that these properties of the parametric bootstrap also hold in nonlinear panel data models with interactive fixed effects.

Simulations corroborate our theoretical results. Furthermore, simulations indicate that bootstrap-based confidence intervals tend to be conservative in finite samples, with coverage probabilities often exceeding the nominal level and thus reflecting a systematic overestimation of uncertainty. This finding is in line with \citet{HigginsJochmans2024}, who document similar behavior in one-way fixed-effect models. They also note that iterating the bootstrap procedure can improve coverage accuracy, but at a substantial computational cost.

To mitigate this issue, we follow the bootstrap literature on skewness correction \citep{hall_removal_1992} and apply a monotone transformation to the estimator prior to constructing confidence intervals. The transformation reduces skewness in the bootstrap distribution, and intervals are then obtained from the transformed scale and mapped back. By the bootstrap delta method \citep{van1998asymptotic}, this procedure preserves asymptotic validity while improving finite-sample performance. Simulations show that the resulting intervals are substantially shorter and achieve more accurate coverage.

All bias correction methods considered in this literature rely on the MLE, i.e., the global maximizer of the loglikelihood function, being obtainable. However, in practice, the loglikelihood function in nonlinear factor models is generally non-concave, implying that standard optimization algorithms may converge to a local instead of the global maximizer. \citet{ChenFVWeidner2021} propose an EM-type algorithm, which guarantees convergence only to a local maximizer. They suggest using multiple starting values to increase the probability of reaching the global maximizer; nevertheless, this strategy does not ensure global maximization and can substantially increase the computational burden. More recently, \citet{zeleneev2026tractable} propose a computationally efficient two-step estimator based on nuclear-norm penalization, which is shown to attain the global maximizer with probability one and is asymptotically equivalent to the MLE. In this paper, we adopt their two-step estimator in conjunction with the bootstrap procedure, thereby ensuring global maximization at low computational cost and preserving asymptotically valid inference.
\\[0.5em]
\noindent \textbf{Related literature}

\noindent
This paper contributes to the literature on large-$T$ bias correction in nonlinear panel data models with fixed effects. Since the seminal work of \citet{HahnNewey2004}, many studies have aimed to correct the first-order incidental parameter bias in a large-$T$ framework. This includes research on one-way fixed effects panel models \citep{fernandez2009fixed,HahnKuersteiner2011,DhaeneJochmans2015,arellano2016likelihood,HigginsJochmans2024}, additive two-way fixed effects panel models \citep{FVWeidner2016}, and network models \citep{Graham2017,dzemski2019empirical,yan2019statistical,HUGHES2026106130}; see \citet{FVWeidner2018} for a review. More recently, some studies have begun to address second-order bias \citep{DHAENE2021227, schumann2023second}
as well as higher-order bias correction \citep{dhaene2023approximate,bonhomme2024neyman}.

We also contribute to the literature on fixed-effect estimation in models with interactive fixed effects. Linear panel models with interactive fixed effects have been extensively studied, including the determination of the number of factors \citep{bai2002determining, onatski2009testing, ahn2013eigenvalue}, the estimation of common parameters \citep{Bai2009, MoonWeidner2015, MoonWeidner2017, BeyhumGautier2019, BeyhumGautier2022}, inference in the presence of weak factors \citep{Onatski2012, BaiNg2023, ChoiYuan2025, jiang2025biascorrectionfactoraugmentedregression, ArmstrongWeidnerZeleneev2025}, and extensions to network models \citep{sassi2024linear}; see \citet{BaiWang2016} for a comprehensive review. The theoretical development for nonlinear factor models remains comparatively limited. Notable contributions in this area include \citet{WANG2022180}, who studies the fixed-effect MLE in nonlinear pure factor models; \citet{GaoLiuPengYan2023}, who analyze binary choice models with interactive fixed effects and heterogeneous slope parameters and propose an information criterion for selecting the number of factors; and \citet{ChenFVWeidner2021}, who address bias correction in general nonlinear factor models and introduce an eigenvalue ratio test for factor number selection. More recently, \citet{zeleneev2026tractable} and \citet{yao2025low} propose a nuclear-norm penalized estimator that provides computational guarantees, facilitating reliable estimation.
\citet{boneva2017discrete} and \citet{chen2025common} generalize the CCE estimator to nonlinear settings.

Our work is also closely related to bootstrap-based bias correction and inference methods in the panel data literature. \citet{GONCALVES2015407} study bootstrap inference for linear dynamic panel models with individual fixed effects. \citet{KimSun2016} and \citet{HigginsJochmans2024} adopt parametric bootstrap methods for nonlinear panel models with one-way fixed effects. \citet{gonccalves2011moving} and \citet{higgins2025inference} introduce moving-block bootstrap procedures for dynamic panel data models.  \citet{LI2024105684} employ bootstrap methods for inference for treatment effects in interactive fixed-effect panel models. \citet{cavaliere2024bootstrap} develop a general framework for bootstrap inference in the presence of bias, showing that even when the bias cannot be consistently estimated, appropriately designed bootstrap methods can provide valid inference. Finally, our work is also related to the literature that uses transformation-based methods to improve bootstrap confidence intervals  \citep{efron1987better,konishi_normalizing_1991,hall_removal_1992}.

\section{Methodology}
\subsection{Model setup}
We consider nonlinear panel data models with interactive fixed effects. The model specification is
\begin{equation}\label{model 2.1}
Y_{it}|X_{it},\beta_{0},\alpha_0,\gamma_0\sim f\left(\cdot|X_{it}^{\prime}\beta_0+\alpha_{i0}^{\prime}\gamma_{0t}\right)
\end{equation}
for $i=1,\ldots,N$ and $t=1,\ldots,T$, where we observe 
$\{Y,X\}=\{Y_{it},X_{it}\}_{i=1\ldots,N;t=1,\ldots,T}$, $Y_{it}$ is the scalar response variable, $X_{it}$ is a $d_x$-dimensional covariate vector, the function $f$ is a known density with respect to some dominating measure, $\beta_0$ is the common parameter vector, and $\alpha_{i0}$ and $\gamma_{0t}$ are unobserved $d_f$-dimensional individual and time effects that appear in the model through a factor structure. We treat the number of factors $d_f$ as known.\footnote{In practice, we can estimate $d_f$ consistently; see \citet{ChenFVWeidner2021} and \citet{GaoLiuPengYan2023}.} 

We aim to estimate $\beta_0$, treating 
$\alpha_{0} = (\alpha_{1,0}, \ldots, \alpha_{N,0})$ and 
$\gamma_{0} = (\gamma_{1,0}, \ldots, \gamma_{T,0})$ as nuisance parameter matrices.
We make no assumptions on the joint distribution of $(X,\alpha_{0},\gamma_{0})$.
This flexibility is important since it reduces the likelihood
of model misspecification. Unlike the individual fixed-effect model \citep{HahnNewey2004, HahnKuersteiner2011}, which only allows the individual effect 
to be a time-invariant scalar, the interactive fixed-effect model 
accommodates multiplicative fixed effects. 
The additive two-way fixed-effect model \citep{FVWeidner2016}, which includes time effects 
to capture aggregate shocks---a common feature in economic applications---can, in fact, be viewed as a special case of model \eqref{model 2.1}. 
Interactive fixed effects allow for a much richer structure 
of unobserved heterogeneity \citep{Freyberger2018}.

\subsection{Interactive fixed-effect estimator}

Let $\mathcal{B}\subset \mathbb{R}^{d_x}$, $\mathcal{A}\subset  \mathbb{R}^{d_f}$, and $\mathcal{G}\subset  \mathbb{R}^{d_f}$ denote the parameter spaces of $\beta_0$, $\alpha_{i0}$, and $\gamma_{t0}$, respectively, and let $\Theta = \mathcal{B} \times \mathcal{A}^N \times \mathcal{G}^T$ denote the full parameter space of $\theta_0 = (\beta_0, \alpha_0, \gamma_0)$. 
The MLE of $\theta_0$ is
\begin{equation}\label{mle problem}
\hat\theta=
(\hat{\beta},\hat{\alpha},\hat{\gamma})
=\underset{(\beta, \alpha, \gamma) \in \mathcal{B} \times \mathcal{A}^N \times \mathcal{G}^T}{\arg\max}\,\mathcal{L}(\beta, \alpha, \gamma; Y, X),
\end{equation}
where $\mathcal{L}(\beta, \alpha, \gamma; Y, X)$ is the loglikelihood function,
\begin{equation}
\mathcal{L}(\beta, \alpha, \gamma; Y, X)
= \sum_{i=1}^N \sum_{t=1}^T 
\log f\!\left(Y_{it} \mid X_{it}^{\prime}\beta + \alpha_{i}^{\prime}\gamma_{t}\right).
\label{eq:loglikelihood}
\end{equation}
Since the dimension of the nuisance parameter matrices $\alpha_0$ and $\gamma_0$ increase with $N$ and $T$,
the estimation of $\beta_0$ depends on a large number of nuisance parameter estimates, which can induce substantial bias in $\hat\beta$. Even as $N,T \to \infty$, this bias persists in the limiting distribution of $\hat{\beta}$ and thus affects inference. Our main contribution is to show that the parametric bootstrap eliminates the asymptotic bias and delivers asymptotically valid confidence intervals.

As in the linear factor model \citep{Bai2009}, the unobserved effects $\alpha_{i0}$ and $\gamma_{t0}$ can only be identified up to a rotational indeterminacy since, for any nonsingular $d_f\times d_f$ matrix $A$, we have
$\alpha_{i0}'\gamma_{t0}=(A'\alpha_{i0})'(A^{-1}\gamma_{t0})$ for all $i$ and $t$.
Different normalizations do not affect the estimation of $\beta_0$ nor the estimation of average partial effects (APEs). For further discussion on normalization in factor models, see 
\citet{BaiWang2016}.  

We are also interested in APEs. For a chosen function $\mu_{it}(\beta,\alpha_{i},\gamma_{t})$ that characterizes the causal effect of interest to the researcher, we define the APE as
\begin{align}\label{APE}
\varDelta(\theta_0) = \frac{1}{NT}\sum_{i=1}^{N}\sum_{t=1}^{T}\mathbb{E}_{\theta_0}\!\left[\mu_{it}(\beta_0,\alpha_{i0},\gamma_{t0})\right],
\end{align}
where $\mathbb{E}_{\theta_0}$ denotes the conditional expectation, given the true values of the unobserved effects.\footnote{We focus on conditional APEs, given the realizations of the unobserved effects. Studying unconditional APEs requires additional assumptions on the distribution of $\{\alpha_{i0}\}_{i=1,\ldots,N}$ and $\{\gamma_{t0}\}_{t=1,\ldots,T}$, such as i.i.d.\ sampling or weak dependence. For further discussion, see \citet{FVWeidner2016}.}
For example, if $Y_{it}$ is binary and we consider the effect of a binary covariate $X_{it,k}$ on the probability that $Y_{it}=1$, we set $\mu_{it}(\beta,\alpha_{i},\gamma_{t})= f(1|Z_{it}+\beta_k)-f(1|Z_{it})$, where 
$Z_{it}=X_{it,-k}^{\prime}\beta_{-k}+\alpha_i'\gamma_t$, and $X_{it,-k}$ and $\beta_{-k}$ denote $X_{it}$ and $\beta$, respectively, with their $k$-th element removed.
If $Y_{it}$ is binary and $X_{it,k}$ is continuous, the average marginal effect of interest is obtained by setting $\mu_{it}(\beta,\alpha_{i},\gamma_{t}) = \beta_k\, f(1|X_{it}^{\prime}\beta+\alpha_{i}^{\prime}\gamma_{t})$.
A natural estimator of $\varDelta(\theta_0)$ is the MLE,
\begin{align}\label{plugin-ape}
\varDelta(\hat{\theta})
= \frac{1}{NT}\sum_{i=1}^{N}\sum_{t=1}^{T}
\mu_{it}(\hat{\beta},\hat{\alpha}_{i},\hat{\gamma}_{t}),
\end{align}
which inherits the bias of $(\hat{\beta},\hat{\alpha},\hat{\gamma})$.

\citet{ChenFVWeidner2021} study the MLE of model~\eqref{model 2.1} and propose a PCA-based algorithm to compute the estimator. They also provide analytical and jackknife corrections for the incidental parameter bias of
$\hat\beta$ and $\varDelta(\hat{\theta})$, thereby restoring asymptotically valid inference.
However, when $T$ is small, the analytical and jackknife bias corrections
may not always perform well; see our simulation results in Section 5. Recently, \citet{HigginsJochmans2024}
showed that the parametric bootstrap corrects the first-order
bias and provides asymptotically valid confidence intervals for $\beta_0$ in nonlinear
panel models with individual effects. Their parametric bootstrap confidence intervals
do not rely on normal approximations, and perform better compared to analytical corrections.
In this paper, we extend this idea to interactive fixed-effect models
and show in simulations that the parametric bootstrap delivers better small-$T$ performance in nonlinear interactive fixed-effect panel models compared with existing bias-correction methods. 
We further propose an improved bootstrap confidence interval based on a monotone transformation of the estimator.

\subsection{Bootstrap inference}
In this section, we describe the parametric bootstrap and informally explain why it yields a bias-corrected estimate of $\beta_0$ and asymptotically correct confidence intervals for $\beta_0$.
We begin with the data-generating process (DGP) used to generate 
$B$ bootstrap samples, each denoted by \(Y^*= \{Y^*_{it}\}_{i=1\ldots,N;t=1,\ldots,T} \).
Under the assumption that the covariates $X$ are strictly exogenous, the parametric bootstrap DGP is obtained 
from model \eqref{model 2.1} by plugging in the MLE \( \hat{\theta} \) and keeping \( X \) fixed:
\begin{equation}\label{bootstrap DGP}
Y^*_{it}|X_{it},\hat{\beta},\hat{\alpha},\hat{\gamma}\sim f\left(\cdot|X_{it}^{\prime}\hat{\beta}+\hat{\alpha_{i}}^{\prime}\hat{\gamma_{t}}\right).
\end{equation}
The bootstrap DGP in \eqref{bootstrap DGP} is the same as in \citet{HigginsJochmans2024} except that we have interactive fixed effects instead of one-way individual effects.

For each bootstrap dataset \( \{Y^*, X\} \), with $Y^*$ independently drawn according to (\ref{bootstrap DGP}), we compute the MLE
\[
\hat{\theta}^*
=(\hat\beta^*, \hat{\alpha}^*, \hat{\gamma}^*) =
\underset{(\beta, \alpha, \gamma) \in \mathcal{B} \times \mathcal{A}^N \times \mathcal{G}^T}{\arg\max}
\mathcal{L}( \beta, \alpha, \gamma;Y^*, X).
\]
Since \( \hat{\theta} \) is a consistent estimator---in the sense that, as $N,T\to\infty$, all of $\hat{\beta},\hat{\alpha}_i,\hat{\gamma}_t$, for every $i$ and $t$, are consistent for $\beta_0,\alpha_{i0},\gamma_{t0}$---\( \hat{\theta} \) 
should be close to \( \theta_0 \) with high probability when $N$ and $T$ are sufficiently large. Hence, under certain smoothness conditions, the parametric bootstrap DGP in \eqref{bootstrap DGP} will also be close to the true DGP, that is, to the DGP in \eqref{model 2.1}, which has $\theta=\theta_0$. Furthermore, the bootstrap datasets \( \{Y^*, X\} \) and the corresponding MLEs generated by these two very similar DGPs are expected to exhibit similar properties.

Under suitable regularity conditions, $\hat{\beta}^*$ admits an expansion analogous to that of $\hat{\beta}$. Moreover, the first-order term in the expansion of $\hat{\beta}^*$ converges in probability to the corresponding term in the expansion of $\hat{\beta}$. This implies that the bias
$\mathbb{E}_{\theta_0}(\hat{\beta}) - {\beta_0}$ can be accurately appoximated by
$\mathbb{E}^*(\hat{\beta}^*) - \hat{\beta}$, 
where $\mathbb{E}^*$ denotes the expectation under the parametric bootstrap DGP \eqref{bootstrap DGP}. 
In turn, $\mathbb{E}^*(\hat{\beta}^*) - \hat{\beta}$ can be approximated arbitrarily well by averaging over a large number, $B$, of parametric bootstrap replications. This average can be used to eliminate the first-order bias of $\hat{\beta}$, yielding the parametric bootstrap bias-corrected estimator
\begin{equation*}
\hat{\beta}_{\rm boot}=
        \hat{\beta}    -          \left(\frac{1}{B}\sum_{b=1}^B \hat{\beta}_b^* - \hat{\beta}  \right)
    = 2 \hat\beta - \frac{1}{B}\sum_{b=1}^B \hat{\beta}_b^*,
  \end{equation*}
where $\hat{\beta}_b^*$ is $\hat{\beta}^*$ computed from the $b$-th bootstrap dataset. Since the first-order bias is eliminated, 
\begin{equation*}
    \mathbb{E}_{\theta_0}\!(
        \hat{\beta}_{\rm boot} 
    ) - \beta_0 
    =
    o\!\left(
        \frac{1}{N} + \frac{1}{T}
    \right).
\end{equation*}
Note that, like other bias correction methods, the parametric bootstrap bias correction only provides asymptotic guarantees of bias reduction. With finite $N$ and $T$, it may amplify higher-order bias terms and, therefore, amplify the bias. In Section~4, we examine the finite-sample performance of the parametric bootstrap through simulations.

Another important observation is that the parametric bootstrap replicates the distribution of \( \hat{\beta} \) with a vanishingly small error.
For every \( a< b \in \mathbb{R}^{d_x} \),
\begin{equation*}
    \mathbb{P}_{\theta_0}(a \leq \sqrt{NT}(\hat{\beta}^* - \hat{\beta}) \leq b) - \mathbb{P}_{\theta_0}(a \leq \sqrt{NT}(\hat{\beta} - \beta_0) \leq b) \to 0,
\end{equation*}
where $\mathbb{P}_{\theta_0}(\cdot)$ denotes probabilities taken under the DGP \eqref{model 2.1}. 
As shown in \citet{HigginsJochmans2024}, this yields asymptotically valid confidence intervals for $\beta_0$ without the need for bias correction. 
To illustrate, we consider the simple case where $d_x= 1$. In case $d_x > 1$, a valid confidence interval for any linear combination $c^\prime \beta_0$ can be constructed similarly. 
Let $F^*$ denote the conditional distribution function of $\hat{\beta}^* - \hat{\beta}$ under DGP~\eqref{bootstrap DGP} given the original sample;  $F^*$ can be estimated to arbitrary precision by the empirical distribution of the parametric bootstrap replicates $ \hat{\beta}_b^*$ by setting $B$ large enough. In the next section, we formally show that \( F^* \) weakly converges to the limit distribution of \( \hat{\beta} -\beta_0\). Similarly, the quantile function corresponding to \( F^* \), defined as
\[
Q^*(\alpha) := \inf\{a: F^*(a) \geq \alpha\}, \qquad \alpha\in\mathbb{R},
\]
also converges to the quantile function of the limiting distribution of \( \hat{\beta}-\beta_0 \). Hence, we can simply use \( Q^*(\alpha) \) as the estimator of the \( \alpha \)-quantile of the distribution of \( \hat{\beta} - \beta_0 \), leading to an asymptotically valid confidence interval:  for every $\alpha\in(0,1)$,
\begin{equation}\label{CI}
\mathbb{P}_{\theta_0}\left(\beta_0 \in \left[ \hat{\beta} - Q^*(1-\alpha/2), \hat{\beta} - Q^*(\alpha/2) \right]\right) \to 1-\alpha,
\end{equation}
where the interval in brackets is a $(1-\alpha)$-level confidence interval for a scalar \( \beta_0 \). The interval can also be written as
\begin{equation*}
\left[ 2\hat{\beta} - Q^*_{\hat\beta^*}(1-\alpha/2), 2\hat{\beta} - Q^*_{\hat\beta^*}(\alpha/2) \right],
\end{equation*}
where $Q^*_{\hat\beta^*}$ is the conditional quantile function of $\hat\beta^*$ given the original sample.

Other methods to construct confidence intervals often rely on the asymptotic normality of \( \hat{\beta} \). They require a consistent estimator, \( \hat{V} \), of the asymptotic variance $V$ of \( \hat{\beta} \),
and an asymptotically unbiased estimator $\hat{\beta}^{bc}$ with the same asymptotic variance $V$. This gives,
for example, \( [\hat{\beta}^{bc} - 1.96 \sqrt{\hat{V}}, \hat{\beta}^{bc} + 1.96 \sqrt{\hat{V}}] \) as a 95\% confidence interval for a scalar $\beta_0$. The parametric bootstrap
does not require a bias-corrected estimator of $\beta_0$ nor a consistent estimator of $V$. Instead, it simply uses the quantile function of \( \hat{\beta}^* - \hat{\beta} \) and it does not rely on the Gaussian approximation to the distribution
of $\hat{\beta}$.

While the parametric bootstrap delivers asymptotically valid inference, its finite-sample performance may be affected by skewness in the distribution of $\hat{\beta}$, especially when the time dimension is small. As discussed above, the bootstrap procedure reproduces and may even amplify this skewness, which can result in asymmetric confidence intervals with overly wide bounds and thus lead to conservative inference.

To mitigate this issue, we propose to apply a monotone transformation to the bootstrap estimator prior to constructing confidence intervals. Let $\varphi(\cdot)$ be a strictly increasing transformation. We then construct the transformed bootstrap confidence interval as
\begin{equation}\label{transform CI}
\left[
\varphi^{-1}\!\left(2\varphi(\hat{\beta}) - Q^*_{\varphi(\hat{\beta}^*)}(1-\alpha/2)\right), \;
\varphi^{-1}\!\left(2\varphi(\hat{\beta}) - Q^*_{\varphi(\hat{\beta}^*)}(\alpha/2)\right)
\right],
\end{equation}
where $Q^*_{\varphi(\hat{\beta}^*)}$ is the conditional quantile function of $\varphi(\hat{\beta}^*)$ given the original sample.

Since $\varphi(\cdot)$ is strictly increasing, this transformation preserves the ordering of the estimator and thus maintains asymptotic validity. At the same time, by reducing skewness in the bootstrap distribution, it yields more symmetric and tighter confidence intervals, thereby improving finite-sample performance.

An alternative approach to account for skewness is the bias-corrected and accelerated (BCa) bootstrap of \citet{efron1987better}. However, this method is not well suited to our setting. The BCa procedure adjusts confidence intervals based on an estimate of bias and skewness inferred from the bootstrap distribution. When the estimator exhibits non-negligible asymptotic bias, as in our case, the bootstrap distribution inherits this bias and is effectively shifted relative to the true sampling distribution.

As a result, the BCa correction may confound bias with skewness and mischaracterize the shape of the distribution. For example, if the asymptotic bias shifts the distribution to the right, the BCa method may interpret this shift as evidence of skewness and consequently adjust the confidence interval in the opposite direction, leading to an excessive leftward shift. This can distort coverage and produce misleading inference. In contrast, our transformation-based approach directly addresses the asymmetry of the distribution while leaving its location unaffected, thereby providing a more reliable correction in the presence of skewness.

By analogy to (\ref{CI}) and (\ref{transform CI}), we can also construct a parametric bootstrap confidence interval for $\varDelta({\theta_0})$. By plugging $\hat{\theta}^*$ into the APE, we obtain 
$$\ \varDelta(\hat{\theta}^*) = \frac{1}{NT} \sum_{i=1}^N \sum_{t=1}^T \mu(X_{it},\hat{\beta}^*,\hat{\alpha}_i^*,\hat{\gamma}_t^*),$$
which allows us to construct a bootstrap confidence interval for
$\varDelta({\theta_0})$, since the distribution of \( \sqrt{NT}(\varDelta(\hat{\theta}^*) - \varDelta(\hat{\theta})) \) approximates that of \( \sqrt{NT}(\varDelta(\hat{\theta}) - \varDelta(\theta_0)) \) as \( N,T \) becomes large.

\section{Asymptotic theory}

In this section, we study the asymptotic properties of the parametric bootstrap in an asymptotic regime where $N$ and $T$ go to infinity with $T/N \to \kappa \in (0,\infty)$. For simplicity, we focus on the original bootstrap estimator. The results for the transformed bootstrap can be obtained by combining bootstrap consistency for the transformed estimator, established via the bootstrap delta method\footnote{This requires that $\varphi$ be continuously differentiable at the true parameter $\beta_0$ with $\varphi'(\beta_0)\neq 0$, so that the delta method and its bootstrap analogue apply; see, e.g., Chapter 23 of \citet{van1998asymptotic}.}, with the fact that, under monotonicity, the coverage of the confidence interval for $\beta_0$ is equivalent to that for $\varphi(\beta_0)$.

Our proofs closely follow those of \citet{KimSun2016} and \citet{HigginsJochmans2024}. We first extend the asymptotic expansion for the general M-estimator established in \citet{FVWeidner2016} to ensure that it holds not only at the true parameter value but also uniformly in a neighborhood of the true value. Achieving this requires that the assumptions in \citet{FVWeidner2016} hold uniformly in a small neighbourhood of $\theta_{0}$. Similar uniformity assumptions are imposed in \citet{HigginsJochmans2024} for the one-way fixed-effect model.
Then we verify that nonlinear panel models with interactive fixed effects satisfy these high-level conditions, thereby establishing the validity of uniform expansions for the MLEs of $\beta_0$ and $\varDelta(\theta_0)$. These uniformity results, combined with the consistency of the MLE and a lemma from \citet{Andrews2005}, show the validity of the parametric bootstrap to approximate the distribution of the MLE.

We treat the unobserved effects as fixed parameters. Alternatively, our analysis can be interpreted as being conditional on the realization of the unobserved effects, as in \citet{FVWeidner2016} and \citet{ChenFVWeidner2021}, provided that the assumptions are modified by conditioning on $(\alpha_0, \gamma_0)$. See Remark~1 of \citet{HahnKuersteiner2011} for more discussion.

Since we generate the bootstrap sample $\{Y^{*}, X\}$ based on model~\eqref{bootstrap DGP}, the covariates are kept fixed when constructing the bootstrap sample. Throughout the paper, we let $\mathbb{E}_{\theta}$ and $\mathbb{P}_{\theta}$ denote the expectation and probability operators under the bootstrap data-generating process.\footnote{Strictly speaking, since the joint distribution of the covariates $X$ and the unobserved effects $\alpha_0,\gamma_0$ is unrestricted, the bootstrap data-generating process is indexed by $(\theta,\alpha_0,\gamma_0)$. Accordingly, the corresponding expectation and probability operators should be written as $\mathbb{E}_{\theta,\alpha_0,\gamma_0}$ and $\mathbb{P}_{\theta,\alpha_0,\gamma_0}$. For notational simplicity, we suppress the dependence on $\alpha_0$ and $\gamma_0$.}

We make the following assumptions. They are similar to, but slightly stronger than, Assumption~1 in \citet{ChenFVWeidner2021}, as some of our assumptions are required to hold uniformly over a neighborhood $\Theta_0$ of the true parameter value.

\begin{assumption}[]\label{assu:parameterspace}
The parameter space $\Theta$ is a compact set. Furthermore, there  
exists a constant $\varepsilon > 0$  and an open neighborhood $\Theta_0 \subset \Theta$ that contains $\theta_0$ and satisfies $d(\theta, \theta_0) < \varepsilon$ 
for every $\theta \in \Theta_0$, where
\begin{align}\label{metric}
d(\theta, \theta_0)
= \big\Vert 
\left(
\beta^{\prime},\, \mathrm{vec}(\alpha)^{\prime},\, \mathrm{vec}(\gamma)^{\prime}
\right)^{\prime}
-
\left(
\beta_0^{\prime},\, \mathrm{vec}(\alpha_0)^{\prime},\, \mathrm{vec}(\gamma_0)^{\prime}
\right)^{\prime}
\big\Vert
\end{align}
and $\Vert \cdot \Vert$ denotes the Euclidean norm.
\end{assumption}

The following assumption involves a mixing condition, which we use to bound the moments. Let $\mathcal{A}_t^i$ and $\mathcal{B}_t^i$ be the $\sigma$-algebras generated by $\{ X_{i1}, Y_{i1}. \ldots, X_{it}, Y_{it} \}$
and $\{ X_{it}, Y_{it}, \allowbreak X_{i,t+1}, Y_{i,t+1} , \ldots \}$, respectively, and define the 
$\alpha$-mixing coefficients
\[
\alpha_i(m,\theta) := \sup_t \sup_{A \in \mathcal{A}_t^i,\, B \in \mathcal{B}_{t+m}^i} 
\left| \mathbb{P}_\theta(A \cap B) - \mathbb{P}_\theta(A)\mathbb{P}_\theta(B) \right|.
\]

\begin{assumption}[Mixing]\label{assu:correct}
(i) Conditional on \( X \), the data $Y=\{Y_{it}\}_{i=1\ldots,N;t=1,\ldots,T}$ 
are generated from model \eqref{model 2.1}.  (ii) For each $t$, the data $\{X_{it},Y_{it}\}_{i=1,\ldots,N}$ is independent across $i$, and
there exists a constant $\mu>0$ such that, for each $i$, $\{X_{it},Y_{it}\}_{t=1,\ldots,T}$  is $\alpha$-mixing with mixing coefficients $\alpha_i(m,\theta)$ satisfying $\sup_{\theta\in\Theta_0}\max_i \alpha_i(m,\theta)=O(m^{-\mu})$ as $m\to\infty$.
\end{assumption}

We assume that the covariates are strictly exogenous.\footnote{However, our approach could be extended to include lagged dependent variables affecting $Y_{it}$. More general feedback mechanisms---such as cases where $Y_{it}$ affects $X_{is}$ for $s>t$---are excluded. For further discussion, see \citet{ChenFVWeidner2021}.} The conditional independence in Assumption \ref{assu:correct}(i) is commonly adopted in the fixed-effects panel data literature \citep{HahnNewey2004, Bonhomme2012, FVWeidner2016, ChenFVWeidner2021}, as it implies that the conditional distribution of $Y_{it}$ does not depend on $\alpha$, $\gamma$, or $X$ beyond what is captured by the specified model.\footnote{Recently, several papers studying average partial effects in nonlinear panel data models with individual fixed effects relax this assumption \citep{GrahamPowell2012, liu2024identification, BotosaruMuris2024}.} 
The mixing condition, also assumed in \citet{FVWeidner2016}, is used to bound the moments and covariances, and to apply the WLLN and CLT. Moreover, we require the mixing coefficients to decay uniformly over $\theta \in \Theta_0$ to uniformly bound the moments, which is necessary to establish uniform convergence of the bias term and uniform asymptotic normality. The uniform decay condition is also adopted in \citet{HigginsJochmans2024}.

\begin{assumption}[Strong factors]\label{ass:strong factor}
There exist positive definite matrices $\Sigma_{\alpha}$ and $\Sigma_{\gamma}$ such that
$\frac{1}{N}\sum_{i=1}^{N} \alpha_{i0}\alpha_{i0}^{\prime} \;\to\; \Sigma_{\alpha}$ and
$\frac{1}{T}\sum_{t=1}^{T} \gamma_{t0}\gamma_{t0}^{\prime} \;\to\; \Sigma_{\gamma}$.
\end{assumption}

Assumption~\ref{ass:strong factor}, commonly referred to as the strong factors assumption, is widely used in panel interactive fixed effects models \citep{Bai2009, MoonWeidner2015, MoonWeidner2017,   ChenFVWeidner2021, GaoLiuPengYan2023}.\footnote{For discussion of weak factors, see \citet{Onatski2012, BeyhumGautier2019, BeyhumGautier2022, BaiNg2023, ChoiYuan2025, jiang2025biascorrectionfactoraugmentedregression, ArmstrongWeidnerZeleneev2025}.} In our setting, this assumption is required to establish a uniform convergence rate of the MLE, which is needed to apply the asymptotic expansion in Appendix~\ref{Appendix C}, a uniform version of the expansion in \citet{FVWeidner2016}. A necessary step is to ensure that the matrices $\frac{1}{N}\sum_{i=1}^{N}\alpha_{i}\alpha_{i}^{\prime}$ and $\frac{1}{T}\sum_{t=1}^{T}\gamma_{t}\gamma_{t}^{\prime}$ converge uniformly, over $\theta \in \Theta_{0}$, to positive definite limits. Given the fact that $\Theta_{0}$ is an $\epsilon$-neighborhood of $\theta_0$, the differences
$\frac{1}{N}\sum_{i=1}^{N}\alpha_{i}\alpha_{i}^{\prime}-\frac{1}{N}\sum_{i=1}^{N}\alpha_{i0}\alpha_{i0}^{\prime}$ and
$\frac{1}{N}\sum_{i=1}^{N}\gamma_{i}\gamma_{i}^{\prime}-\frac{1}{N}\sum_{i=1}^{N}\gamma_{i0}\gamma_{i0}^{\prime}$ 
converge to zero for every $\theta\in\Theta_0$. Therefore, Assumption~\ref{ass:strong factor} is sufficient to establish the required uniform consistency result. 
For a formal proof, see Lemma~\ref{lemma:A1}.

Let $X_k$ be the $N \times T$ matrix elements $X_{it,k}$ ($i = 1, \ldots, N$; $t = 1, \ldots, T$).
For any matrix $A$, define the residual-maker matrix 
$\mathcal{M}_A = \mathbb{I} - A (A^{\prime} A)^{\dagger} A^{\prime}$
where ${}^\dagger$ denotes the Moore-Penrose inverse and $\mathbb{I}$ is the identity matrix of appropriate dimension. Let $\operatorname{Tr}(\cdot)$ denote the trace of a matrix.

\begin{assumption}[Generalized non-collinearity]\label{Generalized non-collinearity}
 There exists a constant $c>0$ such that, for every $\theta\in\Theta_0$ and every $d_x \times T$ matrix $\widetilde\gamma$,
the $d_x \times d_x$ matrix $D(\alpha, \widetilde\gamma)$ with elements
\[
D_{k_1 k_2}(\alpha, \widetilde\gamma)
= (NT)^{-1} \operatorname{Tr}\!\left(
\mathcal{M}_{\alpha'} X_{k_1} \mathcal{M}_{\widetilde\gamma'} X_{k_2}^{\prime}
\right),
\quad k_1, k_2 \in \{1, \ldots, d_x\},
\]
satisfies 
$\inf_{\theta \in \Theta_0} \lambda_{\min}\!\left(D(\alpha, \widetilde\gamma)\right) > c$
wpa1, where $\lambda_{\min}(\cdot)$ denotes the smallest eigenvalue.
\end{assumption}

Assumption~\ref{Generalized non-collinearity} requires that the unobserved factors and factor loadings cannot fully capture the variation of any covariate, and that this condition holds not only for the true factors, but also for all $\alpha$ in a neighborhood of $\alpha_0$. It ensures that $\beta_0$ is point-identified.

Let $Z_{it}(\theta)=X_{it}^{\prime}\beta + \alpha_i^{\prime}\gamma_t$ and define
$\ell_{it}(z) = \log f\!(Y_{it}|Z_{it}(\theta)=z)$.
Let $\partial_{z^q} \ell_{it}(z)$ denote the $q$-th derivative of $\ell_{it}(z)$.
When evaluating these functions at $z=Z_{it}(\theta_0)$, we omit the argument for simplicity.
For instance, we write $\partial_{z^{q}} \ell_{it}$ instead of 
$\partial_{z^{q}} \ell_{it}(Z_{it}(\theta_0))$.

For $z \in \mathbb{R}^q$ and $\delta > 0$, let 
$\mathcal{B}(\delta, z) \subset \mathbb{R}^q$ denote the open ball of radius 
$\delta$ centered at $z$, taken with respect to the topology induced by the 
Euclidean metric on $\mathbb{R}^q$.

\begin{assumption}[Smoothness and moments]\label{assu:smooth}
 There exists $\delta > 0$ such that for every $i,t$ and every $\theta \in \Theta_0$,  the following conditions hold almost surely:
\begin{itemize}
    \item[(i)] The function $\ell_{it}(\cdot)$ is four times continuously differentiable on 
    $\mathcal{B}(\delta,Z_{it}(\theta))$.
    \vspace{0.4em}
    \item[(ii)] The derivatives of $\ell_{it}(\cdot)$ up to fourth order are 
    uniformly bounded in absolute value on $\mathcal{B}(\delta,Z_{it}(\theta))$ by a random envelope 
    $L_{it}(Z_{it}(\theta),\theta) > 0$ satisfying
    \[
    \sup_{\theta \in \Theta_0} 
    \max_{i,t} 
    \mathbb{E}_{\theta}\!\left[L_{it}(Z_{it}(\theta),\theta)^{\,8+\nu} \right] < c,
    \]
    for some constants $\nu>0$ and $c>0$, uniformly in $N,T$.
    \item[(iii)] The regressors $X_{it}$ are uniformly bounded in absolute value over $i,t,N,T$.
\end{itemize}
\end{assumption}

\begin{assumption}[Concave loglikelihood]\label{assu:convexity}
The function $\ell_{it}(z)$ is strictly concave in $z$ a.s.  
Furthermore, there exist constants $\underline{b}>0$, $\overline{b}>0$, and $\delta>0$ such that, for every $\theta \in \Theta_0$,
\[
\underline{b} \leq  \min_{i,t} 
\{-\partial_{z^2} \ell_{it}(z)\}
\leq  \max_{i,t} 
\{-\partial_{z^2} \ell_{it}(z)\} \leq \overline{b} \quad \text{a.s.},
\]
for all $z \in \mathcal{B}(\delta, Z_{it}(\theta))$ and all $N,T$.
\end{assumption}


 In Assumption \ref{assu:smooth}, the dominating function \(L_{it}(Z_{it},\theta) \) may depend on \( \theta\). Unlike \citet{FVWeidner2016}, who only require boundedness at
 the true value \( \theta_0 \), we impose the condition uniformly over $ \Theta_0$. The boundedness condition on $X_{it}$ is also imposed in \citet{ChenFVWeidner2021} and \citet{GaoLiuPengYan2023}.\footnote{In \citet{GaoLiuPengYan2023}, it is required that 
\( Z_{it0} = X_{it}^{\prime}\beta_0 + \alpha_{i0}^{\prime}\gamma_{t0} \)
be uniformly bounded over \( i,t \) wpa1.}
The uniform boundedness condition is necessary for establishing uniform asymptotic expansions and uniform asymptotic normality.

Assumption \ref{assu:convexity} implies point identification for any \( \theta \in \Theta_0 \) that might be the true value, as shown in \citet{FVWeidner2016}. Additionally, we require that the Hessian is uniformly bounded away from zero and infinity, which is necessary for uniform consistency and a uniform expansion of the MLE. Importantly, this assumption is satisfied for many popular nonlinear models, e.g, logit and probit models.

Given Assumptions \ref{assu:parameterspace}--\ref{assu:convexity}, we show in Appendix~\ref{Appendix B} that the following expansion holds uniformly in $\theta\in\Theta_0$:
\begin{equation}\label{eq:uniform_expansion}
\sqrt{NT}(\hat{\beta} - \beta)
    = \overline{W}_{NT}^{-1}(\theta)\bigl(U_{NT}(\theta) + \overline{B}_{NT}(\theta)\bigr)
    + r_{NT}(\theta).
\end{equation}
Here, the leading term $U_{NT}(\theta)$ has mean zero and variance 
$\overline{W}_{NT}(\theta)$.
As $N,T\to\infty$, and provided that the required limits exist, the term $U_{NT}(\theta)$ 
determines the asymptotic variance, whereas the term $\overline{B}_{NT}(\theta)$ governs the asymptotic bias. The remainder term $r_{NT}(\theta)$ is asymptotically negligible uniformly in $\theta \in \Theta_0$. For brevity, we omit the explicit expressions for $U_{NT}(\theta)$, $\overline{B}_{NT}(\theta)$, and $\overline{W}_{NT}(\theta)$, which are involved; detailed formulas are provided in Appendix~\ref{Appendix B}.

To establish uniform asymptotic normality, we require that the first term in the above expansion converges in distribution to a normal distribution and that the bias term converges to a constant uniformly. However, these quantities depend on the realization of the unobserved effects $\alpha$ and $\gamma$. For instance, if the time effect $\gamma_t$ is a nonstationary time series, the bias term $\overline{B}_{NT}(\theta)$ may not converge. Moreover, to ensure uniform asymptotic normality, the approximation error must decay uniformly for every $\theta \in \Theta_0$. Formally, we state the following assumption:
\begin{assumption}
\label{assu:limit}
For every $\theta \in \Theta_0$, the limits
$\overline{W}_\infty(\theta)=\lim_{N,T \to \infty} \overline{W}_{NT}(\theta)$ and 
$\overline{B}_\infty(\theta)=\lim_{N,T \to \infty} \overline{B}_{NT}(\theta)$ exist, and
$\overline{W}_\infty(\theta)$ is positive definite. Furthermore, the converge of 
$\overline{W}_{NT}(\theta)$ and 
$\overline{B}_{NT}(\theta)$ to their limits is uniform over $\Theta_0$.
\end{assumption}
\citet{FVWeidner2016} assume that 
$\overline{W}_{NT}(\theta_{0})$ and $\overline{B}_{NT}(\theta_{0})$ 
converge in probability to nonrandom limits as $N,T\to\infty$, which implies the 
asymptotic normality of $\hat{\beta}$ unconditionally w.r.t.\ the unobserved effects. 
Our assumption is stronger in that it requires the convergence to 
hold for any realization of unobserved effects (almost surely), and uniformly 
over $\Theta_{0}$. 
However, it also allows the limits to depend on $\theta$. A similar assumption is implicitly imposed in \citet{HigginsJochmans2024}.

Given Assumptions \ref{assu:parameterspace}--\ref{assu:limit}, 
we show that bootstrap consistency holds for any realization of the unobserved effects. 
This result provides the theoretical foundation for the parametric bootstrap. 
The corresponding result for one-way fixed-effect models is given in 
\citet{HigginsJochmans2024}.

\begin{theorem}[Uniform asymptotic normality]
\label{theorem3.1} Suppose Assumptions \ref{assu:parameterspace}--\ref{assu:limit} hold. Then, for every $a\in\mathbb{R}^{d_x}$,
\[
\sup_{\theta\in\Theta_0}|\mathbb{P}_{\theta}(\sqrt{NT}(\hat{\beta}-\beta)\leq a)-G_{\theta}(a)| \to 0,
\]
where $G_{\theta}$ is the distribution function corresponding to
$\mathcal{N}\left(\overline{B}_{\infty}(\theta),\overline{W}_\infty(\theta)^{-1}\right)$.
\end{theorem}

The next result establishes that the parametric bootstrap approximates the distribution of the MLE, ensuring the asymptotic validity of confidence intervals constructed using the quantiles of the bootstrap distribution.

\begin{theorem}[Bootstrap consistency for the MLE]
\label{theorem3.2}
Suppose Assumptions \ref{assu:parameterspace}--\ref{assu:limit} hold. Then
\[
\mathbb{P}_{\theta_0}\left[\sup_{a\in\mathbb{R}^{d_x}}
\left|\mathbb{P}_{\hat{\theta}}(\sqrt{NT}(\hat{\beta}^{*}-\hat{\beta})\leq a)-\mathbb{P}_{\theta_{0}}(\sqrt{NT}(\hat{\beta}-\beta_{0})\leq a)\right|>\epsilon\right] \to 0
\]
for all $\epsilon>0$.
\end{theorem}

The following corollary establishes the asymptotic validity of the transformed bootstrap confidence interval. It follows from the bootstrap delta method \citep{van1998asymptotic} that the bootstrap consistency result extends to smooth transformations of the estimator. A more general result can be found in the proof of Theorem~2 in \citet{higgins2025inference}.
\begin{corollary}[Bootstrap consistency for transformed MLE]
\label{corollary3.3}
Suppose Assumptions \ref{assu:parameterspace}--\ref{assu:limit} hold. Let $c \in \mathbb{R}^{d_x}$ be a fixed vector and define $\phi = \varphi(c^\prime \beta)$, where $\varphi:\mathbb{R}\to\mathbb{R}$ is continuously differentiable at $c^\prime \beta_0$ with $\varphi'(c^\prime \beta_0)\neq 0$. Then
\[
\mathbb{P}_{\theta_0}\left[\sup_{a\in\mathbb{R}}
\left|
\mathbb{P}_{\hat{\theta}}\bigl(\sqrt{NT}(\varphi(c^\prime \hat{\beta}^*)-\varphi(c^\prime \hat{\beta}))\le a\bigr)
-
\mathbb{P}_{\theta_0}\bigl(\sqrt{NT}(\varphi(c^\prime \hat{\beta})-\varphi(c^\prime \beta_0))\le a\bigr)
\right|>\epsilon\right]\to 0
\]
for all $\epsilon>0$.
\end{corollary}
In our simulations, we further restrict $\varphi(\cdot)$ to be strictly increasing and continuously differentiable.

We make the following assumptions to establish the validity of the parametric bootstrap for the APEs defined in \eqref{APE}.

\begin{assumption}[]\label{assu:APEs interactive}
The unobserved effects enter through a factor structure:
\[
\mu_{it}(\beta, \alpha_i, \gamma_t)=\mu^f_{it}(\beta, \pi_{it}), \qquad \pi_{it}=\alpha_i^\prime\gamma_t,
\]
for some function $\mu^f_{it}$.
\end{assumption}

\begin{assumption}[]\label{assu:APEs smooth}
There exists $\delta>0$ such that for every $i,t, N, T$ and every $\theta\in\Theta_0$, (i) \( \mu^f_{it}(\cdot, \cdot) \) is four times continuously differentiable over \( \mathcal{B}(\delta;\beta, \pi_{it}) \) a.s.; (ii) the partial derivatives of \( \mu^f_{it}(\cdot, \cdot) \), up to the fourth order, are uniformly bounded in absolute value over $\mathcal{B}(\delta;\beta, \pi_{it}) $. 
\end{assumption}

Given Assumptions \ref{assu:parameterspace}--\ref{assu:APEs smooth}, the following uniform asymptotic expansion holds for APEs:\footnote{See  Theorem \ref{Thm:C2} for expressions of the terms in the expansion.}
\[
\sqrt{NT}\left(\varDelta(\hat{\theta}) - \varDelta(\theta)\right)
    =  U_{NT}^{\varDelta}(\theta) + \overline{B}_{NT}^{\varDelta}(\theta) + r_{NT}^{\varDelta}(\theta).
\]
Assuming that limits exist, as $N,T\to\infty$,  the leading term $U_{NT}^{\varDelta}(\theta)$ contributes to the asymptotic variance, the second term $\overline{B}_{NT}^{\varDelta}(\theta)$ captures the asymptotic bias, and the remainder term $r_{NT}^{\varDelta}(\theta)$ is asymptotically negligible uniformly over $\Theta_0$.  

To establish uniform asymptotic normality of the MLE of APEs and bootstrap consistency, we also require a uniform convergence condition. Let the conditional variance of the leading term be denoted by
$\overline{W}_{NT}^{\varDelta}(\theta)$.
\begin{assumption}
\label{assu:limit:ape}
For every $\theta \in \Theta_0$, the limits
$\overline{W}^\varDelta_{\infty}(\theta) =\lim_{N,T \to \infty} \overline{W}^\varDelta_{NT}(\theta)$ and
$\overline{B}^\varDelta_{\infty}(\theta) =\lim_{N,T \to \infty} \overline{B}^\varDelta_{NT}(\theta)$
exist, and $\overline{W}^\varDelta_{\infty}(\theta)>0$. Furthermore, the convergence of 
$\overline{W}^\varDelta_{NT}(\theta)$ and $\overline{B}^\varDelta_{NT}(\theta)$ to their limits is uniform over $\Theta_0$.
\end{assumption}

We can now show bootstrap consistency for APEs.
\begin{theorem}[Uniform asymptotic normality for APEs]
\label{theorem3.3} Suppose Assumptions \ref{assu:parameterspace}--\ref{assu:limit:ape} hold. Then, for every $a\in\mathbb{R}$,
\[
\sup_{\theta\in\Theta_0}|\mathbb{P}_{\theta}(\sqrt{NT}(\varDelta(\hat{\theta}) - \varDelta(\theta))\leq a)-G^\varDelta_{\theta}(a)| \to 0,
\]
where $G^\varDelta_{\theta}$ is the distribution function corresponding to
$\mathcal{N}(\overline{B}^\varDelta_{\infty}(\theta),\overline{W}^\varDelta_{\infty}(\theta)^{-1})$.
\end{theorem}

\begin{theorem}[Bootstrap consistency for APEs]\label{theorem3.4}
Suppose Assumptions \ref{assu:parameterspace}--\ref{assu:limit:ape} hold. Then
\[
\mathbb{P}_{\theta_0}\left[\sup_{a\in \mathbb{R}}\left|\mathbb{P}_{\hat{\theta}}(\sqrt{NT}(\varDelta(\hat{\theta}^{*})-\varDelta(\hat{\theta}))\leq a)-\mathbb{P}_{\theta_0}
(\sqrt{NT}(\varDelta(\hat{\theta})-\varDelta(\theta_{0}))\leq a)
\right|>\epsilon\right] \to 0
\]
for all $\epsilon>0$.
\end{theorem}

We focus on APEs of the form~\eqref{APE}, 
which means that we consider the average partial effect evaluated at the realization of the unobserved effects, rather than averaged over the marginal distributions of the unobserved effects (i.e., the unconditional APEs).
There are two reasons for this choice. 
First, studying the properties of unconditional APEs requires additional assumptions on the distributions of $\{\alpha_i\}_{i=1,\ldots,N}$ and $\{\gamma_t\}_{t=1,\ldots,T}$, 
for example, that $\{\alpha_i\}$ are i.i.d.\ across $i$ and $\{\gamma_t\}$ form a stationary time series.
Second, as discussed in \citet{FVWeidner2016}, although the plug-in estimator $\hat{\Delta}$ is still consistent for the unconditional APEs, its convergence rate is slower than $\sqrt{NT}$.

To see this, let $\mathbb{E}(\Delta)$ denote the unconditional APE, where the expectation is taken over the marginal distributions of the unobserved effects.
The difference between the plug-in estimator and the unconditional APE can be decomposed as
\begin{align} \label{eq:decomposition}
{\Delta(\hat\theta)}-\mathbb{E}(\Delta)
= \underbrace{{\Delta(\hat\theta)}-\Delta(\theta)}_{\text{parameter estimation error}}
+ \underbrace{\Delta(\theta)-\mathbb{E}(\Delta)}_{\text{sample mean error}},
\end{align}
where the first term arises from parameter estimation error and the second term reflects the difference between the sample mean and the population mean.

As discussed in \citet{FVWeidner2016}, under mild weak-dependence assumptions on $\{\alpha_i\}_{i=1,\ldots,N}$ and $\{\gamma_t\}_{t=1,\ldots,T}$, the variance of the second term dominates that of the first term. Hence, in this case, there is no asymptotic bias, but the convergence rate becomes slower than $\sqrt{NT}$.\footnote{\citet{FVWeidner2016} show that if $\{\alpha_i\}_{i=1,\ldots,N}$ and $\{\gamma_t\}_{t=1,\ldots,T}$ are both independent sequences and $\alpha_i$ and $\gamma_t$ are independent for all $i,t$, then the convergence rate of the second term in \eqref{eq:decomposition} is $\sqrt{NT/(N+T-1)}$, see Remark 4 in \citet{FVWeidner2016} for further details.} 
Our bootstrap procedure focuses on addressing the asymptotic bias in the first term, although the resulting bias-corrected estimator may still provide improved finite-sample performance for the estimation of unconditional APEs. 
For further details, see Remark~4 in \citet{FVWeidner2016}.

\section{Implementation}


Since the loglikelihood function in \eqref{eq:loglikelihood} is not strictly concave due to the presence of interactive fixed effects, it may be difficult to compute the MLE, i.e., the global maximizer of the loglikelihood. \citet{ChenFVWeidner2021} suggest using multiple initial values to increase the probability of reaching the global maximizer. However, this strategy does not guarantee that the global maximizer is reached, and it may substantially increase the computational burden, particularly when combined with bootstrapping.

To address the computational challenge, \citet{zeleneev2026tractable} extend the nuclear-norm penalized estimator proposed by \citet{moon2026nuclear} to nonlinear factor models. Their approach replaces the non-convex rank constraint arising from the interactive fixed effects with a nuclear-norm penalty, thereby relaxing the original optimization problem to a strictly convex, computationally tractable problem. \citet{zeleneev2026tractable} establish the consistency of the resulting penalized estimator and further show that using this estimator as the initial value in a gradient descent algorithm leads to convergence to the global maximizer of the loglikelihood, making it asymptotically equivalent to the MLE.

Given the computational efficiency and asymptotic equivalence of the two-step estimator proposed by \citet{zeleneev2026tractable}, we adopt this estimator as a computationally convenient proxy for the MLE in the subsequent analysis. Formally, we first solve the following nuclear-norm penalized optimization problem:
\begin{equation}\label{nuc problem}
(\hat{\beta}_{\text{nuc}}, \hat{\Sigma}_{\text{nuc}})
= \underset{(\beta, \Sigma) \in \mathcal{B} \times \mathbb{R}^{N \times T}}{\arg\max} 
\left\{
\frac{1}{NT}\mathcal{L}(\beta, \Sigma)
+ \frac{\varphi}{\sqrt{NT}} \left\lVert \Sigma \right\rVert_{\text{nuc}}
\right\}
,
\end{equation}
where $\mathcal{L}(\beta,\Sigma)
= \sum_{i=1}^{N}\sum_{t=1}^{T}
\!\log f\big(Y_{it}\mid X_{it}^{\prime}\beta+\Sigma_{it}\big)$
is the loglikelihood function without low-rank constraint on $\Sigma$, $\Sigma$ is an $N \times T$ matrix collecting the unobserved effects for each $(i,t)$,  
$\lVert\cdot\rVert_{\text{nuc}}$ is the nuclear norm, 
and $\varphi$ is a tuning parameter.

Section 4.3 of \citet{zeleneev2026tractable} proposes a data-dependent procedure for selecting the tuning parameter $\varphi$. Specifically, they first compute a two-way fixed-effect estimator, obtaining $\tilde{\beta}$ and the scalar effects $\tilde{\alpha}_1,\ldots,\tilde{\alpha}_N$,  and $\tilde{\gamma}_1,\ldots,\tilde{\gamma}_T$, and form an initial guess
\begin{equation}\label{choicevarphi}
\tilde{\varphi}
= 1.05 \left\| \partial_{\Sigma} \mathcal{L}\bigl(\tilde{\beta}, \tilde{\Sigma}\bigr) \right\|_{\mathrm{op}},
\end{equation}
where $\tilde{\Sigma}_{it}=\widetilde{\alpha_i}+\widetilde{\gamma_t}$ and $\partial_{\Sigma} \mathcal{L}(\beta,\Sigma)\in\mathbb{R}^{N\times T}$ is the matrix of partial derivatives of the objective function with respect to $\Sigma$, and $\|\cdot\|_{\mathrm{op}}$ denotes the operator norm (i.e., the largest singular value). 

Solving \eqref{nuc problem} with \(\varphi=\tilde{\varphi}\) yields \(\tilde{\beta}_{\mathrm{nuc}}\) and an initial estimate \(\tilde{\Sigma}_{\mathrm{nuc}}^{\mathrm{init}}\). The singular value decomposition of \(\tilde{\Sigma}_{\mathrm{nuc}}^{\mathrm{init}}\) is then computed, and the first \(d_f\) left and right singular vectors are retained to construct initial estimated factor loadings \(\tilde{\alpha}_{\mathrm{nuc}}^{\mathrm{init}}\in\mathbb{R}^{d_f\times N}\) and factors \(\tilde{\gamma}_{\mathrm{nuc}}^{\mathrm{init}}\in\mathbb{R}^{d_f\times T}\). The tuning parameter is subsequently updated as
\begin{equation}
\label{eq:update}
\hat{\varphi}
=
1.05
\left\|
\partial_{\Sigma}\mathcal{L}\bigl(\tilde{\beta}_{\mathrm{nuc}},\tilde{\Sigma}_{\mathrm{nuc}}\bigr)
\right\|_{\mathrm{op}},
\end{equation}
where the entries of \(\tilde{\Sigma}_{\mathrm{nuc}}\) are given by
$(\tilde{\Sigma}_{\mathrm{nuc}})_{it}
=
(\tilde{\alpha}_{\mathrm{nuc},i}^{\mathrm{init}})^{\prime}(\tilde{\gamma}_{\mathrm{nuc},t}^{\mathrm{init}}).$

In an unreported simulation, we found that the choice in \eqref{choicevarphi} tends to induce excessive shrinkage in the singular values of the initial estimate \(\tilde{\Sigma}_{\mathrm{nuc}}^{\mathrm{init}}\). To mitigate this issue, we instead scale down the initial choice and set
\[
\tilde{\varphi}
= 0.5 \left\| \partial_{\Sigma} \mathcal{L}\bigl(\tilde{\beta},  \tilde{\Sigma}\bigr) \right\|_{\mathrm{op}}
\]
as the initial tuning parameter. We then update the tuning parameter as in \citet{zeleneev2026tractable}, that is, using \eqref{eq:update}, thereby maintaining consistency with their theoretical framework. The final choice of the tuning parameter is then $\phi=\hat{\varphi}$, to be used in \eqref{nuc problem} to obtain the nuclear-norm
penalized estimator 
$(\hat{\beta}_{\text{nuc}}, \hat{\Sigma}_{\text{nuc}})$.
Additional computational details can be found in Sections 4.1 and 4.3 of \citet{zeleneev2026tractable}.

Next, we apply a gradient descent algorithm initialized at the nuclear-norm penalized estimator. Specifically, we first recover the factor structure by computing the singular value decomposition of $\hat{\Sigma}_{\text{nuc}}$, which yields the estimators $\hat{\alpha}_{\text{nuc}}$ and $\hat{\gamma}_{\text{nuc}}$. The gradient descent iterations are then carried out starting from $(\hat{\beta}_{\text{nuc}}, \hat{\alpha}_{\text{nuc}}, \hat{\gamma}_{\text{nuc}})$.
The iterative updates are given by
\begin{align}
\beta^{(k+1)} &= \beta^{(k)} - S_\beta \, \partial_\beta \mathcal{L}(\beta^{(k)}, \alpha^{(k)}, \gamma^{(k)}), \notag \\
\alpha^{(k+1)} &= \alpha^{(k)} - S_\alpha \, \partial_\alpha \mathcal{L}(\beta^{(k)}, \alpha^{(k)}, \gamma^{(k)}), \label{eq:gradient-update} \\
\gamma^{(k+1)} &= \gamma^{(k)} - S_\gamma \, \partial_\gamma \mathcal{L}(\beta^{(k)}, \alpha^{(k)}, \gamma^{(k)}), \notag
\end{align}
where $S_\beta$, $S_\alpha$, and $S_\gamma$ denote the step sizes for the corresponding parameter blocks;
see \citet{zeleneev2026tractable} for details about the step sizes. We repeat this algorithm until convergence of all three parameters. Since the resulting two-step estimator is asymptotically equivalent to the MLE, we denote it simply by $\hat\theta$.\footnote{Theoretical guarantees for the convergence of this algorithm are provided in Theorem 5 of \citet{zeleneev2026tractable}. We compute $\hat{\theta}$ using the R package \texttt{NNRPanel} developed by \citet{zeleneev2026tractable}. For more details, see \url{https://github.com/wszhang-econ/NNRPanel}.}

In our theoretical and simulation results, we treat the number of factors, $d_f$, as known and fixed. In applications, however, the number of factors can be estimated using the eigenvalue-ratio test \citep{ahn2013eigenvalue}; see also
\citet{ChenFVWeidner2021}, \citet{GaoLiuPengYan2023}, and \citet{zeleneev2026tractable}.


The following algorithm summarizes our bootstrap inference procedure for $\beta_0$.

\begin{algorithm}[H]
\caption{Bootstrap inference for $\beta_0$}
\label{alg:bootstrap}
\begin{algorithmic}[1]
\State Compute the two-step MLE $\hat{\theta}$ based on the original sample.
\State Substitute the estimated $\hat{\theta}$ into model~\eqref{eq:loglikelihood} to generate $B$ parametric bootstrap samples. 
\For{$b = 1, \ldots, B$}
    \State Compute the two-step MLE $\hat{\theta}^{(b)}$ based on the $b$-th parametric bootstrap sample.
\EndFor

\State Compute the bias-corrected estimator 
\vspace{-0.5cm}
\[
\hat{\beta}_{\text{BC-Mean}} = 2\hat{\beta} - \operatorname{mean}\{\hat{\beta}^{(1)},\ldots,\hat{\beta}^{(B)}\}, \quad
\hat{\beta}_{\text{BC-Median}} = 2\hat{\beta} - \operatorname{median}\{\hat{\beta}^{(1)},\ldots,\hat{\beta}^{(B)}\}.
\]

\State \textbf{(Standard bootstrap CI)} For a scalar $\beta_0$, given a confidence level $1-\alpha$, compute
$Q^*_{\hat\beta^*}(\alpha/2)$
and
$Q^*_{\hat\beta^*}(1-\alpha/2)$,
the quantiles of $\{\hat{\beta}^{(1)},\ldots,\hat{\beta}^{(B)}\}$. The confidence interval is
\[
[ 2\hat{\beta} - Q^*_{\hat\beta^*}(1-\alpha/2),\;
  2\hat{\beta} - Q^*_{\hat\beta^*}(\alpha/2)].
\]

\State \textbf{(Transformed bootstrap CI)} Let $\varphi(\cdot)$ be a strictly increasing and continuously differentiable transformation. Compute the transformed bootstrap estimates $\{\varphi(\hat{\beta}^{(1)}), \ldots, \varphi(\hat{\beta}^{(B)})\}$ and the corresponding empirical quantile function $Q^*_{\varphi(\hat{\beta}^*)}$. The transformed bootstrap confidence interval is
\[
\left[
\varphi^{-1}\!\left(2\varphi(\hat{\beta}) - Q^*_{\varphi(\hat{\beta}^*)}(1-\alpha/2)\right), \;
\varphi^{-1}\!\left(2\varphi(\hat{\beta}) - Q^*_{\varphi(\hat{\beta}^*)}(\alpha/2)\right)
\right].
\]

\end{algorithmic}
\end{algorithm}

\section{Simulations}

Following \citet{ChenFVWeidner2021}, the two bias-corrected methods can also be applied to the two-step MLE estimator in \citet{zeleneev2026tractable}. One is a plug-in analytical correction based on the explicit form of the first-order bias, and another one is known as the split-panel jackknife method based on the idea of \citet{DhaeneJochmans2015}. We compare the finite-sample performance of our bootstrap bias correction with the two methods.

We focus on the following static model with interactive fixed effects:
\begin{equation} \label{MC}
D_{it} = \mathbf{1}\!\left\{ X_{it}^{\prime}\beta_0 + \alpha_{i0}^{\prime}\gamma_{t0} - \varepsilon_{it} > 0 \right\}, 
\quad i = 1,\dots,N,\; t = 1,\dots,T,
\end{equation}
where $\mathbf{1}\{\cdot\}$ is an indicator function and $\varepsilon_{it}$ follows a standard logistic distribution or a standard normal distribution, corresponding to the logit and the probit model. We let $N=30$, and consider $T \in \{20,30,40\}$. There is only $K=1$ covariate, and the number of factors $d_f = 2$, and the true coefficient is $\beta_0 = 0.5$. For all estimators, we fix $d_f = 2$. The following scenarios are included:

\paragraph{Scenario 1 (Covariates independent of factors):} For each $i$ and $t$, draw factor loadings and factors $\alpha_i,\gamma_t \sim \mathcal{N}(0,I_{d_f})$ and independent regressor factors $\alpha_i^{(x)},\gamma_t^{(x)} \sim \mathcal{N}(0,I_{d_f})$ where $I_{d_f}$ is the identity matrix of dimension $d_f$, and generate covariates as $X_{itk}=\alpha_i^{(x)\prime}\gamma_t^{(x)}+\nu_{itk}$ with $\nu_{itk}\sim\mathcal{N}(0,1)$. 

\paragraph{Scenario 2 (Covariates correlated with factors):} For each $i$ and $t$, draw factor loadings and factors $\alpha_i,\gamma_t \sim \mathcal{N}(0,I_{d_f})$, and generate covariates as $X_{itk}=0.3 \times \alpha_i^{\prime}\gamma_t+\nu_{itk}$ with $\nu_{itk}\sim\mathcal{N}(0,1)$.

\paragraph{Scenario 3 (Outliers in factor loadings):} For each $i$ and $t$, draw factor loadings and factors $\alpha_i,\gamma_t \sim \mathcal{N}(0,I_{d_f})$, fix a random $10\%$ of $\alpha_i = 3$ as outliers, draw independent regressor factors $\alpha_i^{(x)},\gamma_t^{(x)} \sim \mathcal{N}(0,I_{d_f})$, and generate covariates as $X_{itk}=\alpha_i^{(x)\prime}\gamma_t^{(x)}+\nu_{itk}$ with $\nu_{itk}\sim\mathcal{N}(0,1)$.

\paragraph{Scenario 4 (Sparse spikes in factors):} For each $i$ and $t$, draw factor loadings and factors $\alpha_i,\gamma_t \sim \mathcal{N}(0,I_{d_f})$, multiply a random 5\% of $\gamma_t$ by 10 to create sparse spikes, draw independent regressor factors $\alpha_i^{(x)},\gamma_t^{(x)} \sim \mathcal{N}(0,I_{d_f})$, and generate covariates as $X_{itk}=\alpha_i^{(x)\prime}\gamma_t^{(x)}+\nu_{itk}$ with $\nu_{itk}\sim\mathcal{N}(0,1)$.

We repeat the simulation $1000$ times and use $399$ bootstrap samples in each for the bootstrap method. Tables \ref{sim res logit} and \ref{sim res probit} report the relative bias (defined as the bias divided by the true parameter value), standard deviation, $95\%$ coverage rate, and the rejection probabilities under the null hypotheses: $H_0: \beta_0 = 0.7$ for all scenarios. The simulation results are for the logit and probit models, respectively. We compare the uncorrected maximum likelihood estimator, the split panel jackknife, the analytical bias correction, and the bootstrap method, where Boot-Mean uses the mean of bootstrap estimates as the corrected term, while Boot-Median uses the median. 

In most scenarios, the bootstrap-based methods deliver the most effective bias correction, exhibiting uniformly smaller biases than competing approaches. In Scenario 2 of both the logit and probit models, only the bootstrap methods achieve coverage probabilities close to the nominal $95\%$ level, whereas both the split-panel jackknife and the analytical bias correction perform less well in terms of bias reduction and coverage accuracy. As the sample size increases, the gap between the analytical bias correction and the bootstrap methods gradually narrows. The split-panel jackknife also improves with sample size but continues to require relatively large samples to achieve comparable bias reduction. Overall, the results suggest that bootstrap-based bias correction is particularly well suited to relatively small panel data settings, while remaining competitive with analytical bias correction in large samples. Moreover, the bootstrap approach performs well across a range of scenarios, especially when covariates are correlated with factors and loadings.

As discussed above, we find that the bootstrap confidence intervals tend to be overly conservative in finite samples, leading to coverage rates exceeding the nominal confidence levels. This pattern is particularly pronounced when the time dimension is small (e.g., \(T=20\)), as illustrated in Tables~\ref{sim res logit}--\ref{sim res probit}. In both models, the coverage rates of the parametric bootstrap confidence intervals are uniformly above $95\%$ across all scenarios, often substantially so, providing clear evidence of systematic over-coverage.

To address this issue, we consider three monotone transformations to reduce skewness in the bootstrap distribution: the log transformation, the Box--Cox transformation \citep{box1964analysis}, and the Yeo--Johnson transformation \citep{yeo2000new}. 
The log transformation provides a simple adjustment for strictly positive parameters and is particularly effective in mitigating right skewness by compressing the upper tail of the distribution. 
The Box--Cox transformation introduces a tuning parameter that controls the strength and direction of the transformation; in our implementation, we select this parameter in a data-driven manner by minimizing the skewness of the transformed bootstrap estimates.\footnote{
Specifically, we set $\lambda$ equal to
\[
\hat{\lambda}
=
\arg\min_{\lambda \in [-2,2]}
\left|
\frac{
\frac{1}{B}\sum_{b=1}^{B}
\left(\varphi_{\lambda}(\hat{\beta}^{*(b)})-\bar{\varphi}_{\lambda}\right)^3
}{
\left[
\frac{1}{B}\sum_{b=1}^{B}
\left(\varphi_{\lambda}(\hat{\beta}^{*(b)})-\bar{\varphi}_{\lambda}\right)^2
\right]^{3/2}
}
\right|,
\]
where $\hat{\beta}^{*(b)}$, $b=1,\ldots,B$, are the bootstrap estimates, $\varphi_{\lambda}(\cdot)$ is the Box--Cox transformation,  and $\bar{\varphi}_{\lambda} = \frac{1}{B}\sum_{b=1}^{B} \varphi_{\lambda}(\hat{\beta}^{*(b)})$. 
}
The Yeo--Johnson transformation extends this approach to accommodate both positive and negative values, while retaining similar flexibility in adjusting skewness.

The results in Tables~\ref{tab:logit_coverage_all}--\ref{tab:probit_coverage_all} show that, relative to the parametric bootstrap, the transformation-based parametric bootstrap substantially reduces the confidence interval length while bringing the coverage rate closer to the nominal level, particularly in Scenarios 1, 3, and 4, reflecting more efficient inference. This improvement is driven by their ability to mitigate skewness, resulting in more balanced lower and upper miss rates.

An exception is Scenario 2, where all methods perform relatively poorly. This is mainly due to the limited effectiveness of the bootstrap-based bias correction, which leads to a discrepancy between the bootstrap and true finite-sample distributions of the estimator. As a result, the bootstrap fails to accurately approximate the sampling distribution, reducing the effectiveness of the skewness correction and leading to distortions in coverage. 
By contrast, the untransformed parametric bootstrap, which is typically conservative, exhibits relatively better coverage in this case.

The log transformation often yields the shortest intervals, as it compresses right-skewed distributions. Since the MLE is predominantly right-skewed in our simulations, this effect is particularly pronounced. However, its one-sided nature makes it less robust in more complex settings. By contrast, the Yeo-Johnson transformation provides the most reliable overall performance, offering a favorable balance between coverage accuracy and interval length, especially in more asymmetric scenarios.

Overall, these results demonstrate that appropriate monotone transformations can substantially improve the finite-sample performance of bootstrap confidence intervals by reducing conservativeness and yielding tighter, more informative inference.


\begin{table}[H]
\centering
\caption{Simulation results for the logit model, $N=30$}
\label{sim res logit}

\medskip

\scalebox{0.75}{%
\begin{tabular}{llrccc|llrccc}
\hline
\multicolumn{6}{c|}{\textit{Scenario 1: covariates independent of factors}} &
\multicolumn{6}{c}{\textit{Scenario 2: covariates correlated with factors}} \\
$T$ & Method & Bias\,\, & SD & Coverage  & Rejection &
$T$ & Method & Bias\,\, & SD & Coverage  & Rejection \\
\hline
20 & MLE         & 0.326 & 0.122 & 0.615 & 0.148 & 20 & MLE         & 0.668 & 0.189 & 0.371 & 0.208 \\
   & SplitPJ     & $-0.467$ & 0.193 & 0.414 & 0.895 &    & SplitPJ     & $-0.522$ & 0.339 & 0.494 & 0.732 \\
   & Analytical  & 0.155 & 0.103 & 0.892 & 0.308 &    & Analytical  & 0.499 & 0.168 & 0.589 & 0.094 \\
   & Boot-Mean   & $-0.081$ & 0.083 & 0.985 & 0.483 &    & Boot-Mean   & 0.171 & 0.141 & 0.995 & 0.056 \\
   & Boot-Median & $-0.061$ & 0.085 & 0.985 & 0.483 &    & Boot-Median & 0.193 & 0.143 & 0.995 & 0.056 \\
\hline
30 & MLE         & 0.218 & 0.085 & 0.675 & 0.309 & 30 & MLE         & 0.481 & 0.135 & 0.421 & 0.132 \\
   & SplitPJ     & $-0.228$ & 0.121 & 0.595 & 0.921 &    & SplitPJ     & $-0.264$ & 0.214 & 0.575 & 0.724 \\
   & Analytical  & 0.095 & 0.074 & 0.916 & 0.590 &    & Analytical  & 0.361 & 0.123 & 0.595 & 0.099 \\
   & Boot-Mean   & $-0.040$ & 0.066 & 0.979 & 0.756 &    & Boot-Mean   & 0.177 & 0.111 & 0.966 & 0.136 \\
   & Boot-Median & $-0.031$ & 0.066 & 0.979 & 0.756 &    & Boot-Median & 0.186 & 0.111 & 0.966 & 0.136 \\
\hline
40 & MLE         & 0.177 & 0.066 & 0.679 & 0.472 & 40 & MLE         & 0.355 & 0.107 & 0.486 & 0.114 \\
   & SplitPJ     & $-0.150$ & 0.087 & 0.668 & 0.952 &    & SplitPJ     & $-0.229$ & 0.164 & 0.594 & 0.822 \\
   & Analytical  & 0.071 & 0.058 & 0.936 & 0.765 &    & Analytical  & 0.259 & 0.100 & 0.670 & 0.164 \\
   & Boot-Mean   & $-0.024$ & 0.053 & 0.977 & 0.882 &    & Boot-Mean   & 0.128 & 0.092 & 0.952 & 0.265 \\
   & Boot-Median & $-0.019$ & 0.054 & 0.977 & 0.882 &    & Boot-Median & 0.134 & 0.092 & 0.952 & 0.265 \\
\hline
\multicolumn{6}{c|}{\textit{Scenario 3: outliers in factor loadings}} &
\multicolumn{6}{c}{\textit{Scenario 4: sparse spikes in factors}} \\
$T$ & Method & Bias\,\, & SD & Coverage  & Rejection &
$T$ & Method & Bias\,\, & SD & Coverage  & Rejection \\
\hline
20 & MLE         & 0.314 & 0.119 & 0.661 & 0.129 & 20 & MLE         & 0.315 & 0.122 & 0.636 & 0.160 \\
   & SplitPJ     & $-0.458$ & 0.197 & 0.432 & 0.881 &    & SplitPJ     & $-0.447$ & 0.192 & 0.423 & 0.895 \\
   & Analytical  & 0.141 & 0.100 & 0.921 & 0.289 &    & Analytical  & 0.145 & 0.102 & 0.907 & 0.307 \\
   & Boot-Mean   & $-0.097$ & 0.080 & 0.979 & 0.473 &    & Boot-Mean   & $-0.086$ & 0.082 & 0.988 & 0.486 \\
   & Boot-Median & $-0.077$ & 0.081 & 0.979 & 0.473 &    & Boot-Median & $-0.066$ & 0.083 & 0.988 & 0.486 \\
\hline
30 & MLE         & 0.213 & 0.085 & 0.701 & 0.305 & 30 & MLE         & 0.215 & 0.085 & 0.697 & 0.303 \\
   & SplitPJ     & $-0.242$ & 0.116 & 0.567 & 0.921 &    & SplitPJ     & $-0.241$ & 0.117 & 0.559 & 0.932 \\
   & Analytical  & 0.085 & 0.074 & 0.926 & 0.576 &    & Analytical  & 0.090 & 0.075 & 0.922 & 0.579 \\
   & Boot-Mean   & $-0.049$ & 0.065 & 0.981 & 0.745 &    & Boot-Mean   & $-0.042$ & 0.066 & 0.981 & 0.749 \\
   & Boot-Median & $-0.040$ & 0.066 & 0.981 & 0.745 &    & Boot-Median & $-0.033$ & 0.066 & 0.981 & 0.749 \\
\hline
40 & MLE         & 0.174 & 0.068 & 0.711 & 0.460 & 40 & MLE         & 0.166 & 0.069 & 0.730 & 0.512 \\
   & SplitPJ     & $-0.152$ & 0.083 & 0.696 & 0.965 &    & SplitPJ     & $-0.162$ & 0.087 & 0.644 & 0.964 \\
   & Analytical  & 0.063 & 0.060 & 0.929 & 0.788 &    & Analytical  & 0.058 & 0.061 & 0.937 & 0.798 \\
   & Boot-Mean   & $-0.033$ & 0.054 & 0.974 & 0.892 &    & Boot-Mean   & $-0.034$ & 0.055 & 0.979 & 0.889 \\
   & Boot-Median & $-0.027$ & 0.055 & 0.974 & 0.892 &    & Boot-Median & $-0.028$ & 0.056 & 0.979 & 0.889 \\
\hline
\end{tabular}%
}

\medskip

\begin{minipage}{0.95\textwidth}
\footnotesize
\textit{Notes:} model given in \eqref{MC} with one covariate and $\beta_0=0.5$; $1000$ Monte Carlo replications; $399$ bootstrap replications; coverage rates of nominal $95\%$ confidence intervals; rejection probabilities of $H_0:\beta_0=0.7$ at nominal $5\%$ level.
\end{minipage}

\end{table}

\begin{table}[H]
\centering
\caption{Simulation results for the probit model, $N=30$}
\label{sim res probit}

\medskip

\scalebox{0.75}{%
\begin{tabular}{llrccc|llrccc}
\hline
\multicolumn{6}{c|}{\textit{Scenario 1: covariates independent of factors}} &
\multicolumn{6}{c}{\textit{Scenario 2: covariates correlated with factors}} \\
$T$ & Method & Bias\,\, & SD & Coverage  & Rejection &
$T$ & Method & Bias\,\, & SD & Coverage  & Rejection \\
\hline
20 & MLE         & 0.374 & 0.108 & 0.409 & 0.139 & 20 & MLE         & 0.472 & 0.144 & 0.457 & 0.129 \\
   & SplitPJ     & $-0.617$ & 0.242 & 0.249 & 0.959 &    & SplitPJ     & $-0.651$ & 0.259 & 0.341 & 0.911 \\
   & Analytical  & 0.152 & 0.082 & 0.911 & 0.384 &    & Analytical  & 0.283 & 0.124 & 0.744 & 0.130 \\
   & Boot-Mean   & $-0.184$ & 0.058 & 0.981 & 0.719 &    & Boot-Mean   & $-0.047$ & 0.090 & 0.995 & 0.308 \\
   & Boot-Median & $-0.148$ & 0.058 & 0.981 & 0.719 &    & Boot-Median & $-0.016$ & 0.093 & 0.995 & 0.308 \\
\hline
30 & MLE         & 0.264 & 0.070 & 0.418 & 0.294 & 30 & MLE         & 0.315 & 0.092 & 0.501 & 0.144 \\
   & SplitPJ     & $-0.264$ & 0.100 & 0.436 & 0.989 &    & SplitPJ     & $-0.299$ & 0.132 & 0.515 & 0.926 \\
   & Analytical  & 0.100 & 0.057 & 0.902 & 0.704 &    & Analytical  & 0.173 & 0.082 & 0.818 & 0.347 \\
   & Boot-Mean   & $-0.078$ & 0.045 & 0.982 & 0.909 &    & Boot-Mean   & 0.003 & 0.069 & 0.995 & 0.553 \\
   & Boot-Median & $-0.064$ & 0.046 & 0.982 & 0.909 &    & Boot-Median & 0.015 & 0.070 & 0.995 & 0.553 \\
\hline
40 & MLE         & 0.207 & 0.053 & 0.439 & 0.549 & 40 & MLE         & 0.253 & 0.074 & 0.503 & 0.258 \\
   & SplitPJ     & $-0.169$ & 0.071 & 0.553 & 0.998 &    & SplitPJ     & $-0.194$ & 0.094 & 0.614 & 0.961 \\
   & Analytical  & 0.069 & 0.045 & 0.927 & 0.910 &    & Analytical  & 0.128 & 0.066 & 0.850 & 0.559 \\
   & Boot-Mean   & $-0.053$ & 0.039 & 0.979 & 0.984 &    & Boot-Mean   & 0.013 & 0.059 & 0.992 & 0.763 \\
   & Boot-Median & $-0.046$ & 0.039 & 0.979 & 0.984 &    & Boot-Median & 0.020 & 0.060 & 0.992 & 0.763 \\
\hline
\multicolumn{6}{c|}{\textit{Scenario 3: outliers in factor loadings}} &
\multicolumn{6}{c}{\textit{Scenario 4: sparse spikes in factors}} \\
$T$ & Method & Bias\,\, & SD & Coverage  & Rejection &
$T$ & Method & Bias\,\, & SD & Coverage  & Rejection \\
\hline
20 & MLE         & 0.381 & 0.109 & 0.439 & 0.110 & 20 & MLE         & 0.354 & 0.103 & 0.479 & 0.124 \\
   & SplitPJ     & $-0.643$ & 0.259 & 0.255 & 0.961 &    & SplitPJ     & $-0.598$ & 0.220 & 0.276 & 0.966 \\
   & Analytical  & 0.149 & 0.085 & 0.917 & 0.376 &    & Analytical  & 0.134 & 0.081 & 0.938 & 0.419 \\
   & Boot-Mean   & $-0.207$ & 0.061 & 0.984 & 0.698 &    & Boot-Mean   & $-0.197$ & 0.057 & 0.983 & 0.735 \\
   & Boot-Median & $-0.165$ & 0.060 & 0.984 & 0.698 &    & Boot-Median & $-0.159$ & 0.057 & 0.983 & 0.735 \\
\hline
30 & MLE         & 0.266 & 0.075 & 0.458 & 0.276 & 30 & MLE         & 0.249 & 0.072 & 0.465 & 0.323 \\
   & SplitPJ     & $-0.272$ & 0.103 & 0.446 & 0.981 &    & SplitPJ     & $-0.269$ & 0.101 & 0.440 & 0.985 \\
   & Analytical  & 0.098 & 0.061 & 0.912 & 0.679 &    & Analytical  & 0.087 & 0.058 & 0.929 & 0.726 \\
   & Boot-Mean   & $-0.089$ & 0.048 & 0.980 & 0.889 &    & Boot-Mean   & $-0.088$ & 0.047 & 0.967 & 0.903 \\
   & Boot-Median & $-0.074$ & 0.048 & 0.980 & 0.889 &    & Boot-Median & $-0.074$ & 0.048 & 0.967 & 0.903 \\
\hline
40 & MLE         & 0.215 & 0.057 & 0.431 & 0.480 & 40 & MLE         & 0.213 & 0.057 & 0.445 & 0.499 \\
   & SplitPJ     & $-0.171$ & 0.075 & 0.555 & 0.994 &    & SplitPJ     & $-0.158$ & 0.067 & 0.611 & 0.997 \\
   & Analytical  & 0.069 & 0.047 & 0.943 & 0.889 &    & Analytical  & 0.070 & 0.048 & 0.930 & 0.891 \\
   & Boot-Mean   & $-0.058$ & 0.040 & 0.975 & 0.978 &    & Boot-Mean   & $-0.053$ & 0.041 & 0.981 & 0.971 \\
   & Boot-Median & $-0.049$ & 0.040 & 0.975 & 0.978 &    & Boot-Median & $-0.045$ & 0.041 & 0.981 & 0.971 \\
\hline
\end{tabular}%
}

\medskip

\begin{minipage}{0.95\textwidth}
\footnotesize
\textit{Notes:} model given in \eqref{MC} with one covariate and $\beta_0=0.5$; $1000$ Monte Carlo replications; $399$ bootstrap replications; coverage rates of nominal $95\%$ confidence intervals; rejection probabilities of $H_0:\beta_0=0.7$ at nominal $5\%$ level.
\end{minipage}

\end{table}

\begin{table}[H]
\centering
\caption{Comparison of confidence intervals, logit model, $N=30$}
\label{tab:logit_coverage_all}

\medskip

\scalebox{0.75}{%
\begin{tabular}{llcccc|llcccc}
\hline
\multicolumn{6}{c|}{\textit{Scenario 1: covariates independent of factors}} &
\multicolumn{6}{c}{\textit{Scenario 2: covariates correlated with factors}} \\
$T$ & Method & Coverage & Length & LMR & UMR &
$T$ & Method & Coverage & Length & LMR & UMR \\
\hline
20 & Boot        & 0.985 & 0.558 & 0.000 & 0.015 & 20 & Boot        & 0.995 & 0.783 & 0.002 & 0.003 \\
   & Log         & 0.941 & 0.342 & 0.048 & 0.011 &    & Log         & 0.636 & 0.508 & 0.364 & 0.000 \\
   & Box-Cox     & 0.933 & 0.369 & 0.054 & 0.013 &    & Box-Cox     & 0.796 & 0.578 & 0.204 & 0.000 \\
   & Yeo-Johnson & 0.960 & 0.384 & 0.027 & 0.013 &    & Yeo-Johnson & 0.826 & 0.594 & 0.172 & 0.002 \\
\hline
30 & Boot        & 0.979 & 0.366 & 0.000 & 0.021 & 30 & Boot        & 0.966 & 0.523 & 0.034 & 0.000 \\
   & Log         & 0.929 & 0.257 & 0.054 & 0.017 &    & Log         & 0.627 & 0.384 & 0.373 & 0.000 \\
   & Box-Cox     & 0.934 & 0.281 & 0.045 & 0.021 &    & Box-Cox     & 0.783 & 0.439 & 0.217 & 0.000 \\
   & Yeo-Johnson & 0.945 & 0.287 & 0.034 & 0.021 &    & Yeo-Johnson & 0.801 & 0.444 & 0.199 & 0.000 \\
\hline
40 & Boot        & 0.977 & 0.288 & 0.000 & 0.023 & 40 & Boot        & 0.952 & 0.412 & 0.045 & 0.003 \\
   & Log         & 0.952 & 0.215 & 0.031 & 0.017 &    & Log         & 0.692 & 0.319 & 0.308 & 0.000 \\
   & Box-Cox     & 0.949 & 0.235 & 0.031 & 0.020 &    & Box-Cox     & 0.835 & 0.364 & 0.163 & 0.002 \\
   & Yeo-Johnson & 0.952 & 0.238 & 0.028 & 0.020 &    & Yeo-Johnson & 0.843 & 0.366 & 0.155 & 0.002 \\
\hline
\multicolumn{6}{c|}{\textit{Scenario 3: outliers in factor loadings}} &
\multicolumn{6}{c}{\textit{Scenario 4: sparse spikes in factors}} \\
$T$ & Method & Coverage & Length & LMR & UMR &
$T$ & Method & Coverage & Length & LMR & UMR \\
\hline
20 & Boot        & 0.979 & 0.578 & 0.000 & 0.021 & 20 & Boot        & 0.988 & 0.566 & 0.000 & 0.012 \\
   & Log         & 0.957 & 0.354 & 0.032 & 0.011 &    & Log         & 0.950 & 0.349 & 0.043 & 0.007 \\
   & Box-Cox     & 0.942 & 0.382 & 0.042 & 0.016 &    & Box-Cox     & 0.950 & 0.375 & 0.038 & 0.012 \\
   & Yeo-Johnson & 0.956 & 0.399 & 0.028 & 0.016 &    & Yeo-Johnson & 0.967 & 0.391 & 0.021 & 0.012 \\
\hline
30 & Boot        & 0.981 & 0.378 & 0.000 & 0.019 & 30 & Boot        & 0.981 & 0.369 & 0.000 & 0.019 \\
   & Log         & 0.942 & 0.264 & 0.043 & 0.015 &    & Log         & 0.938 & 0.260 & 0.049 & 0.013 \\
   & Box-Cox     & 0.953 & 0.289 & 0.030 & 0.017 &    & Box-Cox     & 0.943 & 0.283 & 0.039 & 0.018 \\
   & Yeo-Johnson & 0.959 & 0.296 & 0.023 & 0.018 &    & Yeo-Johnson & 0.948 & 0.289 & 0.033 & 0.019 \\
\hline
40 & Boot        & 0.974 & 0.297 & 0.000 & 0.026 & 40 & Boot        & 0.979 & 0.292 & 0.001 & 0.020 \\
   & Log         & 0.938 & 0.220 & 0.038 & 0.024 &    & Log         & 0.945 & 0.218 & 0.040 & 0.015 \\
   & Box-Cox     & 0.952 & 0.240 & 0.023 & 0.025 &    & Box-Cox     & 0.957 & 0.238 & 0.025 & 0.018 \\
   & Yeo-Johnson & 0.955 & 0.243 & 0.020 & 0.025 &    & Yeo-Johnson & 0.960 & 0.241 & 0.021 & 0.019 \\
\hline
\end{tabular}%
}

\medskip

\begin{minipage}{0.95\textwidth}
\footnotesize
\textit{Notes:} model given in \eqref{MC} with one covariate and $\beta_0=0.5$; $1000$ Monte Carlo replications; $399$ bootstrap replications; coverage rates of nominal $95\%$ confidence intervals; Boot: parametric bootstrap; Log/Box-Cox/Yeo-Johnson: transformed parametric bootstrap; LMR: lower miss rate; UMR upper miss rate.
\end{minipage}

\end{table}

\begin{table}[H]
\centering
\caption{Comparison of confidence intervals, probit model, $N=30$}
\label{tab:probit_coverage_all}

\medskip

\scalebox{0.75}{%
\begin{tabular}{llcccc|llcccc}
\hline
\multicolumn{6}{c|}{\textit{Scenario 1: covariates independent of factors}} &
\multicolumn{6}{c}{\textit{Scenario 2: covariates correlated with factors}} \\
$T$ & Method & Coverage & Length & LMR & UMR &
$T$ & Method & Coverage & Length & LMR & UMR \\
\hline
20 & Boot        & 0.981 & 0.567 & 0.000 & 0.019 & 20 & Boot        & 0.995 & 0.704 & 0.004 & 0.001 \\
   & Log         & 0.929 & 0.349 & 0.060 & 0.011 &    & Log         & 0.757 & 0.463 & 0.243 & 0.000 \\
   & Box-Cox     & 0.904 & 0.375 & 0.080 & 0.016 &    & Box-Cox     & 0.858 & 0.518 & 0.142 & 0.000 \\
   & Yeo-Johnson & 0.934 & 0.392 & 0.052 & 0.014 &    & Yeo-Johnson & 0.887 & 0.533 & 0.110 & 0.003 \\
\hline
30 & Boot        & 0.982 & 0.377 & 0.000 & 0.018 & 30 & Boot        & 0.995 & 0.464 & 0.005 & 0.000 \\
   & Log         & 0.930 & 0.261 & 0.057 & 0.013 &    & Log         & 0.828 & 0.333 & 0.172 & 0.000 \\
   & Box-Cox     & 0.912 & 0.283 & 0.071 & 0.017 &    & Box-Cox     & 0.899 & 0.373 & 0.101 & 0.000 \\
   & Yeo-Johnson & 0.935 & 0.293 & 0.050 & 0.015 &    & Yeo-Johnson & 0.915 & 0.381 & 0.084 & 0.001 \\
\hline
40 & Boot        & 0.979 & 0.297 & 0.000 & 0.021 & 40 & Boot        & 0.992 & 0.365 & 0.008 & 0.000 \\
   & Log         & 0.934 & 0.217 & 0.047 & 0.019 &    & Log         & 0.857 & 0.279 & 0.143 & 0.000 \\
   & Box-Cox     & 0.920 & 0.234 & 0.060 & 0.020 &    & Box-Cox     & 0.913 & 0.311 & 0.087 & 0.000 \\
   & Yeo-Johnson & 0.937 & 0.242 & 0.043 & 0.020 &    & Yeo-Johnson & 0.926 & 0.317 & 0.073 & 0.001 \\
\hline
\multicolumn{6}{c|}{\textit{Scenario 3: outliers in factor loadings}} &
\multicolumn{6}{c}{\textit{Scenario 4: sparse spikes in factors}} \\
$T$ & Method & Coverage & Length & LMR & UMR &
$T$ & Method & Coverage & Length & LMR & UMR \\
\hline
20 & Boot        & 0.984 & 0.573 & 0.000 & 0.016 & 20 & Boot        & 0.983 & 0.548 & 0.000 & 0.017 \\
   & Log         & 0.942 & 0.355 & 0.047 & 0.011 &    & Log         & 0.944 & 0.338 & 0.044 & 0.012 \\
   & Box-Cox     & 0.919 & 0.382 & 0.067 & 0.014 &    & Box-Cox     & 0.918 & 0.364 & 0.070 & 0.012 \\
   & Yeo-Johnson & 0.944 & 0.397 & 0.040 & 0.016 &    & Yeo-Johnson & 0.944 & 0.378 & 0.043 & 0.013 \\
\hline
30 & Boot        & 0.980 & 0.378 & 0.000 & 0.020 & 30 & Boot        & 0.967 & 0.357 & 0.001 & 0.032 \\
   & Log         & 0.933 & 0.265 & 0.053 & 0.014 &    & Log         & 0.934 & 0.252 & 0.051 & 0.015 \\
   & Box-Cox     & 0.915 & 0.289 & 0.068 & 0.017 &    & Box-Cox     & 0.916 & 0.273 & 0.066 & 0.018 \\
   & Yeo-Johnson & 0.937 & 0.296 & 0.046 & 0.017 &    & Yeo-Johnson & 0.936 & 0.279 & 0.045 & 0.019 \\
\hline
40 & Boot        & 0.975 & 0.297 & 0.000 & 0.025 & 40 & Boot        & 0.981 & 0.292 & 0.000 & 0.019 \\
   & Log         & 0.929 & 0.220 & 0.048 & 0.023 &    & Log         & 0.932 & 0.218 & 0.050 & 0.018 \\
   & Box-Cox     & 0.913 & 0.240 & 0.061 & 0.026 &    & Box-Cox     & 0.918 & 0.238 & 0.058 & 0.024 \\
   & Yeo-Johnson & 0.932 & 0.243 & 0.043 & 0.025 &    & Yeo-Johnson & 0.936 & 0.241 & 0.040 & 0.024 \\
\hline
\end{tabular}%
}

\medskip

\begin{minipage}{0.95\textwidth}
\footnotesize
\textit{Notes:} model given in \eqref{MC} with one covariate and $\beta_0=0.5$; $1000$ Monte Carlo replications; $399$ bootstrap replications; coverage rates of nominal $95\%$ confidence intervals; Boot: parametric bootstrap; Log/Box-Cox/Yeo-Johnson: transformed parametric bootstrap; LMR: lower miss rate; UMR upper miss rate.
\end{minipage}

\end{table}

\section{Empirical application: technology spillovers in the presence of latent heterogeneity}
In this section, we revisit the dataset originally from \citet{bloom} and studied by \citet{burdaapplication}, who use a Bayesian panel probit model to estimate a patent equation using firm-level data on research and development ($R\&D$).\footnote{The dataset is publicly available at \url{http://qed.econ.queensu.ca/jae/datasets/burda002/}.} Their study explores a long-standing question in the $R\&D$ literature concerning the relationship between innovation, proxied by firms' patenting activity, and spillover effects arising from strategic interactions across firms. These interactions arise in two main spaces: technology space, defined by the similarity of firms' underlying technologies and associated with beneficial knowledge spillovers, and product market space, defined by the degree of competition in overlapping product markets and associated with business-stealing effects. Furthermore, as emphasized by \citet{burdaapplication}, the innovation and patenting behavior are also likely to be affected by unobserved heterogeneity at the firm and time level. 

In \citet{burdaapplication}, they report results from different models considering unobserved heterogeneity, such as the Bayesian panel probit model with two latent effects, the fixed probit model with time dummies, and the random effects probit model with time dummies. However, they do not consider an interactive fixed effect structure in the models. Hence, we employ our estimator to this empirical application.

The original dataset consists of an unbalanced panel of $729$ U.S. firms observed over the period $1981-2001$. To construct a balanced panel, we restrict attention to firms with complete observations throughout the sample period. In addition, firms that do not register any patents over the entire time span are excluded. The resulting sample comprises $328$ firms observed annually from $1981$ to $2001$. The dependent variable is a binary indicator equal to one if a firm files at least one patent in a given year. The key explanatory variables capture two types of spillover effects. Technological spillovers (SpillTech) are constructed using information on the distribution of patents across technological classes. Specifically, technological proximity between firms is measured using the Jaffe distance \citep{jaffe1986}, calculated as the uncentered correlation (cosine similarity) of their patent shares across different technology classes. The corresponding spillover variable is defined as the weighted sum of other firms' $R\&D$ stocks, where weights reflect technological proximity; see details in \citet{burdaapplication}. Product market spillovers (SpillSIC) are constructed analogously using the distribution of firm sales across industries. Product market proximity is measured as the uncentered correlation of firms' sales shares across industry classifications, and the associated spillover variable is defined as the proximity-weighted sum of other firms' $R\&D$ stocks. In addition, we include firm-level controls for the stock of $R\&D$ ($R\&D$ stock) and firm sales (Sales), both constructed from accounting data. All continuous independent variables are expressed in logarithms, lagged by one period, and all regressions include a dummy where lagged $R\&D$ stock is zero.

Prior to estimation, all continuous independent variables are standardized by subtracting their sample means and dividing by their respective standard deviations. We report the results of the following three estimators: the two-step maximum likelihood estimator proposed by \citet{zeleneev2026tractable} (IFE), a plug-in analytical bias correction (IFE-Analytical), and a median bootstrap-based bias correction (IFE-Bootstrap). We use the eigenvalue-ratio test \citep{ahn2013eigenvalue}, slightly modified as in \citet{zeleneev2026tractable}, to determine the number of factors, which is $1$ in this empirical application. In the bootstrap procedure, we draw $399$ bootstrap samples and treat the number of factors as known, which is $1$.

Economic theory provides clear predictions regarding the effects of $R\&D$ spillovers on patenting activity. In the absence of endogenous patenting behavior, product market spillovers, which captures market rivalry, are expected to have no effect on innovation outcomes, implying a coefficient close to zero. In contrast, technological spillovers are expected to have a positive effect, as knowledge diffusion enhances firms' innovative capacity. These predictions serve as a theoretical benchmark for evaluating the empirical results presented in Table \ref{app}. 

\begin{table}[H]
  \centering
  \caption{Parameter Estimates and Standard Errors}
  \medskip
  \label{app}
  \small
  \begin{tabular}{llll}
    \toprule
    Covariate & IFE & IFE-Analytical & IFE-Bootstrap \\
    \midrule
    $\log(\text{SpillTech})$ & 0.139***    & 0.222*** & 0.132** \\
                             & (0.044)     & (0.044)  & (0.047) \\
    \addlinespace
    $\log(\text{SpillSIC})$  & 0.065       & 0.060    & 0.045 \\
                             & (0.039)     & (0.039)  & (0.044) \\
    \addlinespace
    $\log(\text{R\&D Stock})$ & 0.269***   & 0.255*** & 0.287*** \\
                             & (0.074)     & (0.074)  & (0.082) \\
    \addlinespace
    $\log(\text{Sales})$     & 0.440***    & 0.420*** & 0.378*** \\
                             & (0.044)     & (0.044)  & (0.049) \\
    \bottomrule
    \multicolumn{4}{p{0.62\linewidth}}{\textit{Notes:} Standard errors are reported in parentheses. Standard errors for the IFE-Bootstrap estimator are obtained as the standard deviation across $399$ bootstrap replications. Statistical significance is denoted by ** and *** at the 5\% and 1\% levels, respectively. All explanatory variables are lagged by one period. All specifications include a dummy variable indicating observations with zero lagged R\&D stock. Additionally, we apply the Yeo-Johnson transformation to adjust the significance levels of the bootstrap estimates, where the significance level of $\log(\text{SpillTech})$ changes to $1\%$ while all other covariates remain unchanged.}
  \end{tabular}
\end{table}

Consistent with theory, technological spillovers (SpillTech) are positive and statistically significant across all estimation methods, indicating that knowledge diffusion among technologically similar firms increases patenting activity. The estimated magnitude, however, varies with the estimation method: the analytical bias correction yields a noticeably larger coefficient, whereas the bootstrap correction brings it back close to the baseline IFE estimate. This pattern suggests that the analytical correction may over-adjust upward, while the bootstrap provides a more conservative estimate. In contrast, product market spillovers are small and statistically insignificant across all estimators, and this finding is unaffected by either bias-correction approach. This reinforces the interpretation that competitive pressures in product markets do not exert a direct effect on firms' patenting decisions in this setting. The coefficients on firm fundamentals ($R\&D$ stock and sales) from all estimation methods are large and positive, where the bootstrap correction suggests a smaller estimate of sales.

Finally, although \citet{burdaapplication} apply the proposed Bayesian panel probit model on the full unbalanced panel, their application results are broadly consistent with ours, estimated on a sub-balanced panel, in both sign and statistical significance.

Overall, the bootstrap results confirm the robustness of the main conclusions: technological spillovers matter for innovation, product market spillovers do not, and the qualitative inference remains stable across different estimation methods.

\section{Conclusion}

In this paper, we propose a novel bias-correction method for nonlinear panel data models with interactive fixed effects, based on a parametric bootstrap approach. The presence of the incidental parameters problem renders the maximum likelihood estimator biased. We show that, under suitable regularity conditions, the parametric bootstrap replicates the asymptotic distribution of the maximum likelihood estimator. As a result, the bootstrap can be used to perform both bias correction and statistical inference in a unified framework.

Compared with existing bias-correction methods, such as the split-panel jackknife and plug-in analytical approaches, the proposed method does not require an explicit derivation of the bias term and exhibits superior finite-sample performance in panels with relatively small dimensions, as evidenced by our simulation results. In addition, we apply three monotone transformations to reduce skewness in the bootstrap distribution, thereby improving the previously conservative coverage rates.

Finally, we illustrate the empirical relevance of the proposed method through an application to firm-level innovation behavior. The results provide support for key theoretical predictions and highlight the practical usefulness of the method in applied research.

\section*{Supplementary material} \label{supp m}
     
    \paragraph{Replication files:} The replication code for both the simulations and the empirical application is available at \url{https://github.com/Wei-M-Wei/Factor-Bootstrap-replication}. 

\bibliographystyle{apalike}
\bibliography{reference2} 

\appendix

\section*{Appendix}
This Appendix consists of two parts. In Appendix~\ref{Appendix A}, we provide rigorous proofs of the theorems presented in the main text, building on the uniform asymptotic expansion established in the online appendix. In the remaining online Appendices \ref{Appendix B}, and \ref{Appendix C}, we derive asymptotic expansions for \(\sqrt{NT}(\hat{\beta} - \beta)\) and \(\sqrt{NT}\big(\varDelta(\hat{\theta}) - \varDelta(\theta)\big)\) that hold uniformly over $(\beta, \phi)$ in a neighborhood of $(\beta_0, \phi_0)$.

Appendix \ref{Appendix B} introduces the notation and presents several auxiliary lemmas that are needed to apply the framework developed in Appendix~\ref{Appendix C} to the interactive fixed effects model. In Appendix~\ref{Appendix C}, we generalize the expansion for M-estimators with high-dimensional incidental parameters developed by \citet{FVWeidner2016} to a broader setting. A key technical requirement in our framework is that the moment conditions hold uniformly for all $(\beta, \phi)$ in a neighborhood of $(\beta_0, \phi_0)$. In Lemma~\ref{lemma.D2}, focusing on panel data applications, we establish a set of sufficient conditions under which the uniform moment restrictions in Assumption~\ref{assu:uniform regularity} are satisfied. In Appendix~\ref{Appendix B}, we specialize the general framework developed in Appendix~\ref{Appendix C} to panel data models. A central complication in this context is that the log-likelihood function is not globally strictly concave in the presence of interactive fixed effects, thereby precluding a direct application of standard arguments. To circumvent this issue, we adopt the strategy of \citet{ChenFVWeidner2021} and proceed via a sequence of intermediate results.

Specifically, Lemma~\ref{lemma:c1} establishes a preliminary convergence rate for $\hat{\beta}$, implying that the estimator is contained, with high probability, in a shrinking neighborhood of the true parameter. Conditional on this localization result, Lemma~\ref{lemma:c2} establishes the local strict concavity of the log-likelihood function within this neighborhood, thereby recovering the curvature conditions required for higher-order analysis. Lemma~\ref{lemma:c3} further verifies that the score function and its higher-order derivatives satisfy the uniform regularity conditions imposed in Appendix~\ref{Appendix C}.

Taken together, these results justify the application of the general uniform expansion to the panel data setting considered here. Finally, Theorem~\ref{thm:C1:} and Theorem~\ref{Thm:C2} provide explicit expressions for all terms in the expansions of $\hat{\beta}$ and $\varDelta(\hat{\theta})$, which form the basis for the asymptotic and bootstrap results established in Appendix~\ref{Appendix A}.

\section{Proof of Bootstrap Consistency}\label{Appendix A}

\begin{myproof}[Proof of Theorem \ref{theorem3.1}]
To apply the uniform expansion results in Appendix~\ref{Appendix C}, 
the objective function~\eqref{eq:loglikelihood} needs to be strictly concave for $\theta$. 
However, this property does not hold globally due to the presence of interactive fixed effects. 
Nevertheless, Lemma~\ref{lemma:c3} shows that the likelihood function~\eqref{eq:loglikelihood} 
is strictly concave in a shrinking neighborhood of $\theta$, 
uniformly for every $\theta \in \Theta_0$. 
Given the initial uniform convergence result of $\hat{\theta}$ in Lemma~\ref{lemma:c1}, 
which ensures that the estimate $\hat{\theta}$ lies in this local neighborhood with probability 1, 
we can therefore still apply the uniform expansion in Appendix~\ref{Appendix C}. 

Unlike the one-way fixed effects model studied in \citet{HigginsJochmans2024}, where the Hessian matrix is diagonal, in models with interactive fixed effects, an additional difficulty arises because there is no closed-form expression for the inverse of the Hessian. Nevertheless, Lemma~\ref{lemma:c2} shows that the Hessian can be well-approximated by a diagonal matrix, and that the approximation error is uniformly asymptotically negligible for all $\theta \in \Theta_0$.

Then, following the same argument as in Theorem~4 of \citet{ChenFVWeidner2021}, 
we can apply the uniform expansion in Appendix~\ref{Appendix C} and obtain the following uniform asymptotic expansion for the MLE $\hat{\beta}$:
\[
\sqrt{NT}(\hat{\beta}-\beta_{0})=
\overline{W}_{\infty}^{-1}\Bigg(
\partial_{\beta}\mathcal{L}^{*} 
+ B^{(1)}+B^{(2)}
\Bigg) + o_{p}^{u}(1),
\]
where $B^{(1)}$ and $B^{(2)}$ are asymptotic bias terms, defined in \eqref{eq:biasterms}, and the remainder term is uniformly asymptotically negligible.
The matrix 
\[
\overline{W}_\infty := \lim_{N,T\to \infty} \frac{1}{NT} \sum_{i=1}^{N} \sum_{t=1}^{T} 
\mathbb{E}_{\theta} \Big( \partial_{z^2} \ell_{it} \, \widetilde{X}_{it} \widetilde{X}_{it}^{\prime} \Big)
\]
is positive definite by Assumption~\ref{assu:limit}. 
For the exact expressions of these terms, see Theorem~\ref{thm:C1:}.

Moreover, Theorem~\ref{thm:C1:} shows that 
$B^{(1)} - \mathbb{E}_\theta(B^{(1)})$ and $B^{(2)} - \mathbb{E}_\theta(B^{(2)})$ 
converge to zero in probability uniformly over every $\theta \in \Theta_0$. 
Then, given Assumption~\ref{assu:limit}, 
we have that $B^{(1)} + B^{(2)}$ converges uniformly to 
$\lim_{N,T\to\infty}\mathbb{E}_\theta\big(B^{(1)}+B^{(2)}\big)$. 

Finally, following the proof of Theorem~4.1 in \citet{HigginsJochmans2024}, 
we can apply the Berry--Esseen bound together with the uniform $\alpha$-mixing condition in Assumption~\ref{assu:correct} 
and the uniform moment conditions in Assumption~\ref{assu:smooth} 
to establish the uniform asymptotic normality of 
$\overline{W}_{\infty}^{-1} \partial_{\beta} \mathcal{L}^{*}$, 
and hence also the uniform asymptotic normality of 
$\sqrt{NT}(\hat{\beta}-\beta_{0})$.

\end{myproof}

\bigskip

\begin{myproof}[Proof of Theorem \ref{theorem3.2}]
 We show that
\[
\mathbb{P}_{\theta_0}\left( \sup_{a} \left| \mathbb{P}_{\hat{\theta}}\left(\sqrt{NT}(\hat{\beta}^* - \hat{\beta}) \leq a \right) 
- \mathbb{P}_{\theta_0}\left( \sqrt{NT}(\hat{\beta} - \beta_0) \leq a \right) \right| > \epsilon \right) = o_p(1).
\]

First, observe that
\begin{align*}
&\mathbb{P}_{\theta_0}\left( \sup_{a} \left| \mathbb{P}_{\hat{\theta}}\left(\sqrt{NT}(\hat{\beta}^* - \hat{\beta}) \leq a \right) 
- \mathbb{P}_{\theta_0}\left( \sqrt{NT}(\hat{\beta} - \beta_0) \leq a \right) \right| > \epsilon \right)\\
&\leq 
\mathbb{P}_{\theta_0}\left( \sup_{a} \left| \mathbb{P}_{\theta_0}\left(\sqrt{NT}(\hat{\beta} - \beta_0) \leq a\right) 
- \mathbb{P}_{\theta_0}(G_{\theta_0} \leq a) \right| > \frac{\epsilon}{2} \right) \\
&\quad + \mathbb{P}_{\theta_0}\left( \sup_{a} \left| \mathbb{P}_{\hat{\theta}}\left(\sqrt{NT}(\hat{\beta}^* - \hat{\beta}) \leq a\right) 
- \mathbb{P}_{\theta_0}(G_{\theta_0} \leq a) \right| > \frac{\epsilon}{2} \right) \\
&\leq 
\mathbb{P}_{\theta_0}\left( \sup_{a} \left| \mathbb{P}_{\theta_0}\left(\sqrt{NT}(\hat{\beta} - \beta_0) \leq a\right) 
- \mathbb{P}_{\theta_0}(G_{\theta_0} \leq a) \right| > \frac{\epsilon}{2} \right) \\
&\quad + \mathbb{P}_{\theta_0}\left( \sup_{a} \left| \mathbb{P}_{\hat{\theta}}\left(\sqrt{NT}(\hat{\beta}^* - \hat{\beta}) \leq a\right) 
- \mathbb{P}_{\hat{\theta}}(v_{\hat{\theta}} \leq a) \right| > \frac{\epsilon}{4} \right) \\
&\quad + \mathbb{P}_{\theta_0}\left( \sup_{a} \left| \mathbb{P}_{\hat{\theta}}(v_{\hat{\theta}} \leq a) 
- \mathbb{P}_{\theta_0}(G_{\theta_0} \leq a) \right| > \frac{\epsilon}{4} \right),
\end{align*}
where $G_\theta \sim \mathcal{N}(\overline{B}_\infty(\theta), \overline{W}_\infty(\theta)^{-1})$, with $\overline{B}_\infty(\theta)$ and $\overline{W}_\infty(\theta)$ defined in Assumption~\ref{assu:limit}.

Define
\[
g_{NT}(\theta) = \sup_{a \in \mathbb{R}^{d_x}} 
\left| \mathbb{P}_{\theta}\left(\sqrt{NT}(\hat{\beta} - \beta) \leq a \right) 
- \mathbb{P}_{\theta}\left(G_{\theta} \leq a\right) \right|.
\]

By Theorem \ref{thm:C1:} in Appendix~\ref{Appendix C}, we have $\mathbb{P}_{\theta_0}\left(\hat{\theta} \notin \Theta_0\right) = o_p(1)$. If
$\sup_{\theta \in \Theta_0} g_{NT}(\theta) = o(1)$, then by Lemma~A.1 in \citet{Andrews2005} and Theorem~\ref{theorem3.1}, it follows that
\[
\mathbb{P}_{\theta_0}\left( g_{NT}(\theta_0) > \frac{\epsilon}{2} \right) 
= \mathbb{P}_{\theta_0}\left( \sup_{a \in \mathbb{R}^{d_x}} 
\left| \mathbb{P}_{\theta_0}\left(\sqrt{NT}(\hat{\beta} - \beta_0) \leq a \right) 
- \mathbb{P}_{\theta_0}\left(G_{\theta_0} \leq a\right) \right| > \frac{\epsilon}{2} \right) = o_p(1).
\]

Thus, it suffices to show that $\sup_{\theta\in\Theta_0} g_{NT}(\theta) = o(1)$.
Without loss of generality, let $d_x = 1$. Theorem~\ref{theorem3.1} implies that
\[
\sup_{a\in\mathbb{R}} \sup_{\theta \in \Theta_0} 
\left| \mathbb{P}_{\theta}\left(\sqrt{NT}(\hat{\beta} - \beta) \leq a \right) 
- \mathbb{P}_{\theta}\left(G_{\theta} \leq a\right) \right| = o(1).
\]

Since $\Theta_0 \subset \Theta$ and $\Theta$ is compact, and $\overline{B}_\infty(\theta)$ as well as $\overline{W}_\infty(\theta)$ are continuous in $\theta$ by Assumption~\ref{assu:smooth}, it is possible (for any $k \geq 1$) to find a partition $-\infty = a_0 < a_1 < \cdots < a_M = +\infty$ such that
\[
\sup_{\theta \in \Theta_0} \left(\mathbb{P}_{\theta}(G_\theta \leq a_i) - \mathbb{P}_{\theta}(G_\theta \leq a_{i-1})\right) \leq \frac{1}{k}, \quad \text{for } i=1,\ldots,M.
\]

For each $a \in \mathbb{R}$, there exists $i$ such that $a_i < a \leq a_{i+1}$, so for every $\theta \in \Theta_0$:
\begin{align*}
\mathbb{P}_{\theta}\left(\sqrt{NT}(\hat{\beta} - \beta) \leq a\right) - \mathbb{P}_{\theta}\left(G_\theta \leq a\right)
&\leq \mathbb{P}_{\theta}\left(\sqrt{NT}(\hat{\beta} - \beta) \leq a_{i+1}\right) - \mathbb{P}_{\theta}\left(G_\theta \leq a_i\right) \\
&\leq \left| \mathbb{P}_{\theta}\left( \sqrt{NT}(\hat{\beta} - \beta) \leq a_{i+1}\right) - \mathbb{P}_{\theta}\left(G_\theta \leq a_{i+1}\right) \right| + \frac{1}{k},
\end{align*}
and similarly (lower bound):
\[
\mathbb{P}_{\theta}\left(\sqrt{NT}(\hat{\beta} - \beta) \leq a\right) 
- \mathbb{P}_{\theta}\left(G_\theta \leq a\right)
\geq -\left| \mathbb{P}_{\theta}\left( \sqrt{NT}(\hat{\beta} - \beta) \leq a_{i}\right) - \mathbb{P}_{\theta}\left(G_\theta \leq a_{i}\right) \right| - \frac{1}{k}.
\]

Hence,
\[
\sup_{a} \left| \mathbb{P}_{\theta}\left(\sqrt{NT}(\hat{\beta} - \beta) \leq a\right) - \mathbb{P}_{\theta}\left(G_\theta \leq a\right) \right|
\leq \max_{1 \leq i \leq M} \left| \mathbb{P}_{\theta}\left(\sqrt{NT}(\hat{\beta} - \beta) \leq a_i\right) 
- \mathbb{P}_{\theta}\left(G_\theta \leq a_i\right) \right| + \frac{1}{k}.
\]

Taking sup over $\theta \in \Theta_0$, and since the first term is $o(1)$ by Theorem~\ref{theorem3.1}, and $1/k$ can be arbitrarily small, it follows that:
\[
\sup_{\theta \in \Theta_0} g_{NT}(\theta) = o(1).
\]

Similarly, we have:
\[
\sup_{\theta \in \Theta} \mathbb{P}_{\theta}\left( g_{NT}(\hat{\theta}) > \frac{\epsilon}{4} \right)
= \sup_{\theta \in \Theta} \mathbb{P}_{\theta}\left( \sup_{a} \left| \mathbb{P}_{\hat{\theta}} \left(\sqrt{NT}(\hat{\beta}^* - \hat{\beta}) \leq a\right) 
- \mathbb{P}_{\hat{\theta}}(v_{\hat{\theta}} \leq a) \right| > \frac{\epsilon}{4} \right) = o(1),
\]
by Lemma~A.1 in \citet{Andrews2005} and following a similar argument as the first term.

By continuity of $B_\infty(\theta)$ and $W_\infty(\theta)^{-1}$ in $\theta$ and the continuous mapping theorem,
\[
\mathbb{P}_{\theta_0}\left( \sup_{a \in \mathbb{R}} 
\left| \mathbb{P}_{\hat{\theta}}(v_{\hat{\theta}} \leq a) - \mathbb{P}_{\theta_0}(v_{\theta_0} \leq a) \right| 
> \frac{\epsilon}{4} \right) = o(1).
\]

This completes the proof.
\end{myproof}

\bigskip

\begin{myproof}[Proof of Theorem \ref{theorem3.3}]
Under Assumptions~\ref{assu:parameterspace}--\ref{assu:APEs smooth}, 
Theorem~\ref{Thm:C2} establishes the uniform asymptotic expansion for APEs. 
Given Assumption~\ref{assu:limit:ape}, 
the proof then follows a similar argument as in the proof of Theorem~\ref{theorem3.1}.

\end{myproof}

\begin{myproof}[Proof of Theorem \ref{theorem3.4}]
Theorem \ref{theorem3.3} establishes the uniform asymptotic normality of APEs. The bootstrap consistency proof follows the same arguments as the proof of Theorem \ref{theorem3.2}
\end{myproof}

We first review some notation for norms: 
$\|\cdot\|$ denotes the Euclidean norm for vectors, 
$\|\cdot\|_{F}$ denotes the Frobenius norm for matrices.

\begin{lemma}\label{lemma:A1}
Under Assumptions \ref{assu:parameterspace} and \ref{ass:strong factor}, we have
\[
\sup_{\theta\in\Theta_0}\left\Vert \frac{1}{N}\sum_{i=1}^{N}\alpha_{i}\alpha_{i}'-\Sigma_{\alpha}\right\Vert _{F}\to0,\qquad\sup_{\theta\in\Theta_0}\left\Vert \frac{1}{T}\sum_{t=1}^{T}\gamma_{t}\gamma_{t}'-\Sigma_{\gamma}\right\Vert _{F}\to0.
\]
\end{lemma}

\begin{proof}
It is sufficient to show that
\[
\sup_{\theta\in\Theta_0}
\left\|
\frac{1}{N}\sum_{i=1}^{N}
\left(\alpha_{i}\alpha_{i}^{\prime}-\alpha_{i0}\alpha_{i0}^{\prime}\right)
\right\|_{F} \to 0,
\qquad
\sup_{\theta\in\Theta_0}
\left\|
\frac{1}{T}\sum_{t=1}^{T}
\left(\gamma_{t}\gamma_{t}^{\prime}-\gamma_{t0}\gamma_{t0}^{\prime}\right)
\right\|_{F} \to 0.
\]
The stated result then follows from Assumption~\ref{ass:strong factor}. We only prove the first claim, since the second follows analogously.

For every $\theta\in\Theta_0$,
\begin{align*}
\left\|
\frac{1}{N}\sum_{i=1}^{N} \left(\alpha_{i}\alpha_{i}^{\prime} - \alpha_{i0}\alpha_{i0}^{\prime}\right)
\right\|_{F} 
&=
\left\|
\frac{1}{N}\sum_{i=1}^{N}
\Big[
\alpha_{i}(\alpha_{i}^{\prime}-\alpha_{i0}^{\prime})
+(\alpha_{i}-\alpha_{i0})\alpha_{i0}^{\prime}
\Big]
\right\|_{F} \\
&\le 
\left\|
\frac{1}{N}\sum_{i=1}^{N} \alpha_{i}(\alpha_{i}-\alpha_{i0})^{\prime}
\right\|_{F}
+
\left\|
\frac{1}{N}\sum_{i=1}^{N} (\alpha_{i}-\alpha_{i0})\alpha_{i0}^{\prime}
\right\|_{F} \\
&\le 
\left(\frac{1}{N}\sum_{i=1}^{N} \|\alpha_{i}\|^{2}\right)^{1/2}
\left(\frac{1}{N}\sum_{i=1}^{N} \|\alpha_{i}-\alpha_{i0}\|^{2}\right)^{1/2} \\
& +
\left(\frac{1}{N}\sum_{i=1}^{N} \|\alpha_{i0}\|^{2}\right)^{1/2}
\left(\frac{1}{N}\sum_{i=1}^{N} \|\alpha_{i}-\alpha_{i0}\|^{2}\right)^{1/2}.
\end{align*}

And,
\[
\frac{1}{N}\sum_{i=1}^{N} \|\alpha_{i}-\alpha_{i0}\|^{2}
= \frac{1}{N}\big\|\mathrm{vec}(\alpha)-\mathrm{vec}(\alpha_0)\big\|^{2}
\le \frac{1}{N}\|\theta-\theta_0\|^2
\le \frac{\varepsilon^2}{N},
\]
uniformly over $\theta\in\Theta_0$, and hence
$\sup_{\theta\in\Theta_0}
\frac{1}{N}\sum_{i=1}^{N} \|\alpha_{i}-\alpha_{i0}\|^{2}
\to 0.$

Moreover, Assumption~\ref{ass:strong factor} implies
\[
\Big\|\frac{\alpha_0}{\sqrt N}\Big\|_F^2
=
\frac1N\sum_{i=1}^N \|\alpha_{i0}\|^2
=
\operatorname{tr}\!\left(\frac1N\sum_{i=1}^N \alpha_{i0}\alpha_{i0}'\right)
\to \operatorname{tr}(\Sigma_\alpha),
\]
so that
$\Big\|\frac{\alpha_0}{\sqrt N}\Big\|_F=O(1).$

Then,
\[
\Big\|\frac{\alpha}{\sqrt N}\Big\|_F
\le
\Big\|\frac{\alpha-\alpha_0}{\sqrt N}\Big\|_F
+
\Big\|\frac{\alpha_0}{\sqrt N}\Big\|_F
\le
\frac{\varepsilon}{\sqrt N}
+
\Big\|\frac{\alpha_0}{\sqrt N}\Big\|_F,
\]
and therefore
\[
\sup_{\theta\in\Theta_0}\Big\|\frac{\alpha}{\sqrt N}\Big\|_F=O(1).
\]

Combining the above bounds yields
\[
\sup_{\theta\in\Theta_0}
\left\|
\frac{1}{N}\sum_{i=1}^{N}
\left(\alpha_{i}\alpha_{i}^{\prime}-\alpha_{i0}\alpha_{i0}^{\prime}\right)
\right\|_{F}
\to 0.
\]
This proves the first claim. 
\end{proof}

\clearpage
\begin{center}
    \Large\bfseries ONLINE SUPPLEMENTAL MATERIAL
\end{center}
\addcontentsline{toc}{section}{Online Supplemental Material}
\ifdefined\buildappendix

\clearpage
\fi
\ifdefined\buildmain\else

\section{Uniform Expansion for Panel Models}\label{Appendix B}
\subsection{Notations}
We define the vector of unobserved effects as 
\[
\phi := \left( \mathrm{vec}(\alpha)^{\prime},\, \mathrm{vec}(\gamma)^{\prime} \right)^{\prime}.
\]

Following \citet{ChenFVWeidner2021}, the factors $\alpha_i$ and $\gamma_t$ satisfy the normalization
\[
\Phi := 
\left\{
\phi \in \mathbb{R}^{\dim(\phi)} :
\sum_{i=1}^{N} \alpha_{i0}\alpha_i^{\prime}
=
\sum_{t=1}^{T} \gamma_t\gamma_{t0}^{\prime}
\right\},
\]
where $\alpha_{i0}$ and $\gamma_{t0}$ are the true parameter values. 
This normalization is convenient for theoretical analysis but infeasible in practice. 

The penalized likelihood function is
\[
\mathcal{L}(\beta,\phi)
=
\frac{1}{\sqrt{NT}}
\!\left(
\sum_{i,t} \ell_{it}(z_{it})
-\frac{b}{2}\phi^{\prime}VV^{\prime}\phi
\right),
\qquad
\ell_{it}(z_{it})=\log f(X_{it}^{\prime}\beta+\alpha_i^{\prime}\gamma_t),
\]
where $b > 0$  and $V$ is a $\dim(\phi)\times d_f^{2}$ matrix satisfying 
\[
V^{\prime}\phi
=
\mathrm{vec}\!\left(
\sum_{i}\alpha_{i0}\alpha_i^{\prime}
-
\sum_{t}\gamma_t\gamma_{t0}^{\prime}
\right).
\]

In this section, we use $\overline{\mathbb{E}}$ to denote the expectation conditional on $\phi$.
Let
\[
\overline{\mathcal{H}}^{*}
=
\mathbb{E}_{\theta}\!\left[-\partial_{\phi\phi}\mathcal{L}^{*}\right]
=
\begin{bmatrix}
\overline{\mathcal{H}}_{\alpha\alpha}^{*} & \overline{\mathcal{H}}_{\alpha\gamma}^{*}\\
(\overline{\mathcal{H}}_{\alpha\gamma}^{*})^{\prime} & \overline{\mathcal{H}}_{\gamma\gamma}^{*}
\end{bmatrix},
\]
where
$\overline{\mathcal{H}}_{\alpha\alpha}^{*} \in \mathbb{R}^{Nd_f \times Nd_f},
\overline{\mathcal{H}}_{\gamma\gamma}^{*} \in \mathbb{R}^{Td_f \times Td_f},
\overline{\mathcal{H}}_{\alpha\gamma}^{*} \in \mathbb{R}^{Nd_f \times Td_f}.$
and
\[
\overline{\mathcal{H}}
=
\mathbb{E}_{\theta}\!\left[-\partial_{\phi\phi}\mathcal{L}\right]
=
\overline{\mathcal{H}}^{*}
+\frac{b}{\sqrt{NT}}VV^{\prime}.
\]

The profiled estimators are
\[
\hat{\beta}
=\arg\max_{\beta\in\mathbb{R}^{d_{\beta}}}\mathcal{L}\!\left(\beta,\hat{\phi}(\beta)\right),
\qquad
\hat{\phi}(\beta)
=\arg\max_{\phi\in\mathbb{R}^{d_{\phi}}}\mathcal{L}(\beta,\phi).
\]

We use the following notation for norms. For a matrix $A$, the spectral, Frobenius, and matrix sup norms are denoted by $\|A\|$, $\|A\|_{F}$, and $\|A\|_{\max}$, respectively. For a vector $v \in \mathbb{R}^n$ and $q \in [1,\infty]$, the $\ell_q$ norm is defined as $\|v\|_q = (\sum_{i=1}^n |v_i|^q)^{1/q}$. When $v$ is viewed as an $n \times 1$ matrix, its spectral norm coincides with its Euclidean norm, i.e., $\|v\| = \|v\|_2 = \sqrt{v'v}$. We define the associated orthogonal projection matrix as $\mathcal{M}_A = \mathbb{I} - A (A'A)^{\dagger} A'$, where $\dagger$ denotes the Moore--Penrose generalized inverse and $\mathbb{I}$ the identity matrix of the same dimension as $AA'$. For $r>0$, let $\mathcal{B}(r,\beta) = \{\tilde{\beta}\in\mathbb{R}^{d_f}: \|\tilde{\beta}-\beta\|_2 < r\}$ denote the open $\ell_2$ ball of radius $r$ centered at $\beta$, and more generally, for $q \in [1,\infty]$, let $\mathcal{B}_q(r,\phi) = \{\tilde{\phi}\in\mathbb{R}^{(N+T)d_f} : \|\tilde{\phi}-\phi\|_q < r\}$ denote the open $\ell_q$ ball of radius $r$ centered at $\phi$.

\subsection{Auxiliary Lemmas}
For a sequence of random processes $A_N(\theta)$ whose distribution depends on $\theta \in \Theta_0$, we write $A_N(\theta) = o_p^u(1)$ if, for any $\epsilon>0$, it holds that $\sup_{\theta\in\Theta_0} \mathbb{P}_{\theta}(|A_N(\theta)| > \epsilon) \to 0$ as $N \to \infty$; that is, $A_N(\theta)$ converges uniformly in $\theta\in\Theta_0$ to zero in probability. Similarly, we write $A_N(\theta) = \mathcal{O}_p^u(1)$ if, for any $\epsilon>0$, there exists a constant $k$ such that $\sup_{\theta\in\Theta_0} \mathbb{P}_{\theta}(|A_N(\theta)| > k) < \epsilon$ for all $N$, meaning that $A_N(\theta)$ is uniformly bounded in probability. If $A_N$ is nonstochastic and satisfies $A_N(\theta) = \mathcal{O}_p^u(1)$ or $A_N(\theta) = o_p^u(1)$, we simply write $A_N(\theta) = \mathcal{O}^u(1)$ or $A_N(\theta) = o^u(1)$. For two sequences of random processes $A_N(\theta)$ and $B_N(\theta)$, we write $A_N(\theta) = \mathcal{O}_p^u(B_N(\theta))$ when $A_N(\theta)/B_N(\theta) = \mathcal{O}_p^u(1)$, and we write $A_N(\theta) = o_p^u(B_N(\theta))$ when $A_N(\theta)/B_N(\theta) = o_p^u(1)$. Lemma \ref{propeties of O} gives standard properties of $\mathcal{O}_p^u$ and $o_p^u$.

\citet{FVWeidner2016} established an asymptotic expansion for a general likelihood model with incidental parameters. In Appendix~\ref{Appendix C}, we extend their framework to the uniform case, where the expansion holds uniformly under all parameter values within a small neighborhood of the true parameters. This uniform extension is useful in establishing the consistency of bootstrap confidence intervals.

To apply the framework in Appendix~\ref{Appendix C}, we first establish several auxiliary lemmas. These lemmas are adapted from \citet{ChenFVWeidner2021} and extended to the uniform setting.

\begin{lemma} \label{lemma:c1}
Let Assumptions~\ref{assu:parameterspace}-\ref{assu:convexity} hold. Then, it holds that
\[
\|\hat{\beta} - \beta\| =\mathcal{O}^u_p\left(\sqrt{NT}^{-3/8}\right),
\]
and 
there exists a constant $\sigma>0$ such that
\begin{align*}
\sup_{\widetilde{\beta}\in\mathcal{B}(\beta,\sigma)}\frac{1}{\sqrt{NT}}
\left\|
\hat{\alpha}(\widetilde{\beta})\hat{\gamma}^{\prime}(\widetilde{\beta})
- \alpha\gamma^{\prime}
\right\|_{F}
&=\mathcal{O}_p^u\left((NT)^{-3/8} + \|\widetilde{\beta} - \beta\|\right), \\
\sup_{\widetilde{\beta}\in\mathcal{B}(\beta,\sigma)}\frac{1}{\sqrt{N}}
\left\|
\hat{\phi}(\widetilde{\beta}) - \phi
\right\|
&=\mathcal{O}_p^u\left((NT)^{-3/8} + \|\widetilde{\beta} - \beta\|\right).
\end{align*}
\end{lemma}

\begin{proof}
The proof follows the same steps as the proof of Lemma~1 in \citet{ChenFVWeidner2021} with the modification that the bound holds uniformly on $\theta \in \Theta_0$.
\end{proof}

To apply the results in Appendix~\ref{Appendix C}, we also need to show that the Hessian matrix of the penalized log-likelihood is asymptotically invertible. 
\citet{ChenFVWeidner2021} show that the Hessian matrix can be approximated by an invertible block-diagonal matrix. 
As discussed in \citet{ChenFVWeidner2021}, the matrix $\overline{\mathcal{H}}^{*}$ has $d_f^{2}$ zero eigenvalues, corresponding to transformations of the factors and factor loadings that leave the likelihood function unchanged. 
However, since $VV^{\prime}$ has rank $d_f^{2}$, the penalized Hessian $\overline{\mathcal{H}}$ is invertible. 
The following lemma extends this result to the uniform case.

\begin{lemma}\label{lemma:c2}
Under Assumptions~\ref{assu:parameterspace}, \ref{assu:smooth} and \ref{assu:convexity},
\[
\left\|
\overline{\mathcal{H}}^{-1}
- 
\operatorname{diag}\!\left(
\overline{\mathcal{H}}_{\alpha\alpha}^{*},
\overline{\mathcal{H}}_{\gamma\gamma}^{*}
\right)^{-1}
\right\|_{\max}
=
\mathcal{O}^{u}\!\left( (NT)^{-1/2} \right).
\]
\end{lemma}

\begin{proof}
The proof follows similarly to Lemma~2 in \citet{ChenFVWeidner2021}.
The uniform bound is obtained under the strong factor assumption (see Lemma~\ref{lemma:A1}) and the uniform moment bound.
\end{proof}

Theorem \ref{uniform consistency} requires the objective function to be strictly concave for every $\left(\beta, \phi\right)$. However, due to the structure of the interactive fixed effects, the Hessian of our (non-penalized) likelihood function is not positive definite for all $\left(\beta, \phi\right)$. Nevertheless, given Lemma~\ref{lemma:c1}, it suffices to show that the Hessian is locally positive definite wpa1.

\begin{lemma} \label{lemma:c3}
Let Assumptions~\ref{assu:parameterspace}-\ref{assu:convexity} hold, and let $r_{\beta} = o_{p}(1)$ and $r_{\phi} = o_{p}\big((NT))^{1/4}\big)$. For all $\theta=(\beta,\phi)\in \Theta_0$, the Hessian $\mathcal{H}(\widetilde{\beta}, \widetilde{\phi})$ is positive definite for all 
$\widetilde{\beta} \in \mathcal{B}(r_{\beta}, \beta)$ and $\widetilde{\phi} \in \mathcal{B}(r_{\phi}, \phi)$ wpa 1 (under $\mathbb{P}_\theta$)\footnote{%
Here, $\mathcal{B}(r_{\phi}, \phi) := \{\widetilde{\phi}\in\mathbb{R}^{(N+T)d_f} : \|\widetilde{\phi} - \phi\|_2 \le r_{\phi}\}$ denotes the $\ell_2$-ball of radius $r_{\phi}$ centered at $\phi$.}.
\end{lemma}
\begin{proof}
See Lemma~3 in \citet{ChenFVWeidner2021}. 
\end{proof}

\subsection{Uniform Expansions}
Before introducing the uniform expansion, we introduce additional notation to simplify the expressions appearing in the asymptotic distribution. Let $\ell_{it}(\beta,\alpha_i,\gamma_t)$ denote the individual log-likelihood. Following \citet{FVWeidner2018, ChenFVWeidner2021}, it is convenient to consider the following informationally orthogonal log-likelihood:
\[
\ell_{it}^*(\beta,\alpha_i,\gamma_t)
:= \ell_{it}\big(\beta,\, \alpha_i+\alpha_i^*\beta,\, \gamma_t+\gamma_t^*\beta\big),
\]
where $\alpha_{i,k}^*$ and $\gamma_{t,k}^*$ are $R$-dimensional vectors defined as
\[
(\alpha_{i,k}^*,\gamma_{t,k}^*) 
\in \argmin_{\widetilde{\alpha}_{i,k},\,\widetilde{\gamma}_{t,k}} 
\sum_{i,t} \mathbb{E}_{\theta}\!\left[-\partial_{z^2}\ell_{it}\right]
\left(
\frac{\mathbb{E}_{\theta}\!\left[\partial_{z^2}\ell_{it} X_{it,k}\right]}{\mathbb{E}_{\theta}\!\left[\partial_{z^2}\ell_{it}\right]} 
- \widetilde{\alpha}_{i,k}^\prime\gamma_{t} - \alpha_{i}^\prime\widetilde{\gamma}_{t,k}
\right)^{\!2}.
\]

This construction is useful because it ensures
\[
\frac{\partial^2}{\partial_{\beta}\partial_{\alpha_i}}
\sum_{i,t}\mathbb{E}_{\theta}\!\big[\ell_{it}^*(\beta,\alpha_i,\gamma_t)\big] = 0,
\]
that is, informational orthogonality between $\beta$ and $\alpha_i$.

The derivative of $\ell_{it}^*$ with respect to $\beta$ is given by
\[
\partial_\beta \ell_{it}^*(\beta,\alpha_i,\gamma_t)
= \partial_z\ell_{it}\!\big(X_{it}^\prime\beta+\alpha_i^\prime\gamma_t\big)\,\widetilde{X}_{it},
\]
where $\widetilde{X}_{it} = X_{it} - \varXi_{it}$ and
\[
\varXi_{it,k} = \alpha_{i,k}^{*\prime}\gamma_{t} + \alpha_{i}^\prime\gamma_{t,k}^*.
\]
Analogous expressions hold for derivatives with respect to $\alpha_i$ and $\gamma_t$, as well as for higher-order derivatives. 

Finally, define
\begin{equation}
\label{eq:orthloglike}
\mathcal{L}^*(\beta,\phi)
:= \frac{1}{\sqrt{NT}} \sum_{i=1}^N \sum_{t=1}^T 
\ell_{it}^*(\beta,\alpha_i,\gamma_t).
\end{equation}

Given the above lemmas and Theorem \ref{thm:D1}, we have the following uniform asymptotic expansion.  
\begin{theorem}[]
\label{thm:C1:}
Let Assumptions \ref{assu:parameterspace}--\ref{assu:convexity} hold. Suppose that \(q=8\) and \(v>0\), and define \(\epsilon=\frac{1}{2q+v}\), \(r_{\beta}=\log(NT)(NT)^{-\frac{1}{q}}\), and \(r_{\phi}=(NT)^{-\frac{1}{2q}}\). Then, for every \(\theta \in \Theta_0\), if \(\overline{W}_\infty\) is invertible,
$\sqrt{NT} (\hat{\beta} - \beta)
= \overline{W}_\infty^{-1} \big( U^{(0)} + U^{(1)} \big) + o_p^u(1).$
 where
\begin{align*}
\overline{W}_\infty:=&\lim_{N,T\to \infty}\frac{1}{NT}\sum_{i=1}^{N}\sum_{t=1}^{T}\mathbb{E}_{\theta}\!\left(\partial_{z^{2}}\ell_{it}\widetilde{X}_{it}\widetilde{X}_{it}^{\prime}\right)\\
U^{(0)}:=&\frac{1}{\sqrt{NT}}\sum_{i=1}^{N}\sum_{t=1}^{T}\partial_{z}\ell_{it}\widetilde{X}_{it}\\
U^{(1)}:=&\;-\frac{1}{\sqrt{NT}}\sum_{i=1}^{N}\sum_{t=1}^{T}\Bigg\{\mathbb{E}_{\theta}\!\left[\gamma_{t}^{\!\prime}\!\left(\sum_{h=1}^{T}\gamma_{h}\gamma_{h}^{\!\prime}\mathbb{E}_{\theta}\!\left[\partial_{\alpha_{i}\alpha_{i}}\ell_{ih}\right]\right)^{-1}\gamma_{t0}^{\!}\left(\partial_{z^{2}}\ell_{it}\partial_{z}\ell_{it}\widetilde{X}_{it}\right)\right]\\&+\mathbb{E}_{\theta}\left[\alpha_{i}^{\!\prime}\!\left(\sum_{h=1}^{N}\alpha_{h}\alpha_{h}^{\!\prime}\mathbb{E}_{\theta}\!\left[\partial_{\gamma_{t}\gamma_{t}}\ell_{ht}\right]\right)^{-1}\alpha_{i}\left(\partial_{z^{2}}\ell_{it}\partial_{z}\ell_{it}\widetilde{X}_{it}\right)\right]\Bigg\}\\&-\frac{1}{2}\frac{1}{\sqrt{NT}}\sum_{i=1}^{N}\sum_{t=1}^{T}\left\{ \mathbb{E}_{\theta}\left[\gamma_{t}^{\!\prime}\!\left(\sum_{h=1}^{T}\gamma_{h}\gamma_{h}^{\!\prime}\mathbb{E}_{\theta}\!\left[\partial_{\alpha_{i}\alpha_{i}}\ell_{ih}\right]\right)^{-1}\gamma_{t}^{\!}\!\left(\partial_{z^{3}}\ell_{it}\widetilde{X}_{it}\right)\right]\right.\\&\left.+\mathbb{E}_{\theta}\left[\alpha_{i}^{\!\prime}\!\left(\sum_{h=1}^{N}\alpha_{h}\alpha_{h}^{\!\prime}\mathbb{E}_{\theta}\!\left[\partial_{\gamma_{t}\gamma_{t}}\ell_{ht}\right]\right)^{-1}\alpha_{i}\left(\partial_{z^{3}}\ell_{it}\widetilde{X}_{it}\right)\right]\right\} .
\end{align*}

\end{theorem}
\begin{proof}
To establish the uniform asymptotic expansion for $\sqrt{NT}(\hat{\beta}-\beta_0)$ via Theorem~\ref{thm:D1}, it suffices to verify that Assumption \ref{assu:uniform regularity} holds. Part (i) of Assumption \ref{assu:uniform regularity} follows from the asymptotic regime $T/N \to \kappa \in (0,\infty)$, while part (ii) follows from Assumption \ref{assu:smooth}.

Part (iv) follows from Lemma~\ref{lemma:c2} and 
\[
\left\Vert \operatorname{diag}\!\left(
\overline{\mathcal{H}}_{\alpha\alpha}^{*},
\overline{\mathcal{H}}_{\gamma\gamma}^{*}
\right)^{-1} \right\Vert_{\infty}
= \mathcal{O}_{p}^{u}(1).
\]

Following the same argument as in the proof of Theorem~C.1 in \citet{FVWeidner2016}, Assumptions \ref{assu:parameterspace}--\ref{assu:convexity}, together with a suitable choice of $q$, $\epsilon$, $r_{\beta}$, and $r_{\phi}$, imply that the conditions of Lemma~\ref{lemma.D2} are satisfied. Hence, parts (v) and (vi) follow from Lemma~\ref{lemma.D2}. 

Finally, combining the initial convergence rates in Lemma~\ref{lemma:c1} with the local strict concavity established in Lemma~\ref{lemma:c3}, Theorem~\ref{uniform consistency} applies, and thus part (iii) holds.

Then, by applying Theorem~\ref{thm:D1}, and given the definition of $\mathcal{L}^*(\beta,\phi)$, we obtain
\[
\sqrt{NT}(\hat{\beta}-\beta_{0})
= \overline{W}_{\infty}^{-1}
\left(
\partial_{\beta}\mathcal{L}^{*}
+ B^{(1)} + B^{(2)}
\right)
+ o_{p}^{u}(1),
\]
where
\begin{equation}
\label{eq:biasterms}
\begin{aligned}
B^{(1)} 
&= \left[\partial_{\beta \phi^{\prime}} \widetilde{\mathcal{L}}\right] \overline{\mathcal{H}}^{-1} \mathcal{S}-\left[\partial_{\beta \phi^{\prime}} \overline{\mathcal{L}}\right] \overline{\mathcal{H}}^{-1} \widetilde{\mathcal{H}} \overline{\mathcal{H}}^{-1} \mathcal{S},\\
B^{(2)} 
&= \frac{1}{2} \sum_{a=1}^{\operatorname{dim} \phi}\left(\partial_{\beta \phi^{\prime} \phi_g} \overline{\mathcal{L}}+\left[\partial_{\beta \phi^{\prime}} \overline{\mathcal{L}}\right] \overline{\mathcal{H}}^{-1}\left[\partial_{\phi \phi^{\prime} \phi_g} \overline{\mathcal{L}}\right]\right)\left[\overline{\mathcal{H}}^{-1} \mathcal{S}\right]_g \overline{\mathcal{H}}^{-1} \mathcal{S} .
\end{aligned}
\end{equation}
Where $\partial_{\beta}\mathcal{L}^{*}$ contribute to the asymptotic variance which equal to $U^{(0)}$ by construction. The left terms $B^{(1)}$, $B^{(2)}$ contribute to the asymptotic bias. And we need to verify they
converge uniformly in probability to $U^{(1)}$. We first show the uniform convergence of $B^{(1)}$.

 Hence, following the same argument in \citet{FVWeidner2016}, we obtain
\begin{align*}
& B^{(1)}=-\frac{1}{\sqrt{NT}}\sum_{i=1}^{N}\underbrace{\left(\sum_{t=1}^{T}\gamma_{t}^{\prime}\partial_{z}\ell_{it}\right)\left(\sum_{h=1}^{T}\gamma_{h}\mathbb{E}_{\theta}[\partial_{z^{2}}\ell_{ih}]\gamma_{h}^{\prime}\right)^{-1}\left(\sum_{\tau=1}^{T}\gamma_{\tau}\left(\partial_{z^{2}}\ell_{i\tau}\widetilde{X}_{i\tau}-\mathbb{E}_{\theta}\left(\partial_{z^{2}}\ell_{i\tau}\widetilde{X}_{i\tau}\right)\right)\right)}_{B_{i}^{(1,1)}}\\&-\frac{1}{\sqrt{NT}}\sum_{t=1}^{T}\underbrace{\left(\sum_{i=1}^{N}\alpha_{i}^{\prime}\partial_{z}\ell_{it}\right)\left(\sum_{h=1}^{N}\alpha_{h}\mathbb{E}_{\theta}[\partial_{z^{2}}\ell_{ht}]\alpha_{h}^{\prime}\right)^{-1}\left(\sum_{j=1}^{N}\alpha_{j}\left(\partial_{z^{2}}\ell_{jt}\widetilde{X}_{jt}-\mathbb{\mathbb{E}_{\theta}}\left(\partial_{z^{2}}\ell_{jt}\widetilde{X}_{jt}\right)\right)\right)}_{B_{i}^{(1,2)}}+o_{p}^{u}(1).
\end{align*}

To show the difference of $\frac{1}{\sqrt{NT}}\sum_{i=1}^N B_{i}^{(1,1)}$ and its expectation which is the first term in $U^{(1)}$ converge to zero in probability uniformly, we need to show the variance of $B_i$ is uniformly bounded over $\theta\in \Theta_0$ for all $i$ and given the fact that $B_i$ is independent, then uniform convergence of $\frac{1}{\sqrt{NT}}\sum_{i=1}^N B_i$ is implied by the WLLN and $N/T\to \kappa\in(0,\infty)$.

By Assumptions \ref{ass:strong factor} and \ref{assu:convexity}, and Lemma \ref{lemma:A1} we have
\[
\sup_{\theta\in\Theta_0}
\gamma_{t}^{\prime}
\!\left(
\sum_{h=1}^{T}
\gamma_{h}\,
\mathbb{E}_{\theta}[\partial_{z^{2}}\ell_{ih}]\,
\gamma_{h}^{\prime}
\right)^{\!-1}\!
\gamma_{t}
= O_p(T).
\]
uniformly for every $t,\tau$. Then we only need to check the boundedness of 
$$
\sup_{\theta \in \Theta_{0}} 
\mathbb{E}_{\theta}
\left(\frac{1}{T}
\sum_{t=1}^{T}\sum_{\tau=1}^{T}
\partial_{z}\ell_{it} \left(\partial_{z^{2}}\ell_{i\tau}\widetilde{X}_{i\tau}-\mathbb{E}_\theta\left(\partial_{z^{2}}\ell_{i\tau}\widetilde{X}_{i\tau}\right)\right)
\right)^{2}.$$
Denote $D_{z^{2}}\ell_{it}
= 
\partial_{z^{2}}\ell_{it}\widetilde{X}_{it}
- 
\mathbb{E}_\theta\left(\partial_{z^{2}}\ell_{it}\widetilde{X}_{it}\right)$, then
\[
\begin{aligned}
\sup_{\theta \in \Theta_{0}} 
\mathbb{E}_{\theta}
\left(\frac{1}{T}
\sum_{t=1}^{T}\sum_{\tau=1}^{T}
\partial_{z}\ell_{it} D_{z^{2}}\ell_{i\tau}
\right)^{2}
&=
\sup_{\theta \in \Theta_{0}}
\mathbb{E}_{\theta}
\left[
\left(\frac{1}{T^2}
\sum_{t_{1}=1}^{T}\sum_{\tau_{1}=1}^{T}
\sum_{t_{2}=1}^{T}\sum_{\tau_{2}=1}^{T}
\partial_{z}\ell_{it_{1}}
\partial_{z}\ell_{it_{2}}
D_{z^{2}}\ell_{i\tau_{1}}
D_{z^{2}}\ell_{i\tau_{2}}
\right)
\right].
\end{aligned}
\]

We decompose it as
\begin{align*}
&\sup_{\theta\in\Theta_{0}}
\Bigg\{
\mathbb{E}_{\theta}\!\left[
\frac{1}{T}\sum_{t_{1}=1}^{T}\sum_{t_{2}=1}^{T}
\partial_{z}\ell_{it_{1}}\partial_{z}\ell_{it_{2}}
\right]
\mathbb{E}_{\theta}\!\left[
\frac{1}{T}\sum_{\tau_{1}=1}^{T}\sum_{\tau_{2}=1}^{T}
D_{z^{2}}\ell_{i\tau_{1}}D_{z^{2}}\ell_{i\tau_{2}}
\right]  \\
&\quad+
\frac{1}{T^{2}}\sum_{t_{1}=1}^{T}\sum_{\tau_{1}=1}^{T}
\sum_{t_{2}=1}^{T}\sum_{\tau_{2}=1}^{T}
\mathrm{Cov}_{\theta}\!\left(
\partial_{z}\ell_{it_{1}}\partial_{z}\ell_{it_{2}},
D_{z^{2}}\ell_{i\tau_{1}}D_{z^{2}}\ell_{i\tau_{2}}
\right)
\Bigg\}.
\end{align*}

By Assumptions \ref{assu:correct} and \ref{assu:smooth} and Corollary~15.3 of \citet{Davidson2021}, we have
\[
\sup_{\theta\in\Theta_{0}}
\mathbb{E}_{\theta}\!\left[
\frac{1}{T}\sum_{t_{1}=1}^{T}\sum_{t_{2}=1}^{T}
\partial_{z}\ell_{it_{1}}\partial_{z}\ell_{it_{2}}
\right]
=
\sup_{\theta\in\Theta_{0}}
\frac{1}{T}\sum_{t_{1}=1}^{T}\sum_{t_{2}=1}^{T}
\mathrm{Cov}_{\theta}\!\left(\partial_{z}\ell_{it_{1}},\partial_{z}\ell_{it_{2}}\right)
=O_{p}(1),
\]
and similarly,
\[
\sup_{\theta\in\Theta_{0}}
\mathbb{E}_{\theta}\!\left[
\frac{1}{T}\sum_{\tau_{1}=1}^{T}\sum_{\tau_{2}=1}^{T}
D_{z^{2}}\ell_{i\tau_{1}}D_{z^{2}}\ell_{i\tau_{2}}
\right]
=O_{p}(1).
\]

For the covariance term, note that
\begin{align*}
&\sup_{\theta\in\Theta_{0}}
\frac{1}{T^{2}}
\sum_{t_{1},\tau_{1},t_{2},\tau_{2}=1}^{T}
\mathrm{Cov}_{\theta}\!\left(
\partial_{z}\ell_{it_{1}}\partial_{z}\ell_{it_{2}},
D_{z^{2}}\ell_{i\tau_{1}}D_{z^{2}}\ell_{i\tau_{2}}
\right)  \\
&\quad=
\frac{1}{T^{2}}
\sup_{\theta\in\Theta_{0}}
\mathrm{Cov}_{\theta}\!\left(
\sum_{t_{1},t_{2}=1}^{T}\partial_{z}\ell_{it_{1}}\partial_{z}\ell_{it_{2}},
\sum_{\tau_{1},\tau_{2}=1}^{T}D_{z^{2}}\ell_{i\tau_{1}}D_{z^{2}}\ell_{i\tau_{2}}
\right) \\
&\quad\le
\frac{1}{T^{2}}
\sup_{\theta\in\Theta_{0}}
\sqrt{
\mathrm{Var}_{\theta}\!\left(
\sum_{t_{1},t_{2}=1}^{T}\partial_{z}\ell_{it_{1}}\partial_{z}\ell_{it_{2}}
\right)
\mathrm{Var}_{\theta}\!\left(
\sum_{\tau_{1},\tau_{2}=1}^{T}D_{z^{2}}\ell_{i\tau_{1}}D_{z^{2}}\ell_{i\tau_{2}}
\right)
}.
\end{align*}

Since
\[
\mathrm{Var}_{\theta}\!\left(
\sum_{t_{1},t_{2}=1}^{T}\partial_{z}\ell_{it_{1}}\partial_{z}\ell_{it_{2}}
\right)
=
\mathbb{E}_{\theta}\!\left[
\left(\sum_{t=1}^{T}\partial_{z}\ell_{it}\right)^{4}
\right]
-
\left(
\mathbb{E}_{\theta}\!\left[
\left(\sum_{t=1}^{T}\partial_{z}\ell_{it}\right)^{2}
\right]
\right)^{2},
\]
and the second term is \(O(T^{2})\) by the previous result. Then it follows from a Rosenthal-type moment
inequality for $\alpha$-mixing sequences (Lemma S.2 in \citet{FVWeidner2016})  that
\[
\sup_{\theta\in\Theta_{0}}
\mathbb{E}_{\theta}\!\left[
\left(\sum_{t=1}^{T}\partial_{z}\ell_{it}\right)^{4}
\right]
=O(T^{2}).
\]
Hence,
\[
\sup_{\theta\in\Theta_{0}}
\mathbb{E}_{\theta}\left[
\frac{1}{T^{2}}
\sum_{t_{1},\tau_{1},t_{2},\tau_{2}=1}^{T}
\partial_{z}\ell_{it_{1}}\partial_{z}\ell_{it_{2}}
D_{z^{2}}\ell_{i\tau_{1}}D_{z^{2}}\ell_{i\tau_{2}}
\right]
=O(1).
\]
which is uniformly bounded for every $i$ by Assumption \ref{assu:correct} and \ref{assu:smooth}. Hence, we obtain
\[
\sup_{\theta \in \Theta_0} \mathrm{Var}(\frac{1}{N}\sum_i^N B_{i}^{(1,1)}) = O_p(1).
\]
Since $Y_i$ and $X_i$ are independent across $i$,  it follows from Chebyshev's inequality that, for any $\epsilon > 0$,
\[
\sup_{\theta \in \Theta_0} 
\mathbb{P}_\theta\!\left(
\left| 
\frac{1}{N} \sum_{i=1}^{N} (B_{i}^{(1,1)} - \mathbb{E}_\theta(B_{i}^{(1,1)})) 
\right| > \epsilon
\right) = o(1),
\]

And similarly, we have 
 $$\quad
\sup_{\theta \in \Theta_0} 
\mathbb{P}_\theta\!\left(
\left| 
\frac{1}{T} \sum_{i=1}^{T} (B_{i}^{(1,2)} - \mathbb{E}_\theta(B_{i}^{(1,2)})) 
\right| > \epsilon
\right) = o(1).$$
 Following a similar argument, the $B^{(2)}$ also converges uniformly in probability to its expectation, which is the last two terms in $U^{(1)}$. 

\end{proof}








The next theorem focuses on the APEs defined in (\ref{APE}), which is an extension of Theorem~4.2 in \citet{ChenFVWeidner2021}.

Define $\pi_{it} = \alpha_i^\prime \gamma_t$, and let $\partial_{\pi^q} \mu_{it} = \partial^q \mu_{it}( \beta, \pi_{it}) / \partial \pi_{it}^q$, since by Assumption~\ref{assu:APEs interactive}, the function $\mu(\cdot)$ depends on the unobserved effects only through the interaction term $\alpha_i^\prime \gamma_t$.

To simplify the expression, it is convenient to introduce some new notations. 
Let $\alpha_{i}^*$ and $\gamma_{t}^*$ are both $R$-dimensional vectors defined by
\[
(\alpha_{i}^\mu,\gamma_{t}^\mu) 
\in \argmin_{\alpha_{i},\,\gamma_{t}} 
\sum_{i,t} \mathbb{E}_\theta\!\left[-\partial_{z^2}\ell_{it}\right]
\left(
\frac{\mathbb{E}_\theta\!\left[\partial_{\pi}\mu_{it}\right]}{\mathbb{E}_\theta\!\left[\partial_{z^2}\ell_{it}\right]} 
- \alpha_{i}^\prime\gamma_{t0} - \alpha_{i0}^\prime\gamma_{t}
\right)^{\!2}.
\]
and $\Psi_{it} = \alpha_{i}^{\mu\prime}\gamma_{t0} + \alpha_{i0}^\prime\gamma_{t}^\mu$. Define $d_x\mu_{it}=\partial_{\beta}\mu_{it}-\partial_\pi\mu_{it}\varXi_{it}$, $D_{\pi^q}\mu_{it}=\partial_{\pi^q}\mu_{it}-\partial_{z^{q+1}}\ell_{it}\Psi_{it}$.

\begin{theorem}\label{Thm:C2}
Suppose Assumptions \ref{assu:parameterspace}-\ref{assu:convexity},\ref{assu:APEs interactive}-\ref{assu:APEs smooth} hold and for every $\theta\in\Theta_0$, $\overline{W}_{\infty}$ exist and is invertible. Then, for the APEs defined in (\ref{APE}), the following asymptotic expansion holds uniformly:
\[
\sqrt{NT}(\varDelta(\hat{\theta}) - \varDelta(\theta)) = U_{\varDelta}^{(0)} + U_{\varDelta}^{(1)} + o_{p}^{u}(1),
\]
where
\[
U_{\varDelta}^{(0)}=\frac{1}{\sqrt{NT}}\sum_{i,t=1}^{N,T}\left(\left[\frac{1}{NT}\sum_{i,t=1}^{N,T}\mathbb{E}_{\theta}\bigl(D_{\beta}\mu_{it}\bigr)\right]^{\prime}\overline{W}_{\infty}^{-1}\partial_{z}\ell_{it}\widetilde{X}_{it}-\Psi_{it}\partial_{z}\ell_{it}\right)
\]
and
\begin{align*}
U_{\varDelta}^{(1)}
=\,& \left[\frac{1}{NT} \sum_{i,t=1}^{N,T} \mathbb{E}_{\theta} \bigl(D_{\beta} \mu_{it}\bigr)\right]^{\prime} \overline{W}_{\infty}^{-1} U^{(1)} 
\\
&\quad -\frac{1}{\sqrt{NT}} \sum_{i=1}^{N} \sum_{t=1}^{T}
\Biggl\{
\mathbb{E}_{\theta}\!\left[
\gamma_{t}^{\!\prime}
\left(
\sum_{h=1}^{T} \gamma_{h}\gamma_{h}^{\!\prime} \mathbb{E}_{\theta}\!\bigl[\partial_{z^2} \ell_{ih}\bigr]
\right)^{-1}
\gamma_{t}
\left(
\partial_{z} \ell_{it}\, D_{\pi} \varDelta_{it} + \frac{1}{2} D_{\pi^2} \varDelta_{it}
\right)
\right]
\\
&\quad\qquad
+\mathbb{E}_{\theta}\!\left[\alpha_{i}^{\prime}\left(\sum_{h=1}^{N}\alpha_{h}\alpha_{h}^{\prime}\mathbb{E}_{\theta}\!\bigl[\partial_{z^2}\ell_{hj}\bigr]\right)^{-1}\alpha_{h}\left(\partial_{z}\ell_{it}\,D_{\pi}\varDelta_{it}+\frac{1}{2}D_{\pi^{2}}\varDelta_{it}\right)\right]\Biggr\}.
\end{align*}

\begin{proof}
 $ U^{(1)}, \overline{W}_\infty$ are defined in \ref{thm:C1:}.  See the proof of Theorem~2 in \citet{ChenFVWeidner2021}, with the modification that the bias term converges to its expectation uniformly. This can be established by applying a similar argument as in the proof of Theorem \ref{thm:C1:}.

\end{proof}
\end{theorem}

\section{Uniform results for general likelihood models with nuisance parameters}\label{Appendix C}

\subsection{Uniform version of regularity conditions}

This appendix presents a uniform version of Assumption B.1 in \citet{FVWeidner2016}, which we use to derive a uniform asymptotic expansion for the MLE. As in \citet{FVWeidner2016}, who do not employ the panel data structure of the model nor the particular form of the likelihood, we do not employ the network structure of the model nor the form of the likelihood. Therefore, our notation will be standard panel data notation, with $N$ and $T$ denoting the cross-sectional and time-series dimensions (accommodating panel and network models), and $\mathcal{L}(\beta,\phi)$ denoting a general objective function where $\dim\beta$ is fixed and $\dim\phi$ grows with $N$ or $T$, or both. The following assumption imposes high-level regularity conditions on $\mathcal{L}(\beta,\phi)$, its derivatives, maximizers, etc., which we denote as before, e.g., 
\( \mathcal{S}(\beta,\phi) = \partial_{\phi}\mathcal{L}(\beta,\phi) \), \( \mathcal{H}(\beta,\phi) = -\partial_{\phi \phi^{\prime}} \mathcal{L}(\beta, \phi) \), and so on. Recall that we also denote $\theta = (\beta, \phi)$. We use a bar to denote expectations conditional on $\phi$, e.g., $\partial_{\beta}\overline{\mathcal{L}} = \mathbb{E}_{\phi}[\partial_{\beta}\mathcal{L}]$, and a tilde to denote deviations from these expectations, e.g., $\partial_{\beta}\widetilde{\mathcal{L}} = \partial_{\beta}\mathcal{L} - \partial_{\beta}\overline{\mathcal{L}}$.

We adopt the same notation as in \citet{FVWeidner2016}. For a vector \(x\in\mathbb R^m\), define \(\|x\|_q=(\sum_{j=1}^m |x_j|^q)^{1/q}\). For a matrix \(A=(a_{ij})\in\mathbb R^{m\times n}\), let \(\|A\|_{q}
:= \sup_{x \in \mathbb{R}^{n},\, x \neq 0}
\frac{\|Ax\|_{q}}{\|x\|_{q}}.\), \(\|A\|=\|A\|_2\), \(\|A\|_{\max}=\max_{i,j}|a_{ij}|\), and \(\|A\|_\infty=\max_{1\le i\le m}\sum_{j=1}^n |a_{ij}|\). For higher-order derivatives, we use the analogous induced operator norm. For example,
\[
\|\partial_{\phi\phi\phi}\mathcal L(\beta,\phi)\|_q=\sup_{\|u\|_q,\|v\|_q,\|w\|_q\le 1}\bigl|\partial_{\phi\phi\phi}\mathcal L(\beta,\phi)[u,v,w]\bigr|,
\]
and similarly for \(\partial_{\beta\phi\phi}\mathcal L(\beta,\phi)\).
\[
\|\partial_{\beta\phi\phi}\mathcal L(\beta,\phi)\|_q
=
\sup_{\|u\|\le 1,\|v\|_q\le 1,\|w\|_q\le 1}
\bigl|\partial_{\beta\phi\phi}\mathcal L(\beta,\phi)[u,v,w]\bigr|.
\]
For further discussion of the relations among these norms and their use in asymptotic expansions with incidental parameters, see \citet{FVWeidner2016}. 

The proofs in this section largely overlap with the theory in \citet{FVWeidner2016}, but we include them here for convenience.

\begin{assumption}[Uniform version of regularity conditions]
\label{assu:uniform regularity} Fix some $q>4$ and $0\leq\epsilon<$$1/8-1/(2q)$.
Fix some $r_{\beta}=r_{\beta,N}>0$ and $r_{\phi}=r_{\phi,N}>0$, with $r_{\beta}=o((NT)^{-1/(2q)-\epsilon})$
and $r_{\phi}=o((NT)^{-\epsilon})$. For every $\theta=(\beta,\phi)\in\Theta_{0}$, the following conditions hold for the chosen values of $q,\epsilon,r_\beta$, and $r_\phi$. \newline
(i) $\frac{\operatorname{dim}\phi}{\sqrt{NT}}\rightarrow a$, where $0<a<\infty$. \newline
(ii) $\mathcal{L}(\widetilde{\beta},\widetilde{\phi})$
is four times continuously differentiable on $\mathcal{B}(r_{\beta},\beta)\times\mathcal{B}_{q}(r_{\phi},\phi)$ 
wpa1. \newline
(iii) $\sup_{\widetilde{\beta}\in\mathcal{B}(r_{\beta},\beta)}\Vert \widehat{\phi}(\widetilde{\beta})-\phi\Vert _{q}=o_{p}^{u}(r_{\phi})$. \newline
(iv) $\inf_{\theta\in\Theta_{0}}\overline{\mathcal{H}}>0$ and 
$\Vert \overline{\mathcal{H}}^{-1}\Vert _{q}=\mathcal{O}_{p}^{u}(1)$. \newline
(v) For the $q$-norm,
\begin{align*}
& \|\mathcal{S}\|_{q}=\mathcal{O}_{p}^{u}((NT)^{-1/4+1/(2q)}) , \quad \hfill\hfill\Vert \partial_{\beta}\mathcal{L}\Vert =\mathcal{O}_{p}^{u}(1), \quad
\|\widetilde{\mathcal{H}}\|_{q}=o_{p}^{u}(1), \\
& \hfill\hfill\Vert \partial_{\beta\phi^{\prime}}\mathcal{L}\Vert _{q}=\mathcal{O}_{p}^{u}((NT)^{1/(2q)}), \quad
\Vert \partial_{\beta\beta^{\prime}}\mathcal{L}\Vert =\mathcal{O}_{p}^{u}(\sqrt{NT}), \\
& \hfill\hfill\Vert \partial_{\beta\phi\phi}\mathcal{L}\Vert _{q}=\mathcal{O}_{p}^{u}((NT)^{\epsilon}), \quad
  \Vert \partial_{\phi\phi\phi}\mathcal{L}\Vert _{q} =\mathcal{O}_{p}^{u}((NT)^{\epsilon}),
\end{align*}
and
\begin{align*}
&\sup_{\widetilde{\beta}\in\mathcal{B}(r_{\beta},\beta)}\sup_{\widetilde{\phi}\in\mathcal{B}_{q}(r_{\phi},\phi)}\Vert \partial_{\beta\beta\beta}\mathcal{L}(\widetilde{\beta},\widetilde{\phi})\Vert =\mathcal{O}_{p}^{u}(\sqrt{NT}),\\
&\sup_{\widetilde{\beta}\in\mathcal{B}(r_{\beta},\beta)}\sup_{\widetilde{\phi}\in\mathcal{B}_{q}(r_{\phi},\phi)}\Vert \partial_{\beta\beta\phi}\mathcal{L}(\widetilde{\beta},\widetilde{\phi})\Vert _{q}=\mathcal{O}_{p}^{u}((NT)^{1/(2q)}),\\
&\sup_{\widetilde{\beta}\in\mathcal{B}(r_{\beta},\beta)}\sup_{\widetilde{\phi}\in\mathcal{B}_{q}(r_{\phi},\phi)}\Vert \partial_{\beta\beta\phi\phi}\mathcal{L}(\widetilde{\beta},\widetilde{\phi})\Vert _{q}=\mathcal{O}_{p}^{u}((NT)^{\epsilon}),\\
&\sup_{\widetilde{\beta}\in\mathcal{B}(r_{\beta},\beta)}\sup_{\widetilde{\phi}\in\mathcal{B}_{q}(r_{\phi},\phi)}\Vert \partial_{\beta\phi\phi\phi}\mathcal{L}(\widetilde{\beta},\widetilde{\phi})\Vert _{q}=\mathcal{O}_{p}^{u}((NT)^{\epsilon}),\\
&\sup_{\widetilde{\beta}\in\mathcal{B}(r_{\beta},\beta)}\sup_{\widetilde{\phi}\in\mathcal{B}_{q}(r_{\phi},\phi)}\Vert \partial_{\phi\phi\phi\phi}\mathcal{L}(\widetilde{\beta},\widetilde{\phi})\Vert _{q}=\mathcal{O}_{p}^{u}((NT)^{\epsilon}).
\end{align*}
(vi) For the spectral norm $\|\cdot\|$, 
\begin{align*}
&\|\widetilde{\mathcal{H}}\|  =o_{p}^{u}((NT)^{-1/8}), \quad
\Vert \partial_{\beta\beta^{\prime}}\widetilde{\mathcal{L}}\Vert   =o_{p}^{u}(\sqrt{NT}), \quad
\Vert \partial_{\beta\phi\phi}\widetilde{\mathcal{L}}\Vert   =o_{p}^{u}((NT)^{-1/8}),\\
&\Vert \partial_{\beta\phi}\widetilde{\mathcal{L}}\Vert   =\mathcal{O}_{p}^{u}(1), \quad \left\Vert \sum_{g,h=1}^{\dim\phi}\partial_{\phi\phi_{g}\phi_{h}}\widetilde{\mathcal{L}}\left(\overline{\mathcal{H}}^{-1}\mathcal{S}\right)_{g}
\left(\overline{\mathcal{H}}^{-1}\mathcal{S}\right)_{h}\right\Vert   =o_{p}^{u}((NT)^{-1/4}).
\end{align*}
\end{assumption}

\begin{lemma}[Uniform bounds]
\label{lem:D1} Suppose Assumption \ref{assu:uniform regularity}
holds. Then, for all $\theta\in\Theta_0$, we have \newline
(i) $\mathcal{H}(\widetilde{\beta},\widetilde{\phi})>0$ for all $\beta\in\mathcal{B}(r_{\beta},\beta)$
and $\phi\in\mathcal{B}_{q}(r_{\phi},\phi)$ wpa1 uniformly, and 
\[
\begin{aligned}
&\sup_{\widetilde{\beta}\in\mathcal{B}(r_{\beta},\beta)}\sup_{\widetilde{\phi}\in\mathcal{B}_{q}(r_{\phi},\phi)}\Vert\partial_{\beta\beta^{\prime}}\mathcal{L}(\widetilde{\beta},\widetilde{\phi})\Vert  =\mathcal{O}_{p}^{u}(\sqrt{NT}),\\
&\sup_{\widetilde{\beta}\in\mathcal{B}(r_{\beta},\beta)}\sup_{\widetilde{\phi}\in\mathcal{B}_{q}(r_{\phi},\phi)}\Vert\partial_{\beta\phi^{\prime}}\mathcal{L}(\widetilde{\beta},\widetilde{\phi})\Vert_{q}  =\mathcal{O}_{p}^{u}((NT)^{1/(2q)}),\\
&\sup_{\widetilde{\beta}\in\mathcal{B}(r_{\beta},\beta)}\sup_{\widetilde{\phi}\in\mathcal{B}_{q}(r_{\phi},\phi)}\Vert\partial_{\phi\phi\phi}\mathcal{L}(\widetilde{\beta},\widetilde{\phi})\Vert_{q}  =\mathcal{O}_{p}^{u}((NT)^{\epsilon}),\\
&\sup_{\widetilde{\beta}\in\mathcal{B}(r_{\beta},\beta)}\sup_{\widetilde{\phi}\in\mathcal{B}_{q}(r_{\phi},\phi)}\Vert\partial_{\beta\phi\phi}\mathcal{L}(\widetilde{\beta},\widetilde{\phi})\Vert_{q}  =\mathcal{O}_{p}^{u}((NT)^{\epsilon}),\\
&\sup_{\widetilde{\beta}\in\mathcal{B}(r_{\beta},\beta)}\sup_{\widetilde{\phi}\in\mathcal{B}_{q}(r_{\phi},\phi)}\Vert\mathcal{H}^{-1}(\widetilde{\beta},\widetilde{\phi})\Vert_{q}  =\mathcal{O}_{p}^{u}(1).
\end{aligned}
\]
(ii) Furthermore, 
\begin{align*}
&\|\mathcal{S}\| =\mathcal{O}_{p}^{u}(1),
\quad \Vert\mathcal{H}^{-1}\Vert  =\mathcal{O}_{p}^{u}(1),
\quad \Vert\overline{\mathcal{H}}^{-1}\Vert  =\mathcal{O}_{p}^{u}(1) \\
& \Vert\mathcal{H}^{-1}-\overline{\mathcal{H}}^{-1}\Vert  =o_{p}^{u}((NT)^{-1/8}),
\quad \Vert\mathcal{H}^{-1}-(\overline{\mathcal{H}}^{-1}-\overline{\mathcal{H}}^{-1}\widetilde{\mathcal{H}}\overline{\mathcal{H}}^{-1})\Vert  =o_{p}^{u}((NT)^{-1/4}),\\
&\Vert\partial_{\beta\phi^{\prime}}\mathcal{L}\Vert  =\mathcal{O}_{p}^{u}((NT)^{1/4}),\quad
\Vert\partial_{\beta\phi\phi}\mathcal{L}\Vert  =\mathcal{O}_{p}^{u}((NT)^{\epsilon}),\\
&\left\Vert \sum_{g}\partial_{\phi\phi^{\prime}\phi_{g}}\mathcal{L}\left(\mathcal{H}^{-1}\mathcal{S}\right)_{g}\right\Vert  =\mathcal{O}_{p}^{u}((NT)^{-1/4+1/(2q)+\epsilon}),\\
&\left\Vert \sum_{g}\partial_{\phi\phi^{\prime}\phi_{g}}\mathcal{L}\left(\overline{\mathcal{H}}^{-1}\mathcal{S}\right)_{g}\right\Vert  =\mathcal{O}_{p}^{u}((NT)^{-1/4+1/(2q)+\epsilon}).
\end{align*}
\end{lemma}
\begin{proof}
See the proof of Lemma~S.1 on page 7 of the supplementary material of \citet{FVWeidner2016}, with slight modifications to extend the argument to a uniform bound.

\end{proof}

\begin{theorem}[Uniform asymptotic expansion]
\label{thm:D1}
Suppose Assumption \ref{assu:uniform regularity} holds with $q,\epsilon,r_\beta$, and $r_\phi$  satisfying the conditions of Assumption \ref{assu:uniform regularity}. Then, for every
$\theta \in \Theta$ and $\widetilde{\beta} \in \mathcal{B}(r_{\beta}, \beta)$, it holds that
\begin{align*}
\hat{\phi}(\widetilde{\beta}) - \phi &= \mathcal{H}^{-1} \mathcal{S} + \mathcal{H}^{-1} [\partial_{\phi \beta^{\prime}} \mathcal{L}] (\widetilde{\beta} - \beta) + \frac{1}{2} \mathcal{H}^{-1} \sum_{g=1}^{\dim\phi} [\partial_{\phi \phi \phi_{g}} \mathcal{L}] \mathcal{H}^{-1} \mathcal{S} [\mathcal{H}^{-1} \mathcal{S}]_{g} + \mathbb{R}^{\phi}(\widetilde{\beta}), \\
\partial_{\beta} \mathcal{L} (\widetilde{\beta}, \hat{\phi}(\widetilde{\beta})) &= U - \overline{W} \sqrt{NT} (\widetilde{\beta} - \beta) + R(\widetilde{\beta}),
\end{align*}
where $U=U^{(0)}+U^{(1)}$ and
\begin{align*}
\overline{W} & := -\frac{1}{\sqrt{NT} } \left( \partial_{\beta \beta^{\prime}} \overline{\mathcal{L}} + \left[ \partial_{\beta \phi^{\prime}} \overline{\mathcal{L}} \right] \overline{\mathcal{H}}^{-1} \left[ \partial_{\phi \beta^{\prime}} \overline{\mathcal{L}} \right] \right)
, \\
U^{(0)} & := \partial_{\beta} \mathcal{L} + \left[ \partial_{\beta \phi^{\prime}} \overline{\mathcal{L}} \right] \overline{\mathcal{H}}^{-1} \mathcal{S}, \\
U^{(1)}&:=\left[\partial_{\beta \phi^{\prime}} \widetilde{\mathcal{L}}\right] \overline{\mathcal{H}}^{-1} \mathcal{S}-\left[\partial_{\beta \phi^{\prime}} \overline{\mathcal{L}}\right] \overline{\mathcal{H}}^{-1} \widetilde{\mathcal{H}} \overline{\mathcal{H}}^{-1} \mathcal{S}\\+&\frac{1}{2} \sum_{a=1}^{\operatorname{dim} \phi}\left(\partial_{\beta \phi^{\prime} \phi_g} \overline{\mathcal{L}}+\left[\partial_{\beta \phi^{\prime}} \overline{\mathcal{L}}\right] \overline{\mathcal{H}}^{-1}\left[\partial_{\phi \phi^{\prime} \phi_g} \overline{\mathcal{L}}\right]\right)\left[\overline{\mathcal{H}}^{-1} \mathcal{S}\right]_g \overline{\mathcal{H}}^{-1} \mathcal{S} .  \\
\sup_{\widetilde{\beta} \in \mathcal{B}(r_{\beta}, \beta)}\Vert R(\widetilde{\beta} ) \Vert & = o_{p}^{u}(1), \quad
\sup_{\widetilde{\beta} \in \mathcal{B}(r_{\beta}, \beta)}\Vert \mathbb{R}^{\phi}(\widetilde{\beta} ) \Vert_q = o_{p}^{u}(1).
\end{align*}
\end{theorem}
\begin{proof}
This theorem is a uniform version of Theorem B.1 in \citet{FVWeidner2016}. Given Assumption
 \ref{assu:uniform regularity}, the uniform asymptotic expansion holds by the same arguments as in their proof.
\end{proof}

The next theorem extends Theorem~B.3 in \citet{FVWeidner2016} to the uniform case. This extension is useful because it shows that, if the objective function is strictly concave, no additional proof is required to establish uniform consistency, as is the case for logit models with additive fixed effects. 

\begin{theorem}[Uniform consistency]
\label{uniform consistency}
\label{Thm C.1} Suppose Assumptions \ref{assu:uniform regularity}(i)--(ii) and \ref{assu:uniform regularity}(iv)--(vi) hold. Suppose $\mathcal{L}(\beta,\phi)$ is strictly concave over $(\beta,\phi)$ with probability $1$. Let $q>4$ and let $r_\phi$ and 
$r_\beta$ satisfy $N^{-1/2+1/q}=o_{p}(r_{\phi})$
and $N^{1/q}r_{\beta}=o_{p}(r_{\phi})$. Then, (i) 
\[
\sup_{\widetilde{\beta}\in\mathcal{B}(r_{\beta},\beta)}\Vert \hat{\phi}(\widetilde{\beta})-\phi\Vert _{q}=o_{p}^{u}(r_{\phi})
\]
and (ii) if $\sup_{\theta\in\Theta_0}\|U(\theta)\|=\mathcal{O}_P(1)$ and the limit $\overline{W}_\infty(\theta):=\lim_{N,T\to\infty}\overline{W}_{N,T}(\theta)$ exists with  $\inf_{\theta\in\Theta_0}\lambda_{\min}(\overline{W}_\infty(\theta))>0$, then
\[
\Vert \hat{\beta}-\beta\Vert =\mathcal{O}_{p}^{u}((NT)^{-1/4}).
\]

\begin{proof}
The proof of part (i) is identical to the proof of Theorem B.3 in \citet{FVWeidner2016} except that the modification accounts for uniform convergence.\footnote{Note: Part (iii) of Assumption \ref{assu:uniform regularity} is not required for the expansion and order of every term in the proof of part $(i)$. Other regularity conditions can be directly checked by Assumptions \ref{assu:parameterspace}-\ref{assu:convexity}.}

Part $(ii)$: Since part (i) already show that 
$\sup_{\widetilde{\beta} \in \mathcal{B}(r_{\beta}, \beta)} \Vert \widehat{\phi}(\widetilde{\beta}) - \phi \Vert_q = o_p^u(\eta)$ which is the $(iii)$ of the regularity conditions and hence the uniform asymptotic expansion in Theorem \ref{thm:D1} holds. For simplicity, we first consider the case $\dim \beta = 1$; the case $\dim \beta > 1$ follows by an analogous argument in the proof of Theorem B.3 in \citet{FVWeidner2016}.

Let \( \eta = 2{NT}^{-\frac{1}{2}} \overline{W}^{-1} | U| \), and it can be shown that 
\[
\partial_{\beta}\mathcal{L}(\beta+\eta,\widehat{\phi}(\beta+\eta))=U-\overline{W}\sqrt{NT}\eta+R(\beta+\eta).
\]
by Theorem \ref{thm:D1}.

 Given the definition of $\eta$, we have $U-\overline{W}\sqrt{NT}\eta=U-2| U|\leq-| U|$
and $| R(\beta_{0}+\eta)| =o_{p}^{u}(1)+o_{p}^{u}(\sqrt{NT}\eta)$.
Since $|\eta|=\mathcal{O}_{p}^{u}({NT}^{-\frac{1}{2}})$,
\[
\partial_{\beta}\mathcal{L}(\beta+\eta,\widehat{\phi}(\beta_{0}+\eta))\leq o_{p}^{u}(1)+(o_{p}^{u}(1)-1)| U |
\]
we have
\[
\inf_{\theta\in\Theta}\mathbb{P}_{\theta}(\partial_{\beta}\mathcal{L}(\beta+\eta,\widehat{\phi}\left(\beta+\eta\right))\leq 0)\rightarrow1
\]
Similarly, we have
\[
\partial_{\beta}\mathcal{L}\left(\beta-\eta,\widehat{\phi}(\beta-\eta\right))\geq o_{p}^{u}(1)+(o_{p}^{u}(1)+1)| U |
\]
And
\[
\inf_{\theta_\in\Theta_0}\mathbb{P}_{\theta}\left(\partial_{\beta}\mathcal{L}\left(\beta-\eta,\widehat{\phi}\left(\beta-\eta\right)\right)\geq0\right)\rightarrow 1 
\]
\textcolor{black}{holds for sufficient large }$N,T$.

\textcolor{black}{Define} $E_{1}=\left\{ \partial_{\beta}\mathcal{L}\left(\beta+\eta,\widehat{\phi}\left(\beta+\eta\right)\right)\leq0\right\} $
and $E_{2}=\left\{ \partial_{\beta}\mathcal{L}\left(\beta-\eta,\widehat{\phi}\left(\beta-\eta\right)\right)\geq0\right\}$, then we have
\begin{align*}
 \inf_{\theta\in\Theta_0}\mathbb{P}_{\theta}\left(E_{1}\cap E_{2}\right) 
 & \le \inf_{\theta\in\Theta_0}\left(1-\mathbb{P}_{\theta}(E_{1}^{c})-\mathbb{P}_{\theta}(E_{2}^{c})\right)\\
& = 1-\sup_{\theta\in\Theta_0}\mathbb{P}_{\theta}(E_{1}^{c})-\sup_{\theta\in\Theta_0}\mathbb{P}_{\theta}(E_{2}^{c})\\
& \rightarrow  1
\end{align*}
Since $\mathcal{L}$ is strictly concave and differentiable, it means
\begin{align*}
 \inf_{\theta\in\Theta_0}\mathbb{P}_{\theta}\left(\beta-\eta\leq\widehat{\beta}\leq\beta_{0}+\eta\right)
& = \inf_{\theta\in\Theta_0}\mathbb{P}_{\theta}\left(\left|\widehat{\beta}-\beta\right|\leq\eta\right) \rightarrow 1.
\end{align*}
\end{proof}
\end{theorem}

\subsection{Uniform results for APEs}
The next assumption is from assumption $B.2$ in \citet{FVWeidner2016} but extends to the uniform case.
\begin{assumption}[Regularity conditions for asymptotic expansion of $\widehat{\varDelta}$]\label{assu:B2}
Let $q, \epsilon, r_\beta$, and $r_\phi$ be defined as in Assumption~\ref{assu:uniform regularity}. We assume that:

\begin{enumerate}
    \item $(\widetilde{\beta}, \widetilde{\phi}) \mapsto \Delta(\widetilde{\beta}, \widetilde{\phi})$ is three times continuously differentiable in 
    $\mathcal{B}\left(r_\beta, \beta\right) \times \mathcal{B}_q\left(r_\phi, \phi\right)$, wpa1.

    \item The partial derivatives satisfy
    \[
    \|\partial_\beta \mu\| = \mathcal{O}_p^u(1), \quad
    \|\partial_\phi \mu\|_q = \mathcal{O}_p^u\bigl((NT)^{1/(2q)-1/2}\bigr), \quad
    \|\partial_{\phi \phi} \mu\|_q = \mathcal{O}_p^u\bigl((NT)^{\epsilon-1/2}\bigr),
    \]
    and
    \[
    \begin{aligned}
    \sup_{\widetilde{\beta} \in \mathcal{B}(r_\beta, \beta)} \sup_{\widetilde{\phi} \in \mathcal{B}_q(r_\phi, \phi)} \|\partial_{\beta \beta} \mu(\widetilde{\beta}, \widetilde{\phi})\| & = \mathcal{O}_p^u(1), \\
    \sup_{\widetilde{\beta} \in \mathcal{B}(r_\beta, \beta)} \sup_{\widetilde{\phi} \in \mathcal{B}_q(r_\phi, \phi)} \|\partial_{\beta \phi'} \mu(\widetilde{\beta}, \widetilde{\phi})\|_q & = \mathcal{O}_p^u\bigl((NT)^{1/(2q)-1/2}\bigr), \\
    \sup_{\widetilde{\beta} \in \mathcal{B}(r_\beta, \beta)} \sup_{\widetilde{\phi} \in \mathcal{B}_q(r_\phi, \phi)} \|\partial_{\phi \phi \phi} \mu(\widetilde{\beta}, \widetilde{\phi})\|_q & = \mathcal{O}_p^u\bigl((NT)^{\epsilon-1/2}\bigr).
    \end{aligned}
    \]

    \item The partial derivatives of $\widetilde{\mu}$ satisfy
    \[
    \|\partial_\beta \widetilde{\mu}\| = o_P^u(1), \quad
    \|\partial_\phi \widetilde{\mu}\| = \mathcal{O}_p^u\bigl((NT)^{-1/2}\bigr), \quad
    \|\partial_{\phi \phi} \widetilde{\mu}\| = o_P^u\bigl((NT)^{-5/8}\bigr).
    \]
\end{enumerate}
\end{assumption}

The next theorem is a uniform extension of Theorem~B.4 in \citet{FVWeidner2016}.

\begin{theorem}\label{thm.B3}
Suppose Assumption \ref{assu:uniform regularity} and \ref{assu:B2} holds and \(\left\Vert \hat{\beta}-\beta\right\Vert =\mathcal{O}_{p}^{u}\left((NT)^{-\frac{1}{2}}\right)\). Then,
$$
\sqrt{NT}(\varDelta(\hat{\theta}) - \varDelta(\theta))=\left(\partial_{\beta}\overline{\mu}+\left[\partial_{\phi}\overline{\mu}\right]^{\prime}\mathcal{\overline{H}}^{-1}\partial_{\phi\beta^{\prime}}\overline{\mathcal{L}}\right)(\hat{\beta}-\beta)
  +U_{\mu}^{(0)}+U_{\mu}^{(1)}+o_{p}^{u}((NT)^{-\frac{1}{2}}),
$$
where
\begin{align*}
U_{\mu}^{(0)} =  &\left[\partial_{\phi}\overline{\mu}\right]^{\prime}\overline{\mathcal{H}}^{-1}\mathcal{S},\\
U_{\mu}^{(1)} = &\left[\partial_{\phi}\widetilde{\mu}\right]^{\prime}\overline{\mathcal{H}}^{-1}\mathcal{S}-\left[\partial_{\phi}\overline{\mu}\right]^{\prime}\overline{\mathcal{H}}^{-1}\widetilde{\mathcal{H}}\overline{\mathcal{H}}^{-1}\mathcal{S}\\
 & +\frac{1}{2}\mathcal{S}^{\prime}\mathcal{\overline{H}}^{-1}\left[\partial_{\phi\phi^{\prime}}\overline{\mu}+\sum_{g=1}^{\dim\phi}\partial_{\phi\phi\phi_{g}}\overline{\mathcal{L}}\left[\mathcal{H}^{-1}\partial_{\phi}\overline{\mu}\right]{}_{g}\right]\overline{\mathcal{H}}^{-1}\mathcal{S}.
\end{align*}

\begin{proof}
See the proof of Theorem~B.4 in \citet{FVWeidner2016}, with the slight modification that the bound for each term holds uniformly over every $\theta \in \Theta_0$.

\end{proof}
\end{theorem}

\subsection{Sufficient conditions for regularity conditions }
The following lemma is a uniform version of Lemma S.7 in  \citet{FVWeidner2016}. It provides sufficient conditions for parts $(v)$--$(vi)$ of Assumption \ref{assu:uniform regularity} to hold in the panel data context.
\begin{lemma}
\label{lemma.D2}
Let $\mathcal{L}(\theta)=\mathcal{L}(\beta,\phi)=\mathcal{L}(\beta,\alpha,\gamma)=\frac{1}{\sqrt{NT}}[\sum_{i=1}^N\sum_{t=1}^T\ell_{it}(\beta,\pi_{it})-\frac{b}{2}(v^{\prime}\phi)^{2}]$,
where $\pi_{it}=\alpha_{i}+\gamma_{t}$, $\alpha=(\alpha_{1},\ldots,\alpha_{N})^{\prime}$,
$\gamma=(\gamma_{1},\ldots,\gamma_{T})$, and $v=(1_{N}^{\prime},1_{T}^{\prime})^{\prime}$. Assume
that $\ell_{it}(.,.)$ is four times continuously differentiable
on $\Theta_{0}$.
Consider limits as $N,T\rightarrow\infty$ with $N/T\rightarrow\kappa>0$.
Let $4<q\leq8$ and $0\leq\epsilon<1/8-1/(2q)$. Let $r_{\beta}=r_{\beta,NT}>0$ and $r_{\phi}=r_{\phi,NT}>0$,
with $r_{\beta}=o((NT)^{-1/(2q)-\epsilon})$ and $r_{\phi}=o((NT)^{-\epsilon})$.
Assume the following conditions hold.
\begin{enumerate}
\item[(i)] For $k,l,m\in\{1,2,\ldots,\operatorname{dim}\beta\}$,
\begin{align*}
&\frac{1}{\sqrt{NT}}\sum_{i=1}^N\sum_{t=1}^T\partial_{\beta_{k}}\ell_{it}  =\mathcal{O}_{p}^{u}(1),\\
&\frac{1}{NT}\sum_{i=1}^N\sum_{t=1}^T\partial_{\beta_{k}\beta_{l}}\ell_{it}  =\mathcal{O}_{p}^{u}(1),\\
&\frac{1}{NT}\sum_{i=1}^N\sum_{t=1}^T\left\{ \partial_{\beta_{k}\beta_{l}}\ell_{it}-\mathbb{E}_{\phi}\left[\partial_{\beta_{k}\beta_{l}}\ell_{it}\right]\right\}   =o_{p}^{u}(1),\\
&\sup_{\widetilde{\beta}\in\mathcal{B}(r_{\beta},\beta)}\sup_{\widetilde{\phi}\in\mathcal{B}_{q}(r_{\phi},\phi)}\frac{1}{NT}\sum_{i=1}^N\sum_{t=1}^T\partial_{\beta_{k}\beta_{l}\beta_{m}}\ell_{it}\left(\widetilde{\beta},\widetilde{\pi}_{it}\right) =\mathcal{O}_{p}^{u}(1).
\end{align*}
 \item[(ii)] Let $k,l\in\{1,2,\ldots,\operatorname{dim}\beta\}$. For $\xi_{it}(\widetilde{\beta},\widetilde{\phi})=\partial_{\beta_{k}\pi}\ell_{it}\left(\widetilde{\beta},\widetilde{\pi}_{it}\right)$
or $\xi_{it}(\widetilde{\beta},\widetilde{\phi})=\partial_{\beta_{k}\beta_{l}\pi}\ell_{it}\left(\widetilde{\beta},\widetilde{\pi}_{it}\right)$,
\[
\begin{aligned}\sup_{\widetilde{\beta}\in\mathcal{B}(r_{\beta},\beta)}\sup_{\widetilde{\phi}\in\mathcal{B}_{q}(r_{\phi},\phi)}\frac{1}{T}\sum_{t=1}^T\left|\frac{1}{N}\sum_{i}\xi_{it}(\widetilde{\beta},\widetilde{\phi})\right|^{q} & =\mathcal{O}_{p}^{u}(1),\\
\sup_{\widetilde{\beta}\in\mathcal{B}(r_{\beta},\beta)}\sup_{\widetilde{\phi}\in\mathcal{B}_{q}(r_{\phi},\phi)}\frac{1}{N}\sum_{i=1}^N\left|\frac{1}{T}\sum_{t}\xi_{it}(\widetilde{\beta},\widetilde{\phi})\right|^{q} & =\mathcal{O}_{p}^{u}(1).
\end{aligned}
\]
\item[(iii)] Let $k,l\in\{1,2,\ldots,\operatorname{dim}\beta\}$. For $\xi_{it}(\beta,\phi)=\partial_{\pi^{r}}\ell_{it}\left(\widetilde{\beta},\widetilde{\pi}_{it}\right)$,
with $r\in\{3,4\}$, or $\xi_{it}(\widetilde{\beta},\widetilde{\phi})=$ $\partial_{\beta_{k}\pi^{r}}\ell_{it}\left(\widetilde{\beta},\widetilde{\pi}_{it}\right)$,
with $r\in\{2,3\}$, or $\xi_{it}(\beta,\phi)=\partial_{\beta_{k}\beta_{l}\pi^{2}}\ell_{it}\left(\beta,\pi_{it}\right)$,
\[
\begin{aligned}\sup_{\widetilde{\beta}\in\mathcal{B}(r_{\beta},\beta)}\sup_{\widetilde{\phi}\in\mathcal{B}_{q}(r_{\phi},\phi)}\max_{i}\frac{1}{T}\sum_{t=1}^T\left|\xi_{it}(\widetilde{\beta},\widetilde{\phi})\right| & =\mathcal{O}_{p}^{u}(N^{2\epsilon}),\\
\sup_{\widetilde{\beta}\in\mathcal{B}(r_{\beta},\beta)}\sup_{\widetilde{\phi}\in\mathcal{B}_{q}(r_{\phi},\phi)}\max_{t}\frac{1}{N}\sum_{i=1}^N\left|\xi_{it}(\widetilde{\beta},\widetilde{\phi})\right| & =\mathcal{O}_{p}^{u}(N^{2\epsilon}).
\end{aligned}
\]
 \item[(iv)] Moreover, 
\[
\begin{aligned} & \frac{1}{T}\sum_{t=1}^{T}\left|\frac{1}{\sqrt{N}}\sum_{i=1}^N\partial_{\pi}\ell_{it}\right|^{q}=\mathcal{O}_{p}^{u}(1), \\
 &\frac{1}{N}\sum_{i=1}^N\left|\frac{1}{\sqrt{T}}\sum_{t=1}^{T}\partial_{\pi}\ell_{it}\right|^{q}=\mathcal{O}_{p}^{u}(1),\\
 & \frac{1}{T}\sum_{t=1}^{T}\left|\frac{1}{\sqrt{N}}\sum_{i=1}^N\partial_{\beta_{k}\pi}\ell_{it}-\mathbb{E}_{\theta}\left[\partial_{\beta_{k}\pi}\ell_{it}\right]\right|^{2}=\mathcal{O}_{p}^{u}(1),\\
 & \frac{1}{N}\sum_{i=1}^N\left|\frac{1}{\sqrt{T}}\sum_{t=1}^{T}\partial_{\beta_{k}\pi}\ell_{it}-\mathbb{E}_{\theta}\left[\partial_{\beta_{k}\pi}\ell_{it}\right]\right|^{2}=\mathcal{O}_{p}^{u}(1).
\end{aligned}
\]
\item[(v)] The sequences $(\ell_{i1},\ldots,\ell_{iT})$ are conditionally independent across $i$, given $\phi$. 
\item[(vi)] Let $k\in\{1,2,\ldots,\operatorname{dim}\beta\}$.
For $\xi_{it}=\partial_{\pi^{r}}\ell_{it}-\mathbb{E}_{\phi}(\partial_{\pi^{r}}\ell_{it})$,
with $r\in\{2,3\}$, or $\xi_{it}=$ $\partial_{\beta_{k}\pi^{2}}\ell_{it}-\mathbb{E}_{\theta_{0}}(\partial_{\beta_{k}\pi^{2}}\ell_{it})$, we have, for some $\tilde{\nu}>0$ and some $C>0$,
\begin{align*}
&\sup_{\theta\in\Theta_{0}}\max_{i}\mathbb{E}_{\theta}\left[\xi_{it}^{8+\tilde{\nu}}\right]\leq C, \qquad \sup_{\theta\in\Theta_{0}}\max_{i,t}\sum_{s}\mathbb{E}_{\theta}\left[\xi_{it}\xi_{is}\right]\leq C,\\
&\sup_{\theta\in\Theta_{0}}\max_{i}\mathbb{E}_{\theta}\left\{ \left[\frac{1}{\sqrt{T}}\sum_{t=1}^{T}\xi_{it}\right]^{8}\right\} \leq C, \qquad \sup_{\theta\in\Theta_{0}}\max_{t}\mathbb{E}_{\theta}\left\{ \left[\frac{1}{\sqrt{N}}\sum_{i=1}^N\xi_{it}\right]^{8}\right\} \leq C,\\
& \sup_{\theta\in\Theta_{0}}\max_{i,j}\mathbb{E}_{\theta}\left\{ \left[\frac{1}{\sqrt{T}}\sum_{t=1}^{T}\left[\xi_{it}\xi_{jt}-\mathbb{E}_{\phi}\left(\xi_{it}\xi_{jt}\right)\right]\right]^{4}\right\} \leq C,
\end{align*}
uniformly in $N,T$.
\item[(vii)] $\sup_{\theta\in\Theta_{0}}\Vert \overline{\mathcal{H}}^{-1}\Vert _{q}=\mathcal{O}_{p}(1)$. 
\end{enumerate}
Then, parts $(v)$--$(vi)$ of Assumption \ref{assu:uniform regularity} are satisfied with the same
$q,\epsilon,r_{\beta},r_{\phi}$ as used here.
\end{lemma}

\begin{proof}
The proof follows the same steps as that of Lemma S.7 in \citet{FVWeidner2016}, except that $\mathcal{O}_{p}$ is replaced by $\mathcal{O}^u_{p}$. Moreover, Lemma $S.6$ in \citet{FVWeidner2016} holds uniformly for every $\theta \in \Theta_0$. 

\end{proof}

\subsection{Lemmas}
\begin{lemma} \label{propeties of O}
Generally, $\mathcal{O}_{p}^{u}(1)$ and $o_{p}^{u}(1)$ satisfy
\begin{align}
\mathcal{O}_{p}^{u}(1)+\mathcal{O}_{p}^{u}(1) & =\mathcal{O}_{p}^{u}(1), \label{Oup1}\\
\mathcal{O}_{p}^{u}(1)+o_{p}^{u}(1) & =\mathcal{O}_{p}^{u}(1), \label{Oup2}\\
o_{p}^{u}(1)+o_{p}^{u}(1) & =o_{p}^{u}(1),  \label{Oup3}\\
\mathcal{O}_{p}^{u}(1) o_{p}^{u}(1) & =o_{p}^{u}(1).  \label{Oup4}
\end{align}
\end{lemma}
\begin{proof}
Let $X_{n}$ and $Y_{n}$ be $\mathcal{O}_{p}^{u}(1)$. Then, 
for any $\epsilon>0$, there exist $0<k_{X},k_Y<\infty$ and  such that
\begin{align*}
\sup_{\theta\in\Theta_{0}}\mathbb{P}_{\theta}(|X_{n}|>k_{X})&<\frac{\epsilon}{2},\\
\sup_{\theta\in\Theta_{0}}\mathbb{P}_{\theta}(|Y_{n}|>k_{Y})&<\frac{\epsilon}{2}
\end{align*}
for all $n$. Hence
\begin{align*}
\sup_{\theta\in\Theta_{0}}\mathbb{P}_{\theta}(|X_{n}+Y_{n}|>k_{X}+k_{Y})
\leq & \sup_{\theta\in\Theta_{0}}\left(\mathbb{P}_{\theta}(|X_{n}|>k_{X})+\mathbb{P}_{\theta}(|Y_{n}|>k_{Y})\right)\\
\leq & \sup_{\theta\in\Theta_{0}}\mathbb{P}_{\theta}(|X_{n}|>k_{X})
+\sup_{\theta\in\Theta_{0}}\mathbb{P}_{\theta}(|Y_{n}|>k_{Y})\\
\leq & \,\,\epsilon
\end{align*}
for all $n$. This proves (\ref{Oup1}), from which (\ref{Oup2}) follows trivially.

Now, let $X_{n}$ and $Y_{n}$ be $o_{p}^{u}(1)$. Then, for any $\epsilon>0$ and $\delta>0$, we have
\begin{align*}
\sup_{\theta\in\Theta_0}\mathbb{P}_{\theta}\left(|X_{n}|>\frac{\delta}{2}\right)&<\frac{\epsilon}{2},\\
\sup_{\theta\in\Theta_{0}}\mathbb{P}_{\theta}\left(|Y_{n}|>\frac{\delta}{2}\right)&<\frac{\epsilon}{2}
\end{align*}
for all $n$. Hence
\begin{align*}
\sup_{\theta\in\Theta_{0}}\mathbb{P}_{\theta}(|X_{n}+Y_{n}|>\delta)
\leq & \sup_{\theta\in\Theta_{0}}\mathbb{P}_{\theta}\left(|X_{n}|>\frac{\delta}{2}\right)
+\sup_{\theta\in\Theta_{0}}\mathbb{P}_{\theta}\left(|Y_{n}|>\frac{\delta}{2}\right)\\
\leq & \,\,\epsilon
\end{align*}
for all $n$. This proves (\ref{Oup1}).

Finally, let $X_{n}$ be $\mathcal{O}_{p}^{u}(1)$ and let $Y_{n}$ be $o_{p}^{u}(1)$. Then, for any $\delta>0$, 
\begin{align*}
\sup_{\theta\in\Theta_{0}}\mathbb{P}_{\theta}\left(|X_{n}Y_{n}|>\delta\right)
\leq & \sup_{\theta\in\Theta_{0}}\mathbb{P}_{\theta}\left(|X_{n}|>k\text{ or }  |Y_{n}|>\frac{\delta}{k}   \right)\\
\leq & \sup_{\theta\in\Theta_{0}}\mathbb{P}_{\theta}\left(|X_{n}|>k\right)
+\sup_{\theta\in\Theta_{0}}\mathbb{P}_{\theta}\left(|Y_{n}|>\frac{\delta}{k}\right)
\end{align*}
for any $0<k<\infty$. For any $\epsilon>0$, each term on the right-hand side can be made less than $\epsilon/2$ by choosing $k$ and $n$ large enough. This proves (\ref{Oup4}).
\end{proof}

\end{document}